\PassOptionsToPackage{prologue,dvipsnames}{xcolor}

\documentclass{article}
\let\SUP\textsuperscript
\usepackage[T1]{fontenc}
\usepackage{graphicx} 
\usepackage{subcaption}
\usepackage{amsmath}
\usepackage{amssymb}
\usepackage{relsize}
\usepackage{tikz-qtree}
\usepackage{tikz}
\usetikzlibrary{trees,arrows}
\usepackage{verbatim}
\usepackage{braket}
\usepackage{xargs}
\usepackage{xcolor}
\usepackage[colorinlistoftodos,prependcaption,textsize=small]{todonotes}
\usepackage{soul}    % For \st command
\usepackage{xifthen} 

\usepackage{etoolbox} % For conditional commands
\usepackage[title]{appendix}
\usepackage{afterpage}
\usepackage{cite}
\usepackage[numbers, sort&compress]{natbib}
\usepackage{flafter}
\usepackage{float}
\usepackage{algorithm}
\usepackage{algorithmic}
\usepackage{physics}
\usepackage{floatflt}

\usetikzlibrary{calc, shapes, backgrounds}
\usetikzlibrary{shapes.geometric, arrows}

\usetikzlibrary{calc, shapes, backgrounds}
\usetikzlibrary{shapes.geometric, arrows}
\usetikzlibrary{automata, arrows.meta, positioning}
\usepackage{tabu}
\usepackage[hidelinks]{hyperref}
\usepackage[utf8]{inputenc}
\usepackage[acronym]{glossaries}
\usepackage{csquotes}
\usepackage{orcidlink}
\usepackage[export]{adjustbox}
\usepackage[percent]{overpic}

\newcolumntype{L}{>{\centering\arraybackslash}m{5cm}}
\newcolumntype{Z}{>{\centering\arraybackslash}m{1.6cm}}

\usepackage[textwidth=17cm,textheight=21.7cm]{geometry}

\usepackage[shortlabels]{enumitem}
\usepackage{diagbox}
\usepackage{graphicx}
\usepackage{subcaption}
\usepackage{booktabs}
\usepackage{placeins}
\usepackage{amsthm}
\theoremstyle{definition}
\newtheorem{definition}{Definition}[section]
\theoremstyle{lemma}
\newtheorem{proposition}{Proposition}
\theoremstyle{remark}

\usepackage{appendix}
\usepackage{wrapfig}

\title{Designing a Satellite Serviced Quantum Network Backbone for Concurrent Global Connectivity}
\date{September 2025}
\author{Prateek Mantri\orcidlink{0009-0005-3223-9731}\SUP{1}, Stav Haldar\SUP{2}, Albert Williams\SUP{3}, and Don Towsley\SUP{4}\\ \small
\{\SUP{1}pmantri, 
\SUP{2}stavhaldar,
\SUP{3}abwillia,
\SUP{4}towsley\}@cs.umass.edu}
\date{\textit{Robert and Donna Manning College of Information and Computer Sciences \\ University of Massachusetts, Amherst, USA, 01002}}

\begin{document}

\maketitle
\begin{abstract}
Satellite-serviced quantum networks pose an architectural problem distinct from classical satellite networking: because entanglement cannot be copied, and long-lived buffering is technologically constrained for near-term devices, useful end-to-end service requires fixed optical ground infrastructure and simultaneous multi-hop path availability. We investigate the design of a satellite-serviced quantum backbone aimed at supporting \emph{concurrent} global connectivity across a traffic matrix of major population and financial centers under finite waiting-time constraints. Using a discrete-time simulator that combines dynamic orbital motion, a Gaussian-beam optical link-efficiency model, and explicit per-satellite service policies, we evaluate performance using two architecture-level metrics: (i) time-to-connectivity, and (ii) latency-conditioned average active-link strength. Across a broad parameter sweep, we identify three dominant architectural effects. First, \emph{anisotropic} ground-station lattices reduce time-to-connectivity relative to longitudinally collapsed and isotropic baselines by aligning ground infrastructure with latitude-dependent satellite access. Second, \emph{multi-inclination} LEO constellations reduce waiting times for strong connectivity compared to single-inclination constellations at fixed satellite budgets by providing additional visibility for a diverse latitude set. Third, \emph{multi-party} satellite service policies alleviate per-satellite concurrency bottlenecks and substantially reduce time-to-connectivity at stringent traffic-matrix thresholds, including when normalized by total optical-terminal count. We further show that satellite altitude is the dominant physical lever shaping the visibility--loss trade-off, strongly affecting both connectivity latency and achievable link strength, while orbital plane count and satellite packing provide secondary refinements at fixed altitude. Together, these results delineate the architectural conditions required for scalable, concurrent entanglement connectivity in satellite-serviced quantum networks.
\end{abstract}
\section{Introduction}

Satellite-assisted quantum communication offers a promising route to overcoming the fundamental distance limitations of terrestrial optical fiber networks~\cite{pirandola_satellite_2021, Dolinar2011, yin2017satellite, Pan_exp_crypto, Pan_exp_teleport, Pan_satellite_QCS, goswami2025satellites}. Long-distance optical-fiber entanglement distribution requires repeater chains with multiplexing and long-lived quantum memories~\cite{muralidharan2014ultrafast, guha2015rate, muralidharan_optimal_2016, mantri20241-2-wayquantum}. Free-space optical links to satellites, by contrast, experience substantially lower loss over long distances, enabling global entanglement distribution with comparatively modest physical-layer requirements~\cite{Vasylyev2019, sidhu2021advances, haldar2024bell, pirandola_satellite_2021, goswami_satellite-relayed_2023, anipeddi2025optical, Lanning2022}. Satellite platforms therefore enable global entanglement distribution without dense terrestrial repeater chains~\cite{khatri_spooky_2021, yehia_connecting_2024, shao_hybrid_2025}.

Satellite-serviced quantum networks differ fundamentally from classical satellite communication systems. Classical Low-Earth orbit (LEO) satellite networks are shaped by throughput, coverage, handover efficiency, user access, latency variation, and transient congestion~\cite{chowdhury2006handover, li2020user, ccelikbilek2022survey, liu2024handover}. These challenges are significant, but classical systems can mitigate connectivity variation through routing, buffering, retransmission, transport-layer adaptation, beam steering, and store-and-forward operation~\cite{valentine2025leotcp, perrin2025measurement}. In quantum networks, however, the transported resource is entanglement: it cannot be copied~\cite{wootters1982single}, cannot be buffered indefinitely~\cite{sangouard_quantum_2009, Childress_Hanson_2013, sangouard_quantum_2011, schafer2018fast, semenenko2022entanglement, azuma2022quantum}, and must be established across simultaneous multi-hop paths within finite memory-coherence windows~\cite{li2021effective, semenenko2022entanglement, azuma2022quantum, pant_routing_2019}. As a result, the spatial relationship between ground-station placement, orbital visibility, and per-satellite servicing capability becomes a first-order architectural constraint rather than a secondary deployment detail~\cite{pirandola2021satellite, anipeddi2025optical}.

Despite this, existing satellite quantum-network proposals often inherit ground-station layouts and constellation structures from classical or link-centric studies~\cite{yehia_connecting_2024}. For example, Khatri \emph{et al.}~\cite{khatri_spooky_2021} consider a global grid with equal separation in latitude and longitude, while Shao \emph{et al.}~\cite{shao_hybrid_2025} use approximately uniform Euclidean spacing over continental regions. Other studies have looked into demand-centric grids or random locations~\cite{brito2021satellite, yehia_connecting_2024, anipeddi2025optical}. These choices are natural baselines, but they do not explicitly account for the latitude-dependent structure of LEO satellite visibility. In circumpolar constellations, satellite pass density increases with latitude, creating more orbital access near the poles than near the equator. A ground-station lattice that is uniform in geometric spacing or angular coordinates therefore does not necessarily match the spatial distribution of satellite service opportunities. This mismatch can leave some regions connectivity-limited while concentrating ground-station resources in regions where satellite access is already abundant.

The constellation architecture introduces a coupled design dimension. Prior work has studied regimes ranging from single high-altitude satellites serving regional networks~\cite{shao_hybrid_2025, yehia_connecting_2024} to modest single-shell LEO constellations with a few hundred satellites~\cite{khatri_spooky_2021, yehia_connecting_2024}. However, several network-level questions remain insufficiently explored: how ground-station geometry affects global connectivity latency; how orbital diversity mitigates visibility gaps; how altitude trades visibility against link quality; and how per-satellite servicing capability interacts with ground-station placement to determine achievable concurrent connectivity under finite waiting-time constraints.

These questions are becoming increasingly relevant as satellite optical-access technologies mature~\cite{cubesat:2020, oi_cubesat_2017, Nano_sat}. Recent progress in chip-scale entanglement sources~\cite{mahmudlu_fully_2023, simmons_scalable_2024, chapman_-chip_2025} and multi-terminal optical-access technologies~\cite{aguilar_multiple_nodate, goorjian_using_nodate} suggests that future satellites may be able to host multiple entanglement sources and serve several ground stations concurrently. Commercial LEO satellites already operate multiple simultaneous optical links~\cite{starlink_technology, AmazonLeo}, and compact terminals such as TESAT's SCOT20 support multi-terminal architectures even on CubeSat-scale platforms~\cite{tesat_scot20}. Such capabilities can convert excess geometric visibility into network concurrency, motivating a re-examination of ground-station placement, constellation structure, and satellite operation.

In this work, we study a global satellite-assisted quantum backbone in which LEO satellites distribute photonic Bell pairs to a lattice of optical ground stations connected to terrestrial fiber networks. We isolate infrastructure-level design principles from protocol-layer considerations~\cite{caleffi2022_1} and focus on a conservative operating regime: concurrent global connectivity under finite waiting-time constraints, with minimal reliance on long-lived quantum memories. This perspective emphasizes architectural bottlenecks that persist even under favorable physical-layer assumptions. We evaluate architectures using network-level metrics including time-to-connectivity, traffic-matrix connectivity, and latency-conditioned active-link strength.

Our study is organized around three design principles. First, we introduce an \emph{anisotropic ground-station lattice}, in which station spacing varies with latitude according to a tunable parameter~$\alpha$, compensating for both longitudinal compression and latitude-dependent satellite visibility. Second, we study \emph{multi-shell LEO constellations}, which distribute satellites across inclination shells to reduce synchronized coverage gaps and improve temporal continuity. Third, we consider \emph{multi-point satellite service}, where a single satellite can distribute bipartite entanglement links to multiple ground stations concurrently.

In summary, we make the following contributions:
\begin{itemize}[itemsep=0em]
    \item We introduce an anisotropic ground-station grid design that accounts explicitly for latitude-dependent satellite visibility and reduces global connectivity waiting times relative to isotropic and longitudinal baselines.
    \item We evaluate satellite constellation parameters, including orbital structure, satellites per plane, and altitude, with the objective of minimizing time-to-connectivity under finite waiting-time constraints. We show that coupling a near-polar shell ($98^\circ$) with a mid-inclination shell ($53^\circ$) mitigates visibility gaps and enables reliable high-threshold global connectivity at fixed satellite budgets.
    \item We show that multi-point satellite service substantially reduces time-to-connectivity by alleviating per-satellite concurrency limits, including when normalized by total terminal budget.
\end{itemize}

This manuscript is organized as follows: Section~\ref{sec:system} introduces the system design including the ground station grid architecture, satellite constellations, and on-board and on-ground technology requirements. Section~\ref{sec:Perf_Eval} introduces the strategy for evaluating performance of different architectures, including the metrics considered in this manuscript. Section~\ref{sec:Results} presents the results comparing the different designs, and outlines design insights from the evaluation. This is followed by Section~\ref{sec:Conclusion} outlining the conclusions of this study, and a discussion on future directions.

%%%%%%%%%%%%%%%%%%%%%%%%%%%%%%%%%%%
%%%%%% Sec: System Description %%%%
%%%%%%%%%%%%%%%%%%%%%%%%%%%%%%%%%%%
\section{System Description}
\label{sec:system}

We begin by introducing the key architectural parameters that govern satellite constellation geometry and temporal access patterns. Because these parameters recur throughout the system model and performance evaluation, we first summarize them at a conceptual level and illustrate their geometric meaning in Fig.~\ref{fig:constellation_cartoon}, before specifying the full system assumptions used in our simulations.

\subsection{Background: Satellite constellation}
\label{subsec:background_orbital_model}

Satellite-serviced networks are governed by a small set of orbital design parameters that determine where, when, and how satellites can interact with ground infrastructure. Since, satellite motion follows well-established orbital mechanics, many system-level performance characteristics can be traced to a limited number of architectural degrees of freedom. We first define these parameters precisely and then describe their operational implications. Fig.~\ref{fig:constellation_cartoon} provides a schematic illustration of the geometry introduced below.

\vspace{-0.5em}
\paragraph{Orbital altitude.}
The \emph{orbital altitude} $h$ is defined as the height of a satellite above the Earth's surface, measured relative to a fixed mean Earth radius $R_{\oplus}$, which is treated as constant throughout this work. Together with $R_{\oplus}$, the altitude determines the satellite's orbital period and its geometric distance to ground stations. Operationally, lower-altitude satellites experience reduced free-space path loss and shorter propagation delay, but cover a smaller region of the Earth's surface at any given time, requiring larger constellations to achieve continuous or near-continuous global access.

\vspace{-0.5em}
\paragraph{Orbital inclination.}
The \emph{orbital inclination} $i$ is the angle between the satellite's orbital plane and the Earth's equatorial plane. Inclination determines the maximum and minimum latitudes reached by the satellite ground track and therefore controls which latitude bands can be served by a constellation. For a given latitude, the orbital altitude $h$, and the orbital inclination determines the temporal frequency for a satellite overpass.

\vspace{-0.5em}
\paragraph{Elevation and zenith angles.}
At a ground station, the \emph{elevation angle} $\theta$ of a satellite is defined as the angle between the satellite line-of-sight vector and the local horizontal plane. The corresponding \emph{zenith angle} is defined as the angle between the line-of-sight vector and the local vertical and is given by $z = 90^\circ - \theta$. Throughout this work, satellite visibility is determined by imposing a minimum elevation constraint $|\theta| \geq \theta_{\min}$ (or equivalently $|z| \leq z_{\max} $). Note that throughout this work $z_{\max}$ is introduced as a fixed constraint and not an independent optimization parameter.

\vspace{-0.5em}
\paragraph{Satellite footprint.}
For a satellite at orbital altitude $h$ and a prescribed minimum elevation constraint $\theta \ge \theta_{\min}$, its \emph{footprint} is defined as the geographic region on the Earth's surface consisting of all points on ground from which the satellite is visible under this constraint at a given instant. The footprint size determines the number of ground stations that are geometrically visible at any given instant during a satellite pass and therefore sets the geometric upper bound on achievable spatial concurrency. Formally, the satellite footprint is the set
\begin{align*}
    \mathcal{F}(h, \theta_{\min}, t) &= \left\{ x \in \mathcal{S} \quad : \quad \theta(x; h,t) \ge \theta_{\min} \right\},
\end{align*}
\noindent where $\theta(x;h,t)$ denotes the elevation angle of the satellite relative to ground location $x$ at time $t$, defined only for ground locations with a direct line-of-sight (LOS) to the satellite, and $\mathcal{S} \subset \mathbb{S}^2$ denotes the set of ground locations under consideration. For fixed $\theta_{\min}$, the footprint diameter increases with altitude admitting more candidate ground stations per pass, but at the cost of longer optical path lengths and correspondingly weaker individual links due to increased free-space loss. 

\vspace{-0.5em}
\paragraph{Orbital planes and RAAN.}\label{subsec:planes_RAAN}
Satellites are organized into \emph{orbital planes}, each corresponding to a distinct great-circle trajectory around the Earth. The orientation of an orbital plane is specified by its \emph{right ascension of the ascending node} (RAAN, denoted by $\Omega$, with distinct planes indexed by $\{\Omega_p\})$, which defines the longitude at which the orbit crosses the equatorial plane in the northbound direction. The number of orbital planes and their relative RAAN spacing determine how coverage is distributed across longitudes and how uniformly visibility opportunities are spread as the Earth rotates beneath the orbital planes.

\vspace{-0.5em}
\paragraph{Satellite constellation and shells.}
A satellite constellation can be viewed as a collection of orbital planes, each characterized by an orbital altitude $h_p$, an inclination $i_p$, and a distinct RAAN $\Omega_p$ for each distinct member of the set $(h,i)$. We denote by $P$ the total number of orbital planes in the constellation, with planes typically spaced uniformly in RAAN within groups of similar orbital parameters to distribute coverage across longitudes. Formally, a constellation is defined as a set of $P$ orbital planes,
\begin{align*}
    \mathcal{P} = \{ \Pi_1, \ldots, \Pi_P \},
\end{align*}

\noindent where each plane $\Pi_p$ is specified by a triple $(h_p, i_p, \Omega_p)$. Planes with common $(h_p,i_p)$ values form a shell. Single-shell constellations correspond to the special case in which all planes share a common altitude $h$ and inclination $i$.

\vspace{-0.5em}
\paragraph{In-plane phase offsets.}
Within a single orbital plane, satellites are placed at fixed \emph{in-plane phase offsets} $\phi_{p,s}$, which specify their relative angular separation along the orbit. For a plane containing $S$ satellites, these offsets are typically chosen uniformly as,

\begin{align*}
    \phi_{p,s} = \frac{2\pi s}{S}, \quad s = 0,\ldots,S-1.
\end{align*}

\noindent These phase offsets determine the temporal spacing between successive satellite passes over a given region and therefore directly affect revisit times and short-term connectivity availability at individual ground stations.

\vspace{-0.5em}
\paragraph{Terminal concurrency constraints.}
Each satellite is equipped with a finite number of optical terminals or telescopes (denoted by $T$), limiting the number of ground stations that can be served concurrently during a pass. Similarly, ground stations must support sufficient receiver capacity to sustain parallel links. When more ground stations are visible within a footprint than can be simultaneously served, per-satellite concurrency rather than geometric visibility becomes the binding constraint.

\vspace{-0.5em}
\paragraph{Multi-shell constellations.}
As defined above, a \emph{shell} is defined as a group of orbital planes sharing a common altitude and inclination. We index shells by $\texttt{k}$ and denote by $\mathcal{P}^{(\texttt{k})} = \{ \Pi^{(\texttt{k})}_1, \ldots, \Pi^{(\texttt{k})}_{P_\texttt{k}} \} \subset \mathcal{P}$, the subset of orbital planes comprising shell $\texttt{k}$, where each plane $\Pi^{(\texttt{k})}_p$ is specified by a RAAN value $\Omega^{(\texttt{k})}_p$. A satellite constellation may comprise multiple shells with distinct orbital parameters $\{(h_\texttt{k},i_\texttt{k})\}_{\texttt{k}=1}^\texttt{K}$. Multi-shell constellations introduce diversity in coverage patterns, reducing global visibility gaps and improving robustness against long waiting times for large-scale connectivity.

In this work, we treat these parameters as high-level, tunable architectural variables and examine how their interaction with ground-station placement and satellite service policies governs end-to-end connectivity and latency in satellite-serviced quantum networks.

\begin{figure}[!htb]
    \centering
    \includegraphics[width=0.75\linewidth, trim = 5cm 1.2cm 5cm 2cm, clip]{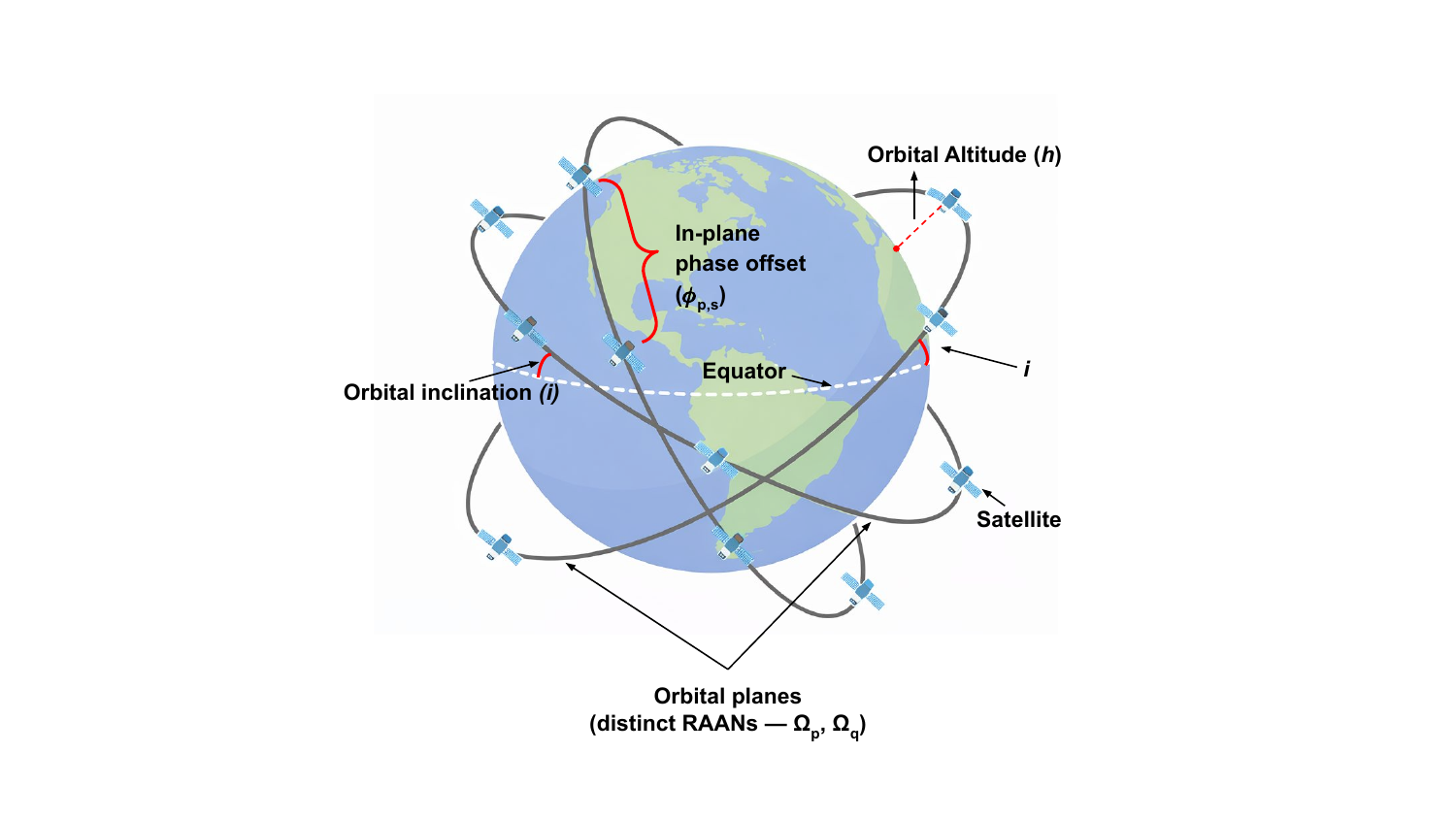}
    \caption{
    \textbf{Schematic illustration of the orbital parameters governing satellite-serviced network geometry.}
    Satellites are arranged in inclined orbital planes distinguished by their right ascension of the ascending node (RAAN), which sets the longitudinal orientation of each plane.
    Within a plane, satellites are separated by in-plane phase offsets that determine along-track spacing and revisit times.
    Orbital altitude controls footprint size and link loss, while orbital inclination determines the latitude range over which ground stations are served. Adapted from~\cite{zhang2024depth}.
    }
    \label{fig:constellation_cartoon}
\end{figure}

\subsection{System Introduction}
\label{subsec:system_intro}

In this work, we view the satellite--ground quantum backbone as an interaction of three design choices: (i) the ground-station layout, (ii) the orbital configuration that governs when and where stations are visible, and (iii) the per-satellite service capability that determines how many links can be supported within a footprint. This section specifies the system model and default assumptions used in the evaluation, with the aim of characterizing network-level concurrent connectivity under finite waiting-time constraints.

\vspace{-0.5em}
\paragraph{Ground segment.}
Ground-station spacing determines how many candidate stations fall within a typical satellite footprint. If stations are too dense, many compete for the same limited set of simultaneous downlinks; if stations are too sparse, footprints contain too few candidates and satellite service capacity is under-exploited. We therefore target a \emph{balanced} regime in which each footprint contains multiple candidates, but not so many that per-satellite concurrency dominates performance.

\vspace{-0.5em}
\paragraph{Space segment.}
Constellation geometry determines the temporal and geographic availability of satellite service. Altitude sets a footprint--loss trade-off: higher orbits increase contact duration and footprint size but reduce per-link efficiency, whereas lower orbits strengthen links but require denser constellations to avoid long outages. Orbital inclination controls latitudinal access. Finally, single-shell deployments can exhibit synchronized visibility gaps that delay the emergence of large connected components at high thresholds; we therefore also consider multi-shell designs that introduce orbital diversity across inclinations.

\vspace{-0.5em}
\paragraph{Hardware capabilities.}
Hardware determines how much of the available visibility can be converted into usable links. Multiple optical terminals per satellite increase the number of concurrent links supportable within a footprint, while larger ground apertures improve collection efficiency. These capabilities are constrained by on-board size, weight, and power (SWaP) limits and by ground-station deployment cost. Accordingly, we treat the optical terminal count per satellite and the satellite allocation across shells as explicit architectural parameters and evaluate their impact on connectivity and waiting-time metrics.

\vspace{-0.5em}
\paragraph{Network architecture.}
We model satellites as serving multiple ground stations concurrently via parallel downlinks. Ground stations connect to nearby metropolitan or regional fiber infrastructure, which enables local aggregation and delivery to end users. While our focus is the satellite-serviced backbone, the model is compatible with hybrid architectures that include higher-altitude layers (e.g., MEO/GEO) for coordination or buffering, and inter-orbit integration when required.

%%%%%%%%%%%%%%%%%%%%%%%%%%%%%%%%%%%%%%%%%%%%%%%%%%%%
%%% subsec: Ground Station Grid Design %%%%%%%%%%%%%
%%%%%%%%%%%%%%%%%%%%%%%%%%%%%%%%%%%%%%%%%%%%%%%%%%%%
\subsection{Ground Station Grid Design}
\label{subsec:GS_grid_design}

Ground-station placement directly shapes how effectively satellite visibility converts to network-wide connectivity. 
Prior work from the classical satellite networking setting motivates treating ground-station infrastructure as a shared, architectural resource rather than as mission-specific endpoints. In particular, recent work on \emph{Ground Stations as a Service} (GSaaS) explores models in which ground stations are owned and operated by multiple providers and are dynamically integrated into satellite missions to reduce deployment cost and improve utilization~\cite{eddy_optimal_2025}. Eddy \emph{et al.} study an existing global ground-station network comprising 91 stations operated by six independent providers, and formulate the problem of ground-station selection and assignment under resource and visibility constraints. While their focus is on operational optimization rather than geometric design, their model underscores the practical relevance of shared, multi-stakeholder ground infrastructure.

Prior proposals on satellite-assisted quantum networks have typically relied on either coordinate-aligned uniform grids~\cite{khatri_spooky_2021, shao_hybrid_2025}---for example, equal-angular (fixed-longitude) placements or uniformly spaced distance-based grids---or demand-driven layouts concentrated around major population centers~\cite{anipeddi2025optical, yehia_connecting_2024}. Each of these approaches exhibits structural limitations at global scale.

Equal-angular (longitudinal) grids suffer from longitudinal convergence: physical inter-station distances shrink with latitude, producing excessive ground-station density near the poles---where satellite visibility is already abundant---and comparatively sparse coverage near the equatorial belt. Uniform distance-based grids avoid this geometric compression but remain visibility-agnostic: they allocate stations evenly in physical space without accounting for the latitude-dependent variation in satellite pass frequency, leading to inefficient use of ground infrastructure in high-visibility regions. Demand-driven, population-based layouts further sacrifice geometric regularity, potentially leaving large geographic regions under-served while complicating routing symmetry and redundancy planning. A scalable backbone therefore requires a placement strategy that simultaneously preserves structural regularity, maintains balanced spatial coverage, and aligns station density with satellite visibility patterns. Such regular layout also supports augmentation through additional ground stations at high demand locations, however such analysis is beyond the scope of this work.

From a connectivity perspective, regular lattices offer structural advantages. Results from bond and site percolation theory show that triangular lattices exhibit lower percolation thresholds than square grids, enabling large connected components to emerge more reliably under intermittent link availability. While other regular lattice pattern such as square or hexagonal grids are also good candidates, their bond percolation thresholds are known to higher than the triangular lattice ~\cite{das_robust_2018}. This property is particularly relevant for satellite-based quantum networks, where connectivity fluctuates in time due to orbital motion and environmental effects. Triangular lattices also provide near-isotropic local connectivity with six equidistant neighbors. In light of the above, only triangular lattices have been considered in our analysis.

Na\"ively projecting such lattices onto the Earth, however, introduces geometric distortion. Fixed angular spacing in longitude leads to excessive crowding near the poles in the northern and southern latitudes---precisely where satellite pass are the densest. To correct for this mismatch, we introduce a latitude-dependent scaling of East--West spacing that counteracts longitudinal compression while preserving the regular structure of the lattice. We formalize these ideas in what follows in this section.

\begin{figure}[t]
\vspace{-1.3cm}
\begin{center}
    \includegraphics[width=1\linewidth, trim={1cm 2.8cm 0cm 3cm},clip]{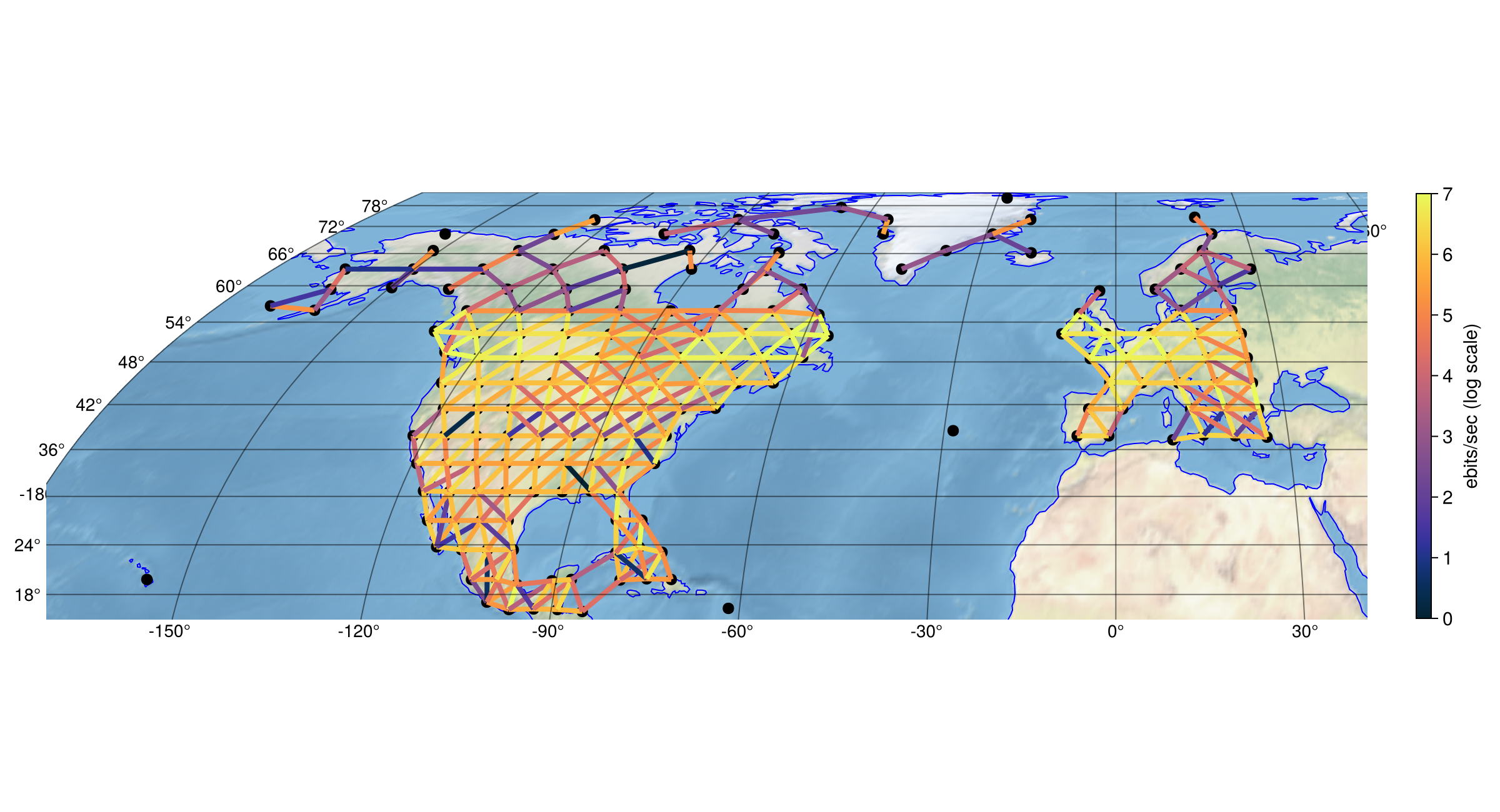}
    \vspace{-9mm}
    \caption{\small
        \textbf{Snapshot of bipartite connectivity (ebits delivered per unit time) for a ground-station grid with 176 stations across North America and Western Europe.} The constellation consists of satellites at an altitude of 500\,km, arranged in 360 orbital planes with 18 satellites per plane. Ninety percent of planes form a primary shell at inclination $53^\circ$, with the remainder allocated to a secondary polar shell. Ground stations are placed on an anisotropic grid with North--South spacing of $3.6^\circ$ and latitude exponent $\alpha=0.8$ in the East--West direction. If a lattice point falls on a large water-body it is snapped to the nearest land location if available within 100 km radius, otherwise the lattice point is dropped. Each satellite carries seven entanglement sources operating at 100\,MHz and supports up to seven concurrent bipartite links. Visibility is restricted to zenith angles $|z|\le 57^\circ$ to limit turbulence-related loss. Ground-station spacing increases toward higher latitudes, illustrating the effect of anisotropic placement.
    }
    \label{fig:triangle-grid}
\end{center}
\vspace{-0.5cm}
\end{figure}

%%%%%%%%%%%%%%%%%%%%%%%%%%%%%%%%%%%%%%%%%%%%%%%%%%%%
%%% subsubsec: ground stations lattice %%%%%%%%%%%%%
%%%%%%%%%%%%%%%%%%%%%%%%%%%%%%%%%%%%%%%%%%%%%%%%%%%%
\subsubsection{Ground-station lattice}
\label{subsubsec:GSlattice}

A first design choice for the ground segment is whether to impose a regular lattice structure. Compared to ad-hoc or demand-driven placement, a structured grid provides uniform geographic coverage and a reproducible baseline for analyzing network-level properties such as connectivity, percolation thresholds, and resilience under intermittent link availability. While demand-driven layouts may reflect present-day usage, a lattice-based backbone avoids embedding assumptions about current traffic patterns and remains robust to long-term shifts in demand.

Within this class of designs, lattice geometry directly affects connectivity and resilience. Square lattices are simple to construct but exhibit strong directional bias: nearest neighbors lie along cardinal directions, while diagonal separations are significantly larger. This particular anisotropy increases the percolation threshold and reduces robustness when links fluctuate in time, due to node/link failure, congestion, or attacks~\cite{das_robust_2018, brito2021satellite}.

Because near-term quantum memories have finite coherence times, and are expected to be an expensive resource, end-to-end connectivity in quantum networks cannot rely on the long-term accumulation of entangled links. Instead, useful connectivity must arise from the concurrent availability of links across neighboring nodes, which via entanglement swapping can produce reliable links between arbitrary points on the graph. This in turn implies that the rapid formation of large connected components critical in networks with time-varying link availability. We, therefore, adopt triangular lattices as the ground-station geometry of choice. This choice of lattice geometry concerns only the local neighbor structure of the grid and is conceptually distinct from the latitude-dependent anisotropy introduced next, which corrects global distortions arising from Earth's curvature.

%%%%%%%%%%%%%%%%%%%%%%%%%%%%%%%%%%%%%%%%%%%%%%%%%%%%
%%% subsubsec: alpha %%%%%%%%%%%%%%%%%%%%%%%%%%%%%%%
%%%%%%%%%%%%%%%%%%%%%%%%%%%%%%%%%%%%%%%%%%%%%%%%%%%%
\subsubsection{Anisotropy parameter \texorpdfstring{$\alpha$}{alpha}}
\label{subsubsec:alpha}
Uniform-angular ground-station grids place stations at equal longitude--latitude increments. When projected onto the Earth's surface, this construction leads to strong longitudinal convergence: physical inter-station distances shrink toward the poles and expand near the equator. The resulting polar crowding is particularly inefficient in satellite-serviced networks, where orbital pass density already increases with latitude, while equatorial regions---despite receiving fewer passes---are comparatively under-sampled.

We address this geometric mismatch using a tunable anisotropy parameter $\alpha$ that controls East--West station spacing as a function of latitude according to,
\begin{align*}
    d(\lambda) = \frac{d_{\mathrm{eq}}}{\cos^\alpha \lambda},
\end{align*}
where $d_{\mathrm{eq}}$ is the equatorial spacing between nearest neighbor ground stations. The exponent $\alpha$ determines how rapidly longitudinal spacing expands with latitude and therefore determines how station density is redistributed across the globe. 

Setting $\alpha=0$ recovers the uniform Euclidean (equi-distant) grid commonly used in prior work. Negative values (e.g., $\alpha=-1$) further concentrate stations at high latitudes, reproducing the longitudinal placement style that appears in earlier satellite-assisted network proposals. Positive values increase East--West spacing with increasing latitude and suppress polar crowding.

While any $\alpha \neq 0$ produces a latitude-dependent (anisotropic) grid, we distinguish between between two qualitatively different regimes. For $\alpha < 0$---which we refer to as \emph{longitudinal densification}---polar over-density is exacerbated as $\alpha$ becomes more negative. For $\alpha > 0$ (termed \emph{anisotropic suppression}) longitudinal compression is actively counteracted and polar station density is reduced with increasing values of $\alpha$. Throughout the remainder of the manuscript, we use the term \emph{anisotropic grid} specifically to refer to the $\alpha > 0$ regime, which corrects longitudinal compression and reduces polar over-density.

Figure~\ref{fig:triangle-grid} illustrates an anisotropic triangular lattice deployed across North America and Western Europe, with equatorial spacing $d_{\mathrm{eq}} = 400 ~\mathrm{km} (\equiv 3.6^\circ$ longitudinal spacing), and latitude exponent $\alpha=0.8$, chosen to balance satellite footprint size and polar crowding. If a lattice point falls on a large water-body it is snapped to the nearest land location if available within 100 km radius, otherwise the lattice point is dropped. Stations are denser at lower latitudes and progressively spaced farther apart towards the poles, reflecting both geometric considerations and satellite pass statistics. Despite this non-uniformity, the lattice retains strong global connectivity, consistent with the favorable percolation properties of triangular tilings. In Section~\ref{sec:Perf_Eval}, we show that such grids achieve substantially shorter connectivity wait times than uniform or longitudinal designs under identical satellite resources.

Figure~\ref{fig:grid_on_sphere} shows the difference between qualitatively different values of $\alpha$, viz., longitudinal densification, equidistant and anistropic suppression. Figure~\ref{fig:lattice_alpha} illustrates the effect of varying $\alpha$ continuously on the ground station density. Increasing $\alpha$ progressively sparsifies polar regions while preserving equatorial density. Number of land ground stations retained after snapping candidate lattice points to the nearest land location within a 100~km tolerance, plotted as a function of $\alpha$. Also as shown in Fig.~ \ref{fig:gs_num} moving from $\alpha=-1$ to $\alpha=0$ substantially reduces the number of ground stations, while increasing $\alpha>0$ further reduces this number. The shaded region highlights the anisotropic design regime used to reduce both the raw lattice density and the number of land-feasible stations needed to cover the globe, while avoiding excessive high-latitude oversampling. Operationally, $\alpha$ provides a single control parameter that trades off longitudinal compression, latitude-dependent satellite visibility, and heterogeneous global demand (typically demand is concentrated between $\approx \pm 60^\circ$ latitudes). Based on these considerations, we focus on $\alpha>0$, and in particular on the range $0.5 \le \alpha \le 1.5$ in our evaluations.

To avoid singular behavior near the poles where $\cos(\lambda)\to 0$, we impose a small numerical floor on the latitude-dependent density factor. Specifically, East--West inter-station spacing is capped below by a minimum physical separation of $50$\,km, preventing pathological oversampling at high latitudes. As a result, polar latitude bands contain at most one, or a small number of, longitude samples. This regularization is applied consistently in all simulations. 

\begin{figure}[!htb]
    \begin{subfigure}{\linewidth}
        \centering
        \caption{Schematic}
        \begin{overpic}[
            width=0.9\linewidth,
            trim=0cm 0cm 0cm 1cm,
            clip,
        ]{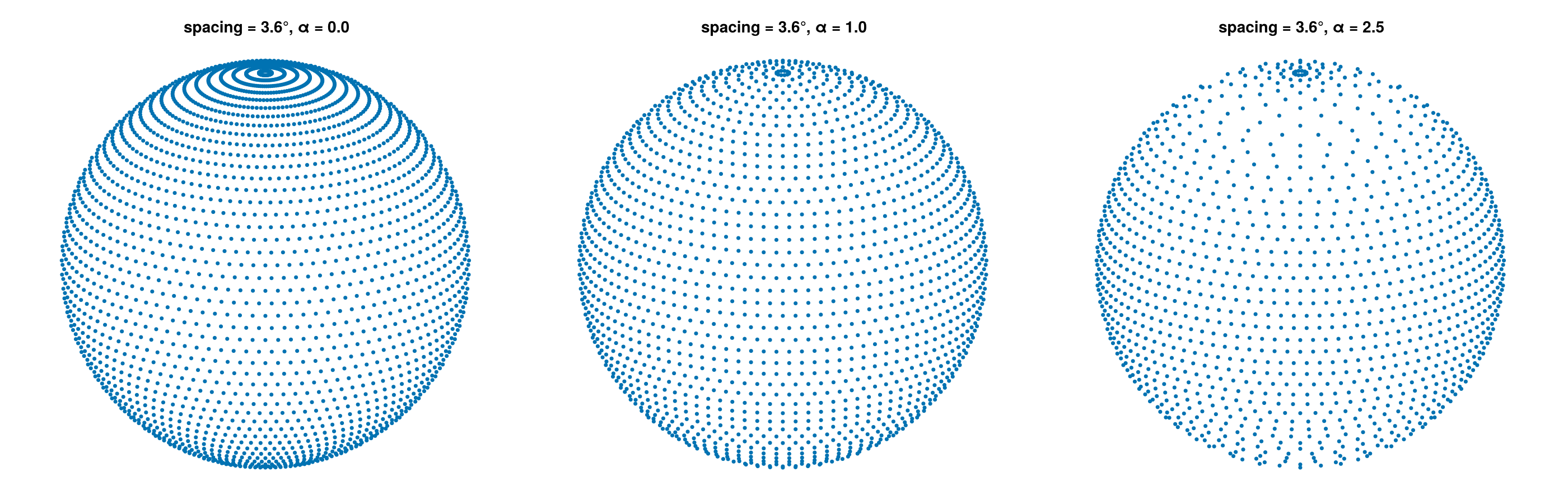}
            \put(12,13){\color{red}\Large $\alpha=-1$}
            \put(46,13){\color{blue}\Large $\alpha=0$}
            \put(78,13){\color{black}\Large $\alpha=1.5$}
        \end{overpic}

        \label{fig:grid_on_sphere}
    \end{subfigure}
\begin{subfigure}{0.48\linewidth}
    \centering    
        \caption{}
    \includegraphics[height = 4.5cm]{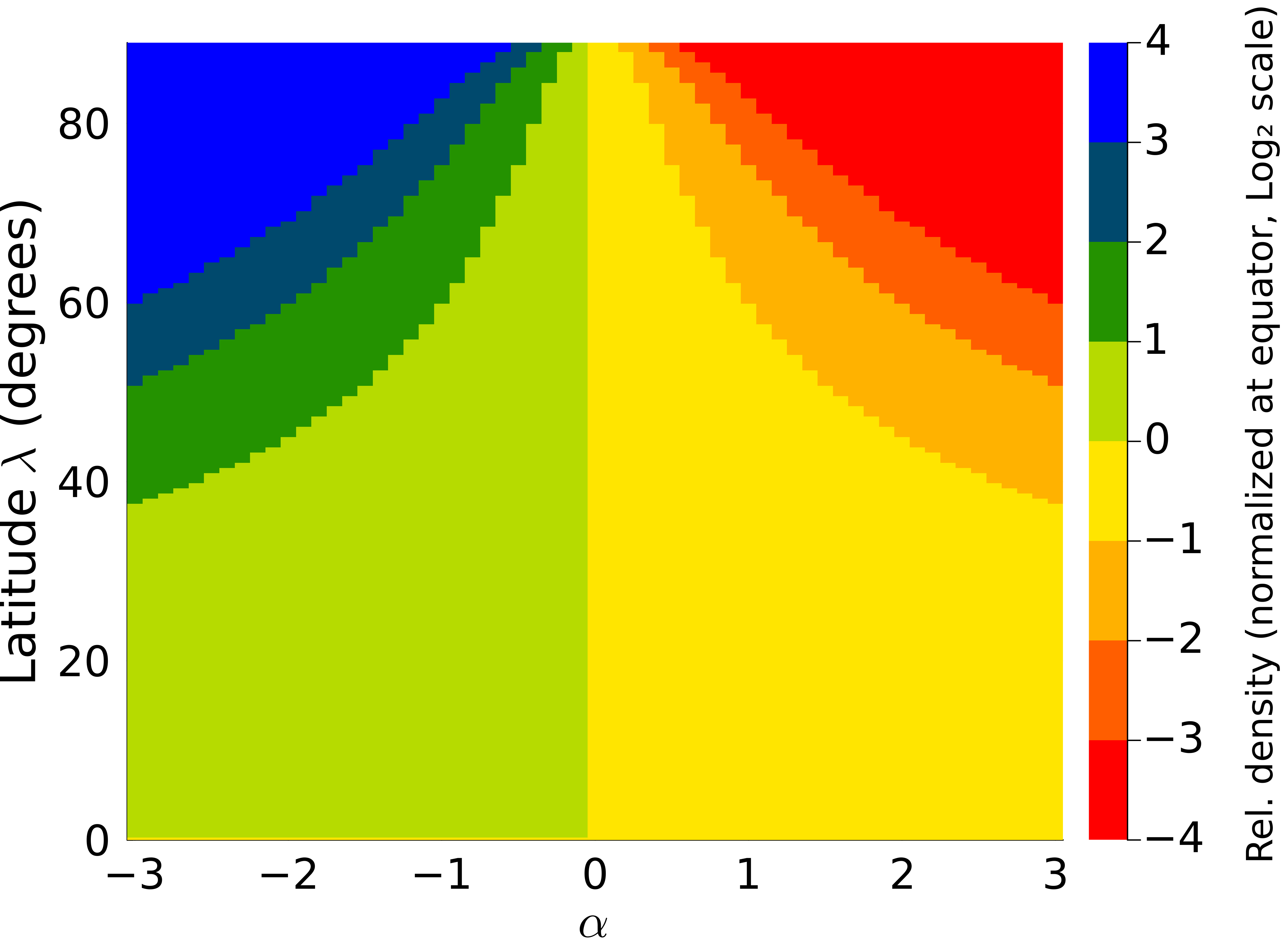}
    \label{fig:lattice_alpha}
\end{subfigure}
\begin{subfigure}{0.48\linewidth}
    \centering  
        \caption{}
    \includegraphics[height = 4.5cm]{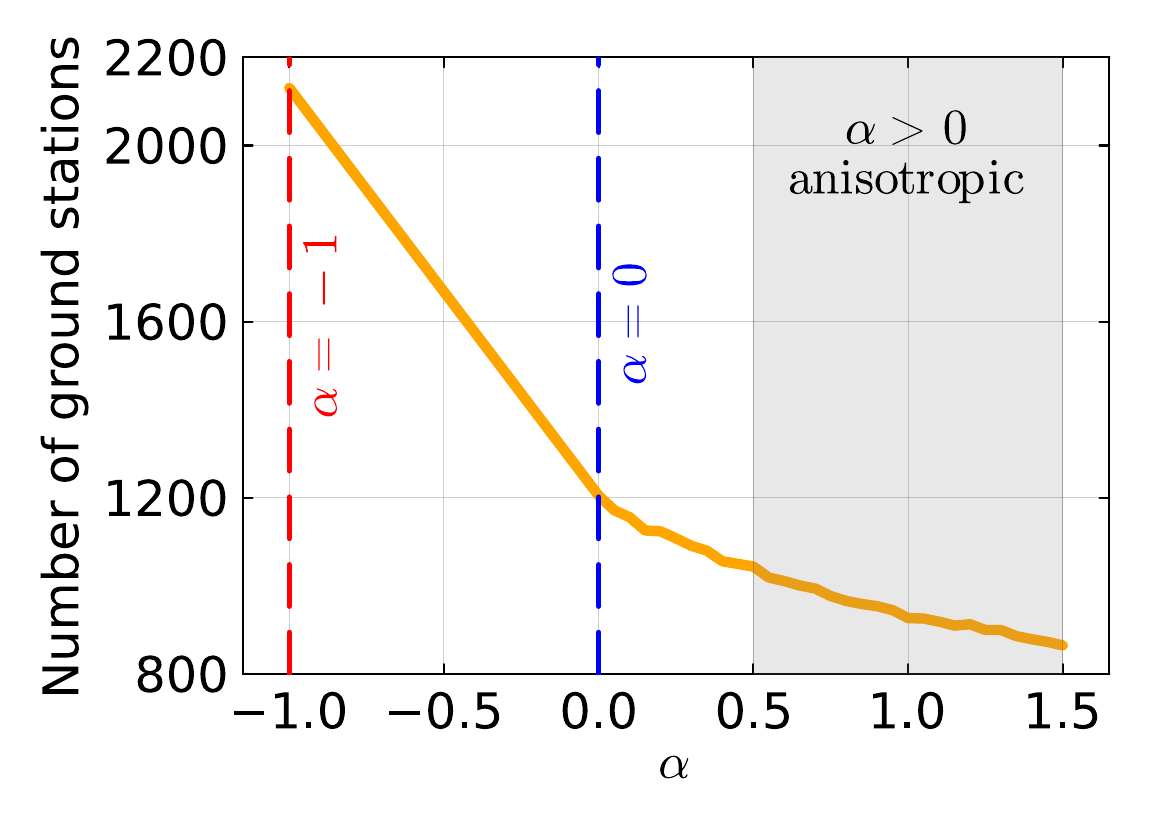}
    \label{fig:gs_num}
\end{subfigure}
\caption{
    \textbf{Ground-station density induced by the anisotropy parameter $\mathbf{\alpha}$.}
    (a) Schematic illustrations of candidate spherical lattice layouts for different values of the anisotropy parameter $\alpha$, showing how increasing $\alpha$ progressively suppresses excessive high-latitude longitudinal densification.
    (b) Relative longitudinal ground-station density as a function of latitude $\lambda$ and anisotropy parameter $\alpha$, with colors showing $\log_2 \rho(\lambda,\alpha)$ for $\rho(\lambda,\alpha)\propto \cos^\alpha\lambda$, normalized to the equatorial density.
    The case $\alpha=-1$ corresponds to longitudinal/equal-angular densification, $\alpha=0$ to an isotropic uniform-distance baseline, and $\alpha>0$ to anisotropic layouts that increase East--West spacing near the poles.
    (c) Number of land ground stations retained after snapping candidate lattice points to the nearest land location within a 100~km tolerance, plotted as a function of $\alpha$.
    Moving from $\alpha=-1$ to $\alpha=0$ substantially reduces the number of feasible land stations by removing polar overpopulation, while increasing $\alpha>0$ further sparsifies the grid.
    The shaded region highlights the anisotropic design regime used to reduce both the raw lattice density and the number of land-feasible stations needed to cover the globe, while avoiding excessive high-latitude oversampling.
}
\end{figure}

%%%%%%%%%%%%%%%%%%%%%%%%%%%%%%%%%%%%%%%%%%%%%%%%%%%%
%%% subsubsec: ground station distance at equator %%
%%%%%%%%%%%%%%%%%%%%%%%%%%%%%%%%%%%%%%%%%%%%%%%%%%%%

\subsubsection{Equatorial spacing \texorpdfstring{$d_{\mathrm{eq}}$}{deq}}
\label{subsubsec:equatorial_spacing}

While the anisotropy parameter $\alpha$ controls how ground-station density redistributes with latitude, the absolute scale of the grid is set by the equatorial spacing $d_{\mathrm{eq}}$. This parameter determines the total number of ground stations and directly affects both deployment cost and the level of concurrency achievable during satellite passes. Excessively small $d_{\mathrm{eq}}$ leads to redundant coverage within individual satellite footprints, while overly large spacing reduces the number of stations simultaneously visible to a satellite and limits parallel entanglement distribution.

We relate $d_{\mathrm{eq}}$ to satellite geometry through a unit-cell coverage condition. Fix a minimum elevation angle $\theta_{\min}$ (equivalently, a zenith cutoff $|z|\le z_{\max}=90^\circ-\theta_{\min}$), and let $D_{\mathrm{fp}}(h,\theta_{\min})$ denote the resulting satellite footprint diameter at altitude $h$. We define the minimum operational LEO altitude $h_{\min}$ as the smallest height satisfying
\begin{align*}
    D_{\mathrm{fp}}(h_{\min},\theta_{\min}) \;\ge\; 2 d_{\mathrm{eq}}.
\end{align*}

This condition ensures that, at equator, a single satellite footprint can fully cover one hexagonal unit cell of the ground-station lattice, consisting of a central hub station and its six nearest neighbors (See Figure~\ref{fig:hub-spoke-ring}).  At higher latitudes, anisotropy distorts the unit cell geometry but satellites in our operating regime typically have sufficiently large footprints (achieved by choosing a altitude higher than needed for a 2$d_{eq}$ footprint) to cover the hub-centered service neighborhood in a single pass. Because $D_{\mathrm{fp}}(h,\theta_{\min})$ increases monotonically with altitude for fixed $\theta_{\min}$, any satellite at altitude $h \ge h_{\min}$ automatically satisfies this concurrency condition. This equatorial calibration provides a conservative baseline for grid scaling.

In all our evaluations, we impose a minimum elevation angle of $\theta_{\min}=33^\circ$ (equivalently, a zenith cutoff $|z|\leq 57^\circ$). Under this constraint, a satellite at altitude $h\approx 285$\,km has a footprint diameter of approximately $800$\,km. Choosing $d_{\mathrm{eq}}=400$\,km therefore ensures that a single satellite footprint can cover one complete hexagonal unit cell. Satellites at higher altitudes only increase the footprint size and continue to satisfy this requirement. We therefore focus on LEO constellations with altitudes in the range $500$--$1400$\,km in the evaluations that follow.

Together, $\alpha$ and $d_{\mathrm{eq}}$ fully specify the ground-station lattice: $\alpha$ controls latitude-dependent redistribution, while $d_{\mathrm{eq}}$ sets the global scale. The next section examines how this ground geometry interacts with satellite constellation design to determine network-level connectivity and waiting times.

\subsubsection{Hub--spoke--ring connectivity model.}

\begin{wrapfigure}{r}{0.5\textwidth}
    \vspace{-1.2em}
    \centering
    \includegraphics[height = 5cm, trim = 8cm 5cm 8cm 2.5cm, clip]{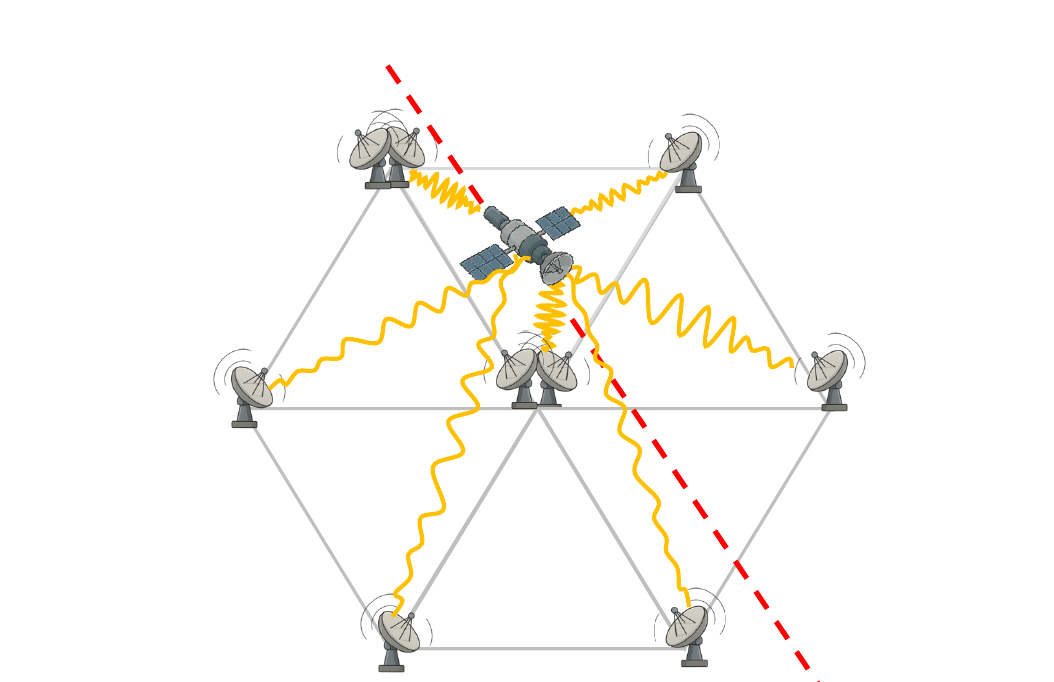}
    \caption{\textbf{Illustration of a hub–spoke–ring service model for a satellite-serviced quantum backbone.} During each satellite pass, the satellite designates a single ground station as a \emph{hub} (nearest visible station) and services a fixed local neighborhood consisting of the hub and its six nearest terrestrial neighbors. Entanglement distribution is restricted to \emph{hub–spoke} links and \emph{ring} links among the neighboring stations, rather than forming a fully connected clique. Ground stations are embedded in a regular terrestrial lattice (gray edges), which supports short-range routing and aggregation across successive satellite passes. As satellites move along their orbits (red dashed tracks), the identity of the hub and its associated neighborhood changes over time, inducing a time-varying but spatially local service pattern that preserves scalability while enabling concurrent global connectivity.}
    \label{fig:hub-spoke-ring}
        \vspace{-1.2em}
\end{wrapfigure}

At each epoch, satellite-induced connectivity among ground stations is constructed using a local hub--spoke--ring topology. For a given satellite, we first identify the set of ground stations that are simultaneously visible under zenith-angle constraints. From this visible set, the station with the highest instantaneous transmissivity (equivalently, the ground-station nearest to the satellite in its footprint) is designated as a local \emph{hub}. The satellite then selects a small neighborhood of geographically nearby ground stations around this hub, based on ground-station--to--ground-station distance, forming a local service region (see Fig.~\ref{fig:hub-spoke-ring}). This hub-based model is supported by prior research~\cite{brito2021satellite} indicating a small-worlds phenomenon in satellite--ground-station networks with the ground-station nearest to the satellite acting as a hub.

Connectivity is instantiated in two stages: (i) \emph{spokes}, connecting the hub to each selected neighbor, and (ii) a \emph{ring}, connecting neighboring stations in the cycle around the hub. Feasible links are weighted by the expected entanglement-generation rate induced by the satellite's instantaneous geometry. This construction yields a sparse, structured local subgraph that reflects the physical broadcast footprint of a satellite pass while supporting concurrent multi-party connectivity. In this regime, inducing richer local connectivity improves resource utilization and increases instantaneous node degree, accelerating the formation of large connected components that are necessary for backbone-style operation. Restricting a satellite with $T$ active terminals to servicing only a small number of disjoint ground-station pairs would significantly under utilize the available geometric configuration; for example, using seven optical terminals to support only three independent pairs discards the regular spatial structure of the footprint, increases swap depth for nearby nodes, and may introduce avoidable visibility or timing conflicts. At the opposite extreme, enabling all $\binom{T}{2}$ pairwise links is neither physically realistic nor protocol-efficient. The hub--spoke--ring topology strikes a balance between these extremes by producing a bounded-degree local subgraph in which any two stations in the service region are at most one hop apart, while remaining compatible with finite optical resources and near-term protocol constraints.

%%%%%%%%%%%%%%%%%%%%%%%%%%%%%%%%%%%%%%%%%%%%%%%%%%%%
%%% subsec: Satellite constellation design %%%%%%%%%
%%%%%%%%%%%%%%%%%%%%%%%%%%%%%%%%%%%%%%%%%%%%%%%%%%%%

\subsection{Satellite Constellation Design}
\label{subsec:Sat_constellation_design}

The second key component of the backbone is the satellite constellation. We focus on low-Earth-orbit (LEO) satellites, motivated by two practical considerations. LEO platforms support substantially stronger satellite--ground optical links than higher-altitude alternatives due to shorter propagation distances and reduced diffraction and atmospheric loss. These stronger links directly improve entanglement generation rates and reduce connectivity waiting times, which are central to network-level performance.

The feasibility of deploying large LEO constellations is well established. Contemporary classical communication systems routinely operate thousands of satellites in orbit, demonstrating that global coverage at scale is technically and economically viable~\cite{noauthor_global_2025, noauthor_global_nodate}. For quantum networking, however, constellation design must balance coverage against entanglement delivery rates, connectivity latency, and on-board hardware constraints. Individual LEO satellites operate under stringent size, weight, and power (SWaP) limits, which constrain the number of entanglement sources, optical terminals, and detectors that can be supported per platform. As a result, simply increasing satellite count does not guarantee improved network performance; constellation geometry must be chosen so that limited on-board resources are converted efficiently into usable, concurrent connectivity.

In this work, we do not attempt to optimize constellation parameters from first principles. Instead, we adopt a \emph{deployment-realistic} baseline inspired by widely deployed Walker-style LEO constellations, such as those used by \emph{Starlink}. Specifically, we use a primary shell at a moderate inclination of $53^\circ$, representative of constellations designed to provide high capacity over densely populated mid-latitude regions, together with a secondary near-polar shell at $98^\circ$ to ensure frequent access at high latitudes. This inclination pair mirrors the structure of existing mega-constellation deployments; for example, Starlink employs mid-inclination shells near $53^\circ$ alongside a near-polar shell at $97.6^\circ$ to balance mid-latitude capacity with robust high-latitude coverage. This choice is not claimed to be optimal; rather, it provides a concrete and interpretable reference configuration against which architectural questions---ground-station geometry, dual-shell deployment, and per-satellite service capability---can be evaluated.

Constellation choices interact strongly with ground-station placement. Increasing the number of orbital planes or introducing additional inclination shells generally reduces connectivity wait times by improving temporal overlap between satellite footprints and the ground-station lattice, but these gains saturate beyond moderate densities. In the following subsections, we examine how altitude, inclination, and shell structure shape this design space, and how these parameters jointly determine achievable global connectivity under finite waiting-time constraints.

\paragraph{Ground-station visibility under single-shell constellations.}
We generate the terrestrial lattice as a constellation-independent candidate backbone. Consequently, for a single-shell constellation with fixed inclination, the full lattice may include high-latitude ground stations that are never visible to any satellite. These stations would clearly be wasteful in an actual deployment and should be pruned after a constellation-specific visibility check. They do not, however, affect the performance conclusions reported here. Our traffic matrix (see Sec.~\ref{subsec:time_to_connectivity}) consists of city pairs that are visible in at least some configurations for all alphas over the sampled altitude range, with non-traffic-matrix ground stations used only as potential relay nodes when they are visible and serviceable. Permanently invisible stations therefore add no usable paths and do not inflate the reported connectivity or link-strength metrics. The comparison between single-shell and augmented dual-shell constellations is thus conservative with respect to this issue: pruning unreachable ground stations would reduce deployment cost but would not change the relative conclusion that augmented dual-shell designs provide stronger serviceable backbones.

\subsubsection{Satellite altitude}
\label{subsubsec:sat_height}

Satellite altitude is a primary design parameter because it jointly determines the spatial footprint of each pass and the efficiency of the satellite--ground quantum channel. At geostationary altitudes ($\sim36{,}000$\,km), continuous visibility over large geographic regions is possible, but the extreme propagation distance leads to prohibitive optical loss, rendering high-rate satellite--ground entanglement distribution impractical with current technology. Medium Earth orbits (MEO), spanning roughly $2{,}000$--$20{,}000$\,km, provide larger footprints than LEO but still incur substantial diffraction loss and long round-trip times, limiting their suitability for direct entanglement generation.

Low Earth orbit (LEO), typically at altitudes of a few hundred to a few thousand kilometers, offers a more favorable balance between link quality and coverage. Although individual satellites provide only intermittent visibility and cover smaller footprints, the reduced propagation distance enables significantly higher entanglement generation rates and shorter connectivity delays. For this reason, we adopt LEO satellites as the primary access layer of the quantum backbone.

Within the LEO regime, altitude introduces a clear and unavoidable trade-off. Increasing altitude enlarges the satellite footprint and increases contact duration, reducing the constellation density required for global coverage. At the same time, higher altitudes weaken the satellite--ground channel due to increased diffraction and atmospheric path length. Lower altitudes strengthen individual links and improve per-pass entanglement rates, but require denser constellations to avoid long temporal gaps in coverage. This trade-off affects not only single-link performance but also network-level connectivity, since footprint size determines how many ground stations can be served concurrently during a pass.

As discussed in Section~\ref{subsubsec:equatorial_spacing}, this footprint--altitude relationship also constrains the choice of ground-station spacing through the unit-cell coverage condition. In our baseline evaluations, we therefore focus on LEO altitudes between $500$ and $1{,}400$\,km. This range comfortably satisfies the unit-cell requirement for our ground-station geometry while spanning the regime in which altitude-driven trade-offs between footprint size, link efficiency, and connectivity latency are most pronounced.

\subsubsection{Shell design}
\label{subsubsec:inclination_angles}

Orbital inclination determines the latitude range over which a satellite provides service and therefore strongly shapes global connectivity. Low-inclination orbits concentrate access near the equator, whereas high-inclination and polar orbits traverse a broad range of latitudes, enabling frequent access to mid- and high-latitude ground stations. Within a given shell, multiple orbital planes at a common inclination are distinguished by their right ascensions of the ascending node (RAAN). RAAN spacing primarily affects longitudinal uniformity and temporal smoothing, but does not substantially alter latitudinal coverage. We therefore treat inclination as the dominant driver of latitude-dependent access pattern.

Inclination choice interacts closely with the anisotropic ground-station lattice introduced in Section~\ref{subsubsec:GSlattice}. Higher-inclination shells predominantly serve high latitude regions, where station density is intentionally reduced to compensate for increased satellite pass frequency, while moderate-inclination shells serve denser equatorial and mid-latitude regions. Combining shells at different inclinations therefore better matches satellite access patterns to ground-station density reducing connectivity wait times and improving robustness against temporally correlated coverage gaps.

A single-shell constellation can provide broad geographic coverage when sufficiently dense, but tends to concentrate visibility within specific latitude bands, and can exhibit correlated temporal gaps, particularly when serving both equatorial and high-latitude ground stations. Introducing a secondary shell mitigates both effects: a moderate inclination serves equatorial and mid-latitude regions, while a polar shell provides more frequent high-latitude access reducing coverage gaps. This augmented dual-shell structure aligns naturally with the anisotropic ground-station lattice.

\paragraph{Choice of inclination shells (53$^\circ$ and 98$^\circ$).} Rather than optimizing inclination angles from first principles, we adopt a deployment-realistic baseline drawn from contemporary large-scale LEO constellations~\cite{starlink_technology}. A primary shell at inclination $i=53^\circ$ provides dense and frequent coverage over populated mid-latitude regions, characteristic of Walker-style constellation designs. A secondary near-polar shell at $i_{\mathrm{polar}}=98^\circ$ representative of sun-synchronous orbits, ensures frequent access to high-latitude regions that would otherwise experience intermittent visibility. A fixed fraction $f_\mathrm{polar}$ of satellites is allocated to the polar shell, with the remainder assigned to the mid-inclination shell. We evaluate how this two-shell structure modifies global connectivity latency and concurrency relative to a single-shell baseline. Since we only take a small fraction of satellites to be reserved for the polar shell ($f_\mathrm{polar} \leq 20\%$), we call this design an ``augmented dual-shell'' approach.

Therefore, we restrict attention to this augmented dual-shell (ADS) architecture throughout, isolating the impact of inclination diversity without introducing additional degrees of freedom. As shown in Section~\ref{sec:Results}, this minimal two-shell structure is sufficient to reduce connectivity wait times at high thresholds relative to single-shell baselines.

\subsection{Hardware requirements}\label{subsec:Sat_GS_hardware}
\subsubsection{Transmitter characteristics (on-board technology)}
\label{subsubsec:Tx_hardware}

Satellite payload design is fundamentally constrained by size, weight, and power (SWaP), which jointly limit the number of optical terminals, entanglement sources, and independent pointing assemblies that can be carried on a single platform. These constraints directly couple per-satellite servicing capability to constellation density: increasing on-board resources improves per-pass concurrency and service capacity, but at the cost of higher payload complexity and power consumption, which may reduce the number of satellites deployable within a fixed launch budget. As a result, satellite-serviced quantum backbones must balance per-satellite capability against constellation scale.

At the optical access layer, link performance is governed by a trade-off between diffraction loss and pointing complexity. The effective beam waist (see Fig. \ref{fig:beam_waist}) at the satellite determines the satellite--ground link transmissivity: smaller beam waists improve collection efficiency at the ground but require  tighter pointing, acquisition, and tracking tolerances, while larger beam waists relax alignment requirements at the expense of increased diffraction loss. For LEO distances, practical designs operate in an intermediate regime that balances aperture size, pointing stability, and SWaP constraints to achieve robust link performance under realistic atmospheric and platform dynamics. Note that we do not include pointing errors in our analysis, instead our choice of $z_{max}$ is motivated by appropriating part of the orbit for pointing and tracking.

\begin{wrapfigure}{R}{0.5\textwidth}
    \centering
    \includegraphics[width=0.9\linewidth, trim = 7cm 4cm 4cm 3cm, clip]{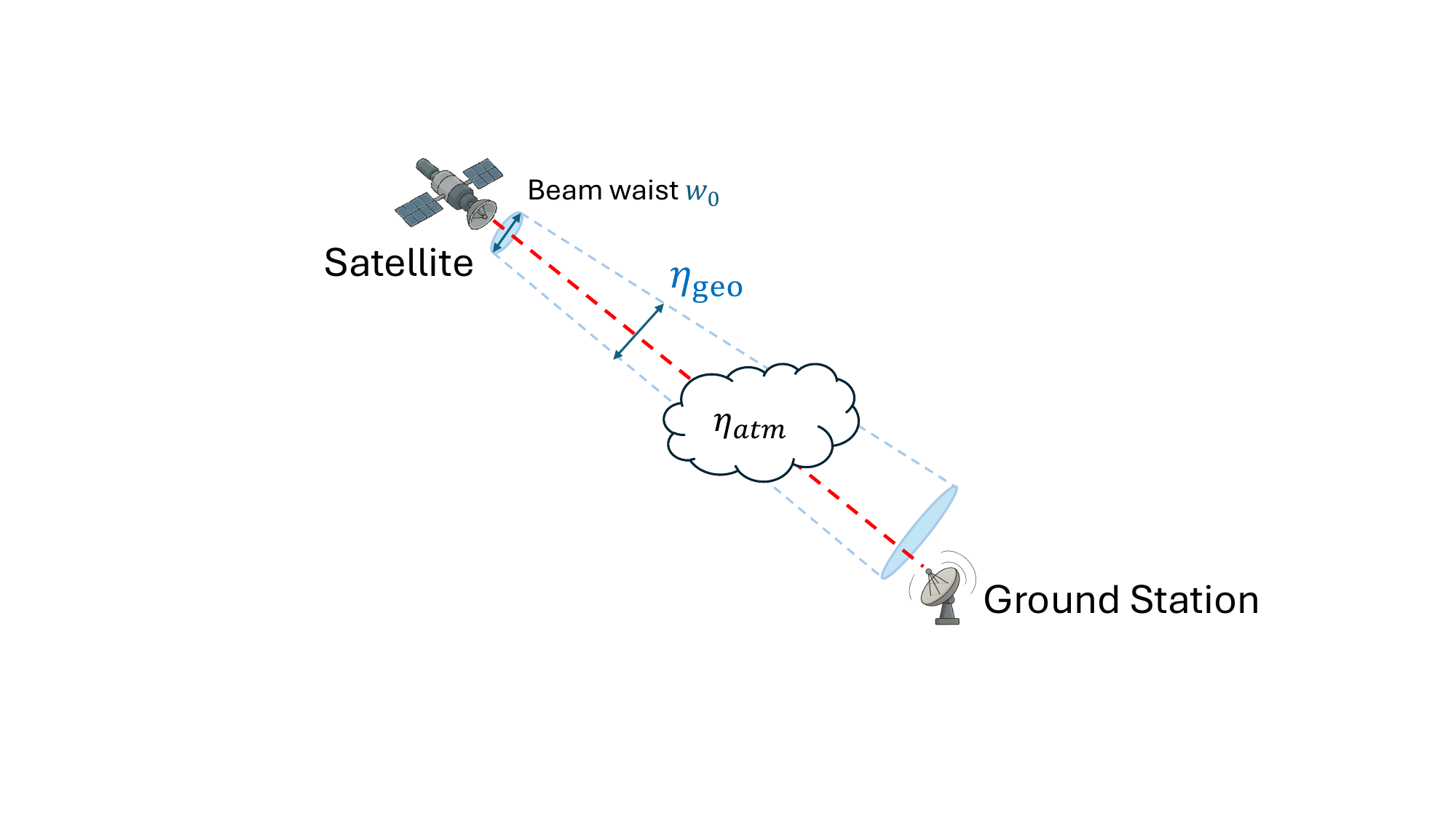}
    \caption{\textbf{Geometry of the satellite-to-ground optical channel.} A diffraction-limited Gaussian beam of waist $w_0$ is emitted from the satellite, experiences atmospheric loss given by $\eta_{\mathrm{atm}}$, and free-space diffraction loss given by $\eta_{\mathrm{geo}}$, and is partially collected by the ground-station receiver aperture. Pointing error arising from satellite jitter and tracking inaccuracy is represented by a lateral displacement of the beam relative to the receiver.}
    \label{fig:beam_waist}
\end{wrapfigure}

From a network perspective, the dominant on-board architectural parameter is the number of independent optical terminals and entanglement sources that can operate  concurrently. Satellites capable of supporting multiple simultaneous downlinks can service several ground-station pairs within a single pass, increasing instantaneous network connectivity and substantially reducing the waiting time for large connected components to form. In contrast, satellites limited to a single active double-downlink per pass must rely on denser constellations to achieve comparable network-level performance, shifting the burden from on-board capability to orbital design.

Recent progress in compact optical terminal technology has made multi-terminal satellite architectures increasingly practical. Operational systems in the classical domain demonstrate the feasibility of concurrent optical links within LEO SWaP constraints: for example, second-generation commercial LEO satellites operate multiple optical terminals simultaneously for inter-satellite communication, establishing multi-terminal optical operation at scale~\cite{starlink_technology}. At smaller satellite scales, compact terminals such as TESAT's SCOT20---with mass below 1.6~kg, volume of approximately 1U, and power consumption below 35~W---enable multiple optical downlinks even on CubeSat-class platforms with limited payload allocations~\cite{cubesat:2020,tesat_scot20}. Beyond deployed systems, systems-level studies have explicitly evaluated LEO platforms supporting up to four simultaneous optical links, identifying power allocation, independent pointing and tracking, and terminal integration as the dominant constraints rather than diffraction or channel loss physics~\cite{aguilar_multiple_nodate}. Recent research has further explored multi-beam and multi-terminal laser architectures for satellites, including independent beam steering and access scheduling for concurrent optical links~\cite{spaander_free-space_2025}.

For quantum applications, optical access-layer capability must be complemented by entanglement sources with sufficiently high pair generation rates to sustain multiple concurrent downlinks. Recent advances in chip-scale integrated entanglement sources, including MHz-class pair-rate demonstrations, indicate a path toward compact sources compatible with multi-terminal operation~\cite{mahmudlu_fully_2023, simmons_scalable_2024}. While satellite-based demonstrations of multi-party entanglement distribution remain an active area of research, these developments suggest that source brightness and integration, rather than fundamental physics, are likely to set near-term limits on concurrency. 

Importantly, the constraints discussed above are predominantly engineering challenges rather than fundamental limitations. Power budgeting across terminals, independent pointing and tracking, thermal management, and payload integration complexity define the practical envelope for multi-terminal operation. Accordingly, we treat the maximum number of concurrently served ground stations per satellite as a tunable \emph{architectural} parameter rather than a fixed hardware constraint.

In the model considered in this manuscript, each satellite is equipped with $T$ optical telescopes (``terminals'') serviced through $2(T-1)$ independent bipartite entanglement-generation sources, and distribute entangled photon pairs to ground stations according to a hub--spoke--ring service policy. Within each time shard of duration $\delta t = 1\,\mathrm{s}$, these sources are scheduled in a time-multiplexed manner across the corresponding hub--spoke and ring edges. In this context, we define \emph{time shards} as fixed-duration discrete time windows of length $\delta t$, over which satellite positions, visibility relations, and the induced network topology are assumed static, enabling snapshot-based evaluation of link generation and connectivity in an otherwise continuously evolving orbital system. The range considered ($ 2 \leq T \leq 7$) spans conservative extensions of current compact-terminal capabilities through plausible next-generation multi-beam payload designs, and is intended as an upper-bound modeling envelope rather than a claim about presently deployed quantum satellite hardware. These assumptions directly define the bipartite and multi-party service models evaluated in Section~\ref{sec:Perf_Eval}.

\subsubsection{Receiver characteristics (ground technology)}
\label{subsubsec:Rx_hardware}

At ground stations, receiver hardware determines both photon collection efficiency and the number of simultaneous satellite links that can be sustained. The primary optical parameter is the telescope aperture diameter, which sets collection efficiency under diffraction and atmospheric turbulence. Larger apertures improve link margins and robustness to loss, but increase site cost, mechanical complexity, and infrastructure requirements. For large-scale deployments, aperture diameters in the order of decimeters range provide a practical compromise between performance and deployability.

At the network level, the number of independent receivers available at each ground station directly limits concurrency. A single receiver supports only one active satellite link at a time, whereas multiple receivers allow a station to maintain parallel links to different satellites during overlapping visibility windows. This increases instantaneous connectivity and reduces idle time when multiple satellites are simultaneously in view. As with on-board terminals, additional receivers improve network performance but increase deployment and operational cost.

In our evaluations, we assume that ground stations are provisioned with sufficient receiver capacity to exploit the satellite visibility enabled by the lattice geometry. Receiver count is therefore not treated as a primary bottleneck, allowing us to isolate the impact of constellation geometry and on-board multiplexing on network-level connectivity.

\subsubsection{Size, weight, and power (SWaP) considerations}
\label{subsubsec:SWAP_req}

All satellite and ground-station hardware choices are constrained by size, weight, and power (SWaP). On board, these limits bound the number of optical terminals, entanglement sources, and pointing assemblies that can be supported on a given platform. Increasing per-satellite capability improves per-pass concurrency and throughput, but increases payload mass, power consumption, and system complexity, potentially reducing the number of satellites that can be deployed.

These constraints induce an architectural trade-off between constellation density and per-satellite capability. Rather than optimizing both simultaneously under an explicit cost model, we adopt a controlled evaluation strategy by either fixing the satellite budget or the total number of optical terminals. This isolates the effect of on-board spatial multiplexing on network-level metrics such as connectivity wait time and average active-link strength, without conflating these effects with changes in satellite count or launch cost.

\subsubsection{Ground-station optical hardware for quantum operations}
\label{subsubsec:quantum_hardware}

In addition to photon collection, ground stations must perform quantum operations to extend and utilize entanglement across the network. The central operation is the Bell-state measurement (BSM), which enables entanglement swapping by interfering and measuring two photons, typically originating from distinct satellite--ground links. In large-scale deployments, BSMs can be implemented using linear-optical interferometers and single-photon detectors. Although linear-optical BSMs have non-unit success probability, their technological maturity and scalability make them well suited for backbone operation.

Ground stations may also be equipped with quantum memories to buffer entangled states while awaiting suitable partners for swapping. In this work, we assume finite-lifetime memories capable of storing states over short timescales. Such buffering relaxes synchronization constraints and improves network performance under variable link availability. We explicitly evaluate scenarios with extended memory lifetimes and observe corresponding gains in throughput and reduced connectivity delays.

Crucially, long-coherence-time quantum memories are not required for the architectural regimes emphasized here. Our analysis focuses on connectivity wait times on the order of one second or less, consistent with near-term, NISQ-era capabilities. In this regime, high-rate satellite--ground links combined with spatial multiplexing, linear-optical BSMs, and short-lived buffering can support global-scale concurrent connectivity.

Our baseline model assumes linear-optical BSMs at ground stations with finite buffering capability, while treating longer-lived quantum memories as an optional enhancement rather than a prerequisite. This separation allows architectural effects to be evaluated independently of specific physical-layer protocol choices.

\subsubsection{Optical fiber backbone for resilience}
\label{subsubsec:Optical_fiber_backbone}

While satellites provide global-scale entanglement distribution, terrestrial optical fiber networks serve as the local interface between ground stations and end users. In our model, each ground station is assumed to be connected via existing metropolitan or regional fiber infrastructure to nearby population centers. This allows entanglement distributed to the ground-station lattice to be delivered efficiently to the cities and regions represented in the traffic matrix.

The role of fiber in this work is not to compete with satellite links, but to provide local connectivity and resilience. Fiber connections mitigate the intrinsic intermittency of satellite access and ensure that temporary visibility gaps or localized outages at individual ground stations do not directly translate into loss of end-to-end connectivity at the application level. Because the distances involved are metropolitan or regional in scale, they are well within the capabilities of current low-loss fiber technology and do not constitute a limiting factor in our analysis.

We do not explicitly model the terrestrial fiber network or its internal routing. Instead, fiber connectivity is treated as a reliable local transport layer that couples the satellite-served ground-station lattice to population centers, allowing us to focus on the global satellite--ground backbone and its impact on network-level connectivity and waiting-time performance.

\subsection{Summary}
The architecture considered in this work combines a latitude-aware ground-station lattice, a multi-shell LEO satellite constellation, and modest on-board and ground-based quantum hardware to form a scalable global quantum networking backbone. Ground stations are placed on an anisotropic triangular lattice parameterized by an equatorial spacing $d_{\mathrm{eq}}$ and a latitude-scaling exponent $\alpha$, ensuring efficient coverage while avoiding excessive station density at high latitudes. Satellite footprints are chosen to cover at least one hexagonal unit cell of this lattice, enabling multi-party connectivity through a hub--spoke--ring service model during each satellite pass.

The satellite constellation is built around LEO platforms to exploit high-rate, low-loss satellite--ground links, with inclination diversity and multi-shell structure used to balance coverage between equatorial and high-latitude regions and to reduce connectivity wait times. On-board spatial multiplexing allows individual satellites to serve multiple ground stations concurrently, while ground stations perform linear-optical Bell-state measurements to extend entanglement across the network.

We assume finite-lifetime quantum memories sufficient to buffer entangled states over short timescales, consistent with near-term, NISQ-era capabilities. Connectivity is evaluated under strict waiting-time constraints (on the order of one second), demonstrating that global-scale concurrent connectivity can be achieved without relying on long-coherence-time memories. Longer-lived memories are treated as an optional enhancement and are shown to further improve performance when available. Together, these assumptions define an architecture compatible with near-term hardware while remaining extensible as quantum technologies mature.
%%%%%%%%%%%%%%%%%%%%%%%%%%%%%%%%%%%%%%%
%%%%%% Sec: Performance Evaluation %%%%
%%%%%%%%%%%%%%%%%%%%%%%%%%%%%%%%%%%%%%%

\section{Performance Evaluation}\label{sec:Perf_Eval}

We evaluate whether satellite-serviced ground-station architectures can operate as concurrent terrestrial quantum-network backbones, rather than as intermittent collections of satellite--ground optical links. We compare three architectural choices: ground-station geometry, constellation structure, and per-satellite concurrent service capability. Section~\ref{subsec:metrics} defines the performance metrics, Sec.~\ref{subsec:evaluation} describes the design questions, and Sec.~\ref{subsec:simulation_methodology} details the simulation methodology.

\begin{wrapfigure}{r}{0.5\linewidth}
\vspace{-2em}
        \centering
    \includegraphics[height = 4.5cm, trim = 0cm 0.2cm 0cm 2cm, clip]{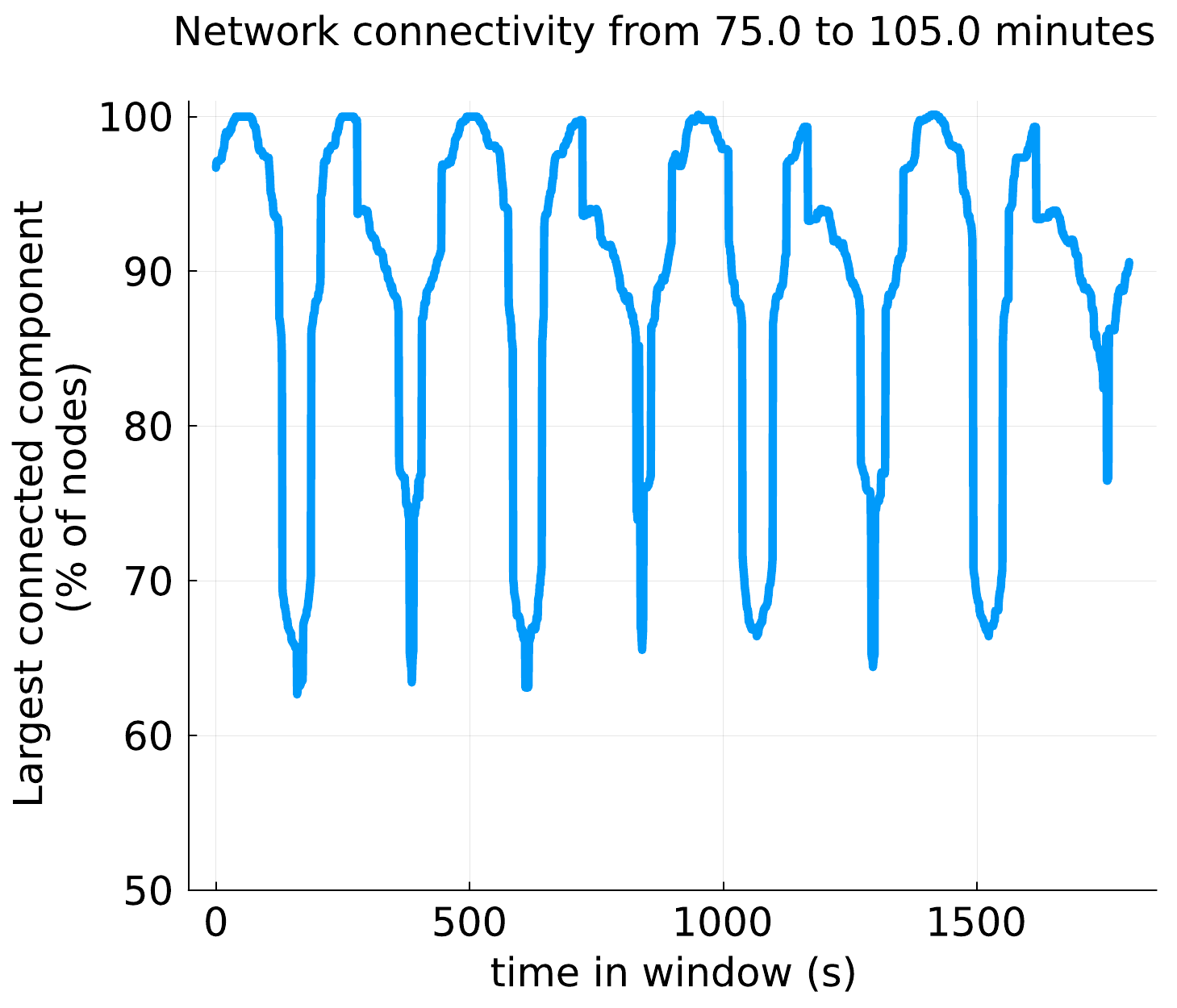}
    \caption{
    \textbf{Illustrative time series of the instantaneous largest connected component (LCC) fraction over a representative thirty-minute window of the satellite-serviced ground-station graph}(see Sec.~\ref{subsec:simulation_methodology} and Appendix~\ref{app:metrics_connectivity} for details).
    The trace shows quasi-periodic reductions in concurrent connectivity caused by time-varying satellite--ground visibility. These low-connectivity intervals motivate the forward-wait metric used in this work: rather than averaging connectivity over time, we measure how long an application must wait until the backbone reaches a specified concurrent-connectivity threshold.
    This figure is intended only as a diagnostic visualization of the temporal connectivity process; quantitative architectural comparisons are made using traffic-matrix waiting-time and latency-conditioned active-link-strength metrics.
    }
    \label{fig:connectivity_troughs}
\end{wrapfigure}

\subsection{Metrics}\label{subsec:metrics}

We use two complementary metrics: time to concurrent connectivity and latency-conditioned average active-link strength. The first measures how long the backbone must wait before a target fraction of traffic-matrix nodes are simultaneously connected. The second measures the average strength of active elementary links for scenarios that satisfy a prescribed waiting-time threshold.

\subsubsection{Time to Concurrent Connectivity}\label{subsec:time_to_connectivity}

Our primary performance metric is the time required for the satellite-serviced backbone to achieve a specified level of concurrent connectivity. At each instant $t$, we identify the largest subset of ground-station pairs for which entanglement can be established end-to-end through the backbone, either directly or via entanglement swapping through intermediate nodes. This defines an instantaneous concurrent connectivity process, representing the fraction (or number) of traffic-matrix nodes that can be simultaneously connected at time $t$.

This notion distinguishes backbone-like operation---where a persistent substrate supports simultaneous end-to-end paths between distant nodes---from regimes dominated by intermittent or sequential links. Simultaneity is evaluated at a simulation resolution of $\delta t = 1~\mathrm{s}$, so connectivity at time $t$ should be interpreted as requiring entanglement resources to persist long enough to support routing and swapping over that epoch, corresponding to a memory-and-control timescale on the order of $\delta t$.

Rather than summarize this connectivity process using temporal averages of the connectivity, we characterize performance using \emph{forward-wait} statistics (Appendix~\ref{app:forward_wait_time}), which measure the expected waiting time until concurrent connectivity exceeds a specified threshold over the simulation horizon. Operationally, this corresponds to the expected time required for all, or a predefined fraction, of traffic-matrix nodes to become connected through viable end-to-end paths. This first-passage perspective captures phase-dependent visibility gaps over the orbital cycle that are obscured by average connectivity measures alone.

Short forward waits indicate favorable satellite visibility and well-distributed coverage, whereas long waits reflect sparse coverage, recurring low-connectivity intervals, or insufficient satellite concurrency that delay backbone availability. Throughout the paper, unless otherwise stated, we report results under latency constraints.

Our evaluation focuses on connectivity over a predefined traffic matrix representing application-relevant demand. Specifically, we construct a traffic matrix consisting of 120 cities: the top 100 ranked global financial centers~\cite{noauthor_global_2025}, supplemented with 20 additional cities selected to improve geographic and socioeconomic coverage (Figure~\ref{fig:traffic_matrix}). Each city is mapped to its nearest ground station on the underlying grid, and connectivity is evaluated between the corresponding ground-station pairs. This construction evaluates backbone performance as experienced by realistic endpoints while remaining agnostic to the precise placement of individual ground stations. Connectivity over the full ground-station grid is evaluated separately in Supplementary Materials~\cite{supplementary_materials} for additional structural context.

Throughout, $\vartheta$ denotes a threshold applied to the fraction of city-pairs in the traffic matrix connected at an instant (see Sec.~\ref{subsec:simulation_methodology} and Appendix~\ref{app:metrics_connectivity} for details).

\begin{figure}[!htbp]
    \centering
    \includegraphics[width=0.95\linewidth]{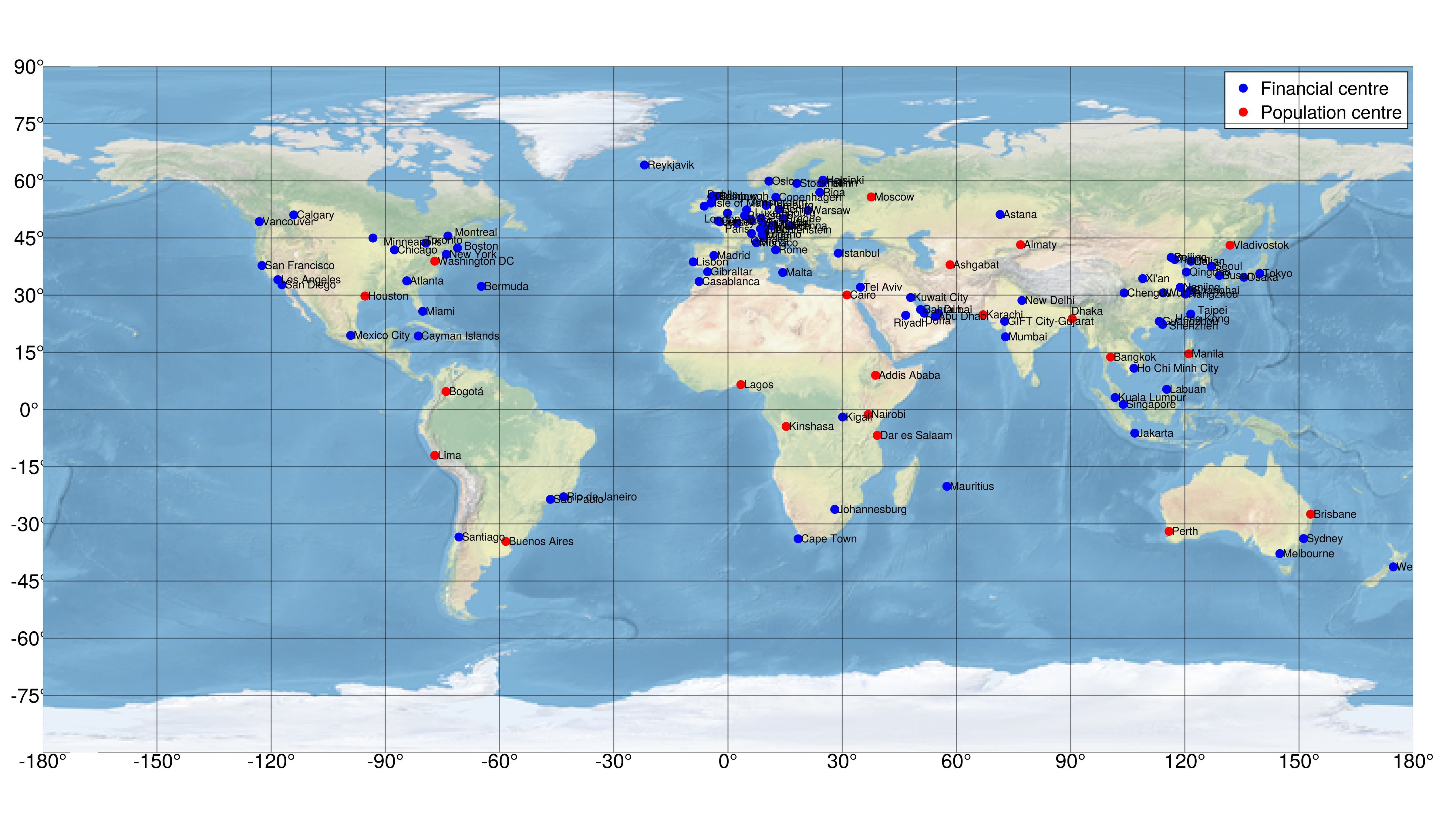}
    \caption{\textbf{City grid used to construct the traffic matrix.} Blue markers denote global financial centers, while red markers denote major population hubs not included in the financial centers list. Together, these cities provide broad geographic coverage across continents and regions and serve as application-relevant endpoints for evaluating backbone connectivity.}
    \label{fig:traffic_matrix}
\end{figure}

\subsubsection{Latency conditioned average link strength}
Connectivity and waiting-time metrics quantify when a sufficiently connected backbone becomes available, but they do not quantify the strength of the elementary entanglement links produced by the satellite service process. We therefore also report the average active-link strength under a latency constraint.

We evaluate the network at discrete simulation epochs
\begin{align*}
    t_k = t_{\mathrm{start}} + k\,\delta t,
    \qquad k=1,\ldots,K,
\end{align*}
where $t_{\mathrm{start}}$ is the simulation start time, and $\delta t$ is the simulation time step. At epoch $t_k$, let $w_{ij}(t_k)$ denote the best direct satellite-serviced elementary-link strength available between ground stations $i$ and $j$, after applying the visibility, zenith-angle, and minimum-link-strength constraints. The active elementary-link set at epoch $t_k$
\begin{align*}
    \mathcal{E}_k := \{(i,j): w_{ij}(t_k)>0\}.
\end{align*}
We define the epoch-level active-link strength as
\begin{align*}
    S_k &=
        \begin{cases}
        \displaystyle
        \frac{1}{|\mathcal{E}_k|}
        \sum_{(i,j)\in \mathcal{E}_k} w_{ij}(t_k),
        & |\mathcal{E}_k|>0,\\[6pt]
        0, & |\mathcal{E}_k|=0.
        \end{cases}
\end{align*}
For a trace with $K$ evaluated epochs, the average active-link strength is
\begin{align*}
    \bar S
    \;=\;
    \frac{1}{K}
    \sum_{k=1}^{K} S_k .
\end{align*}

\noindent We refer to the quantity $\bar S$ as latency-conditioned when it is reported only for configurations whose time-to-connectivity satisfies a prescribed latency requirement, e.g., maximum wait time of $W_{\mathrm{max}} = 10$~s to achieve a connectivity threshold $\vartheta \geq 50\%$ (see Appendix~\ref{app:forward_wait_time}). Thus, the conditioning is applied at the configuration level: among configurations that meet the latency target, $\bar S$ measures the typical strength of the active elementary links produced by the satellite-serviced backbone.

\subsection{Evaluation}\label{subsec:evaluation}
\subsubsection{Design Question 1: Ground-Station Geometry}

A critical design question for a satellite-serviced quantum backbone is how to spatially distribute ground stations to support low-latency, large-scale connectivity. Because satellite visibility and ground-station density both vary strongly with latitude, the choice of grid geometry directly affects how efficiently satellite coverage can be translated into usable end-to-end connectivity.

We consider three ground-station placement strategies: (1) longitudinal (equal-angular), (2) uniform (equidistant), and (3) anisotropic triangular grids. These approaches differ in how they distribute stations across latitudes and longitudes, and in how well they account for geometric effects such as the convergence of longitudes near the poles and the latitude dependence of satellite pass density. In all cases, the satellite constellation, hardware assumptions, and traffic matrix are held fixed, allowing the role of ground-station geometry to be isolated.

We compare these layouts using time-to-connectivity and latency-conditioned average link strengths over the traffic matrix across a range of satellite constellation designs and budgets. This analysis allows us to assess how different ground-station layouts influence backbone responsiveness and simultaneous service capability, and sets the stage for the first design insight presented in Sec.~\ref{sec:Results} (Results).

\subsubsection{Design Question 2: Satellite Constellation Structure}
A second key design question concerns the structure of the satellite constellation itself. Given a fixed satellite budget, it is not obvious whether connectivity can be improved by augmenting a dominant shell with a small number of satellites in a geometrically distinct orbit. While both approaches---single-shell or augmented dual-shell---can provide global coverage in principle, they differ in how satellite visibility overlaps in time and space, and therefore in how connectivity emerges and evolves across the network.

Single-shell constellations concentrate satellites at a common altitude and inclination, leading to regular and geographically structured visibility patterns. While such designs can provide broad coverage, their fixed inclination also determines which latitude bands receive the most service opportunity and which regions are weakly served. An augmented dual-shell constellation retains a primary shell for baseline coverage while introducing a secondary shell with distinct orbital parameters. This secondary shell is not intended to increase raw capacity, but to add inclination diversity, improving serviceability of otherwise weakly covered regions and increasing the availability of intermediate relay nodes without shifting the entire constellation to a higher-inclination orbit.

To evaluate this design question, we compare single-shell (with inclination angle $i = 53^\circ$) and an augmented dual-shell constellations (with inclination angles $i \in \{53^\circ (\mathrm{primary}), 98^\circ (\mathrm{secondary})\}$ under matched satellite budgets and identical ground-station layouts, hardware assumptions, and traffic matrix. Performance is assessed using average concurrent connectivity and time to connectivity over the traffic matrix, allowing us to isolate the impact of adding orbital diversity to an otherwise single-shell backbone. This analysis forms the basis for the second design insight discussed in Section~\ref{sec:Results} (Results).

\subsubsection{Design Question 3: Satellite Concurrent Service Capability}
A third design question concerns the servicing capability of individual satellites. Many existing satellite-assisted quantum networking proposals assume that a satellite establishes entanglement with only a single ground station, or a single ground-station pair, at any given time~\cite{khatri_spooky_2021, shao_hybrid_2025}. While this assumption simplifies payload design and control, it places a fundamental limit on the amount of simultaneous connectivity that can be supported within a satellite's footprint, regardless of the density of ground stations below.

An alternative is to equip satellites with multiple entanglement sources or terminals, allowing them to service multiple ground stations concurrently. Recent progress in compact entanglement sources has made such architectures increasingly plausible within realistic payload constraints~\cite{mahmudlu_fully_2023, simmons_scalable_2024, chapman_-chip_2025}. At the same time, increasing the number or aperture of optical terminals introduces costs in mass, pointing agility, and maneuverability, and can complicate satellite operations. It is therefore not \textit{a priori} clear when increased servicing capability leads to meaningful network-level gains, given platform complexity increases.

To evaluate this design question, we compare configurations in which satellites are restricted to serving a single ground-station pair per time step against configurations that allow satellites to distribute entanglement to multiple ground stations concurrently. All other parameters---including constellation structure, ground-station grid layout, and traffic matrix---are held fixed. Performance is assessed using average concurrent connectivity and time to connectivity over the traffic matrix, allowing us to isolate the impact of satellite servicing capability on backbone performance. The resulting comparisons form the basis for the third design insight presented in Section~\ref{sec:Results}.

We emphasize that the term multi-party connectivity is used here in an architectural rather than state-preparation sense. In the simulations reported below, each active edge represents a bipartite Bell-pair opportunity between two ground stations. However, when number of optical terminals per satellite ($T$) is greater than two, the same satellite footprint and terminal configuration could also be used to distribute genuinely multi-partite entangled states, such as GHZ states, among three or more simultaneously serviced ground stations. We therefore interpret the multi-party service model as a general concurrency abstraction: in this work it is instantiated using bipartite links, while future protocol-specific studies can replace the induced edge model with explicit multi-partite entanglement generation.

\subsubsection{Secondary Architectural Parameters}
In addition to the three primary design questions discussed above, we also examine the influence of several secondary architectural parameters, such as number of satellites per orbital plane, number of orbital planes, and total satellite budget. Satellite altitude is treated separately as a primary physical parameter, as it strongly influences both coverage geometry and link loss. Each of these parameters impacts coverage density, link quality, and temporal availability, and therefore interacts with both concurrent connectivity and time to connectivity.

Satellite altitude primarily affects link loss and achievable entanglement rates, with higher altitudes trading improved coverage footprints for reduced link strength. The number of satellites per plane and the number of orbital planes together determine how evenly satellite visibility is distributed in time and across latitudes, influencing how quickly connectivity emerges and how frequently coverage gaps occur. Increasing the total satellite budget generally improves both connectivity and wait times, but with diminishing returns once coverage becomes dense.

In our evaluation, these parameters are varied systematically to assess their interaction with the primary design choices, while keeping the overall architecture fixed. This distinction separates first-order architectural effects---such as ground-station geometry, constellation structure, and satellite servicing capability---from secondary tuning parameters. Detailed results for these secondary parameters are reported alongside the primary comparisons and, where appropriate, deferred to supplementary material~\cite{supplementary_materials}.

\subsubsection{Ground-station visibility and evaluation scope.}
The terrestrial lattice is generated independently of the satellite constellation and is used as a common candidate backbone across all configurations. Consequently, for single-shell constellations with fixed inclination, some high-latitude ground stations in the full lattice may never become visible to any satellite. These stations would be redundant in a deployment-specific ground-station layout and would naturally be pruned after a constellation-specific visibility analysis.

This does not affect the reported performance metrics. Connectivity, waiting time, and active-link strength are evaluated over the city-pair traffic matrix. All cities in the traffic matrix are visible in at least some configurations within the sampled altitude range. Ground stations outside the traffic matrix are used only as intermediate relay nodes when they are serviceable; stations with no satellite visibility cannot participate in path formation and therefore do not inflate connectivity or link-strength metrics.

Raising the inclination of a single shell would reduce the number of permanently unserviceable high-latitude stations, but it also redistributes satellite availability toward higher latitudes. This creates a different trade-off: additional high-latitude relay coverage may come at the cost of reduced service opportunity over the latitude bands that dominate the traffic matrix. The augmented dual-shell design addresses this tension by adding inclination diversity rather than shifting the entire constellation to a higher-inclination orbit. Our conclusions comparing single-shell and augmented dual-shell designs are therefore insensitive to pruning permanently unserviceable ground stations from the candidate lattice.

\subsection{Simulation methodology}
\label{subsec:simulation_methodology}

We evaluate satellite backbone performance using a discrete-time network simulator that couples satellite orbital dynamics, geometric visibility constraints, and an abstracted entanglement service model. The simulator is implemented in Julia and operates at a fixed temporal resolution of $\delta t = 1$~s over a simulation horizon $T_{\mathrm{sim}}$. All results reported in this manuscript are derived from this simulator. We do not model physical-layer impairments like weather, atmospheric turbulence or cloud cover. Although they introduce additional multiplicative loss, they are shared equally by all architectures and do not alter the architectural ordering studied here.

\vspace{-1em}
\paragraph{Orbital propagation.}
Each satellite is evolved independently using a two-body orbital model specified by its altitude, inclination, right ascension of the ascending node (RAAN), and initial mean anomaly. The mean anomaly parametrizes the satellite's orbital phase at the reference epoch and is advanced in time according to the mean motion, providing a computationally efficient description of motion along a Keplerian trajectory. 

Satellite positions are first computed in an Earth-centered inertial (ECI) frame, which is non-rotating with respect to distant stars and is convenient for computing orbital dynamics. To account for Earth's rotation when evaluating ground visibility, these positions are converted into the Earth-centered Earth-fixed (ECEF) coordinates using standard Earth orientation parameters. Ground stations are treated as fixed points in the ECEF frame, defined by their geodetic latitudes and longitudes.

For link evaluation, satellite and ground-station positions are expressed in a local `topocentric' frame centered at each ground station. This local East--North--Up (ENU) coordinate system enables direct computation of elevation and azimuth angles. A satellite is considered visible to a ground station if its elevation exceeds a specified minimum threshold, ensuring that only physically realizable free-space optical links are included.

\vspace{-1em}
\paragraph{Visibility model.}
At each time step $t$, a satellite $s$ is considered visible to a ground station $g$ if the satellite elevation angle at the ground station exceeds a prescribed threshold. Equivalently, we impose a zenith-angle constraint
\begin{equation}
    z_{s,g}(t) \in [z_{\min}, z_{\max}],
\end{equation}
where $z_{s,g}(t)$ is the zenith angle between the satellite and ground station in the local topocentric frame. Unless otherwise stated, $z_{\max}=57^\circ$ (corresponding to a minimum elevation angle of $33^\circ$). Visibility is purely geometric and deterministic; weather effects and stochastic link outages are not modeled.

\vspace{-1em}
\paragraph{Effective slant-path length.}
The Gaussian beam radius $w(\cdot)$ is evaluated at an \emph{effective slant-path length} $d_{s,g}(t)$ between satellite $s$ and ground station $g$ at time $t$. This quantity represents the physical propagation distance along the satellite--ground line of sight, accounting for Earth curvature and the satellite's instantaneous orbital position. Specifically, $d_{s,g}(t)$ is defined as the Euclidean distance between the satellite position vector and the ground-station position vector in an Earth-centered frame, subject to the visibility constraint $z_{s,g}(t) \le z_{\max}$. This effective distance determines the diffraction-limited beam expansion and therefore directly controls the geometric coupling efficiency $\eta_{\rm geo}(t)$. The operating wavelength is set to $810$nm throughout this work as typical for optical free-space communications~\cite{hearne_wavelength_2025}.

\vspace{-1em}
\paragraph{Satellite service sets.}
For each satellite $s$ at time $t$, we construct the set of visible ground stations
\begin{equation}
    \mathcal{V}_s(t) = \{ g \in \mathcal{G} : z_{s,g}(t) \le z_{\max} \}.
\end{equation}
From this set, the satellite selects a service subset $\mathcal{U}_s(t) \subseteq \mathcal{V}_s(t)$ consisting of at most $T$ ground stations, chosen as the $T$ nearest visible stations according to instantaneous satellite--ground distance. The parameter $T$ captures the maximum number of ground stations that can be concurrently serviced by a satellite and is determined by the assumed payload capability.

\vspace{-1em}
\paragraph{Service topology models.}
Edges between ground stations are generated independently by each satellite according to one of two service models.

\emph{Bi-partite connectivity (BPC):}
In BPC mode, a satellite induces feasible entanglement links between all unordered pairs of ground stations in its service set $\mathcal{U}_s(t)$. For each pair $(g_i,g_j) \subset \mathcal{U}_s(t)$, a weighted edge is instantiated if the corresponding expected entanglement-generation rate exceeds a minimum feasibility threshold, taken here as $R_{\min}=1$ Bell pair per second. In all results reported in this manuscript, BPC operation is instantiated with $T=2$, so that at most one ground-station pair can be active per satellite per time shard.  The threshold $R_{\min}=1~\mathrm{Hz}$ is chosen as the feasibility cutoff corresponding to the minimum rate required to support heralded Bell-pair generation within a one-second shard; varying this threshold shifts absolute performance values but does not alter the qualitative architectural trends reported here.

\emph{Multi-party connectivity (MPC):}
In MPC mode, a satellite designates a single hub ground station (chosen as the nearest element of $\mathcal{U}_s(t)$) and generates edges only between the hub and the remaining stations in $\mathcal{U}_s(t)$. Additionally, nearest-neighbor edges among the non-hub stations are added, forming a constrained hub--spoke--ring topology. This model captures increased per-satellite concurrency enabled by multiple optical terminals within a single satellite footprint. Similar to the BPC case, $R_{\min}=1~\mathrm{Hz}$.

\vspace{-1em}
\paragraph{Link efficiency model.}
For a satellite--ground link between satellite $s$ and ground station $g$ at time $t$, we define
\begin{equation}
\eta_{s,g}(t)
= \underbrace{\left(1-\exp\!\left[-\frac{2a_R^2}{w^2(d_{s,g}(t))}\right]\right)}_{\eta_{\rm geo}(t)} \cdot \underbrace{\eta_{\rm atm}\!\left(z_{s,g}(t)\right)}_{\text{atmospheric transmission}},
\end{equation}
with
\begin{equation}
\label{eq:eta_atm}
\eta_{\rm atm}(z)=\eta_{\rm zen}^{\sec z} \qquad (\text{equivalently } \eta_{\rm atm}(z)=\exp[-\tau_{\rm atm}\sec z],\ \eta_{\rm zen}=e^{-\tau_{\rm atm}}).
\end{equation}
Here $a_R$ is the receiver aperture radius, $w(\cdot)$ is the Gaussian beam radius evaluated at the effective slant-path propagation distance used in the simulator (accounting for Earth curvature), and $z_{s,g}(t)$ denotes the zenith angle at the ground station. \noindent
Here $\eta_{\rm geo}$ captures loss due to diffraction-limited beam broadening and geometric collection of a Gaussian beam by a circular aperture, while $\eta_{\rm atm}(z)=\exp[-\tau_{\rm atm}\sec z]$ models elevation-dependent atmospheric absorption under clear-sky conditions. The zenith transmission coefficient $\eta_{\rm zen}=\exp(-\tau_{\rm atm})$ is treated as a fixed parameter representing clear-sky conditions at the operating wavelength ($810$nm); we do not model weather, turbulence, pointing jitter, or background-noise-induced false heralding in the present study. 

For a ground-station pair $(g_i,g_j)$ served by satellite $s$, the effective pairwise transmission efficiency is modeled as the product of the two downlink efficiencies,
\begin{equation}
    \eta_{i,j}(t) = \eta_{s,g_i}(t)\,\eta_{s,g_j}(t).
\end{equation}
We convert this efficiency into an expected entanglement-generation rate
\begin{equation}
    R_{i,j}(t) = R_s \, \eta_{i,j}(t),
\end{equation}
where $R_s$ is the source attempt rate. In multi-point connectivity (MPC) mode, entanglement sources are time-multiplexed across edges within a shard; $R_s$ is interpreted as the per-edge attempt rate after absorbing this duty cycle, so that relative comparisons across architectures are unaffected. An edge is instantiated if $R_{i,j}(t) \ge R_{\min}$, where $R_{\min}$ denotes a minimum expected entanglement-generation rate used to define link feasibility (taken as one Bell pair per second). If multiple satellites simultaneously serve the same ground-station pair, we retain the one with the maximum rate among them. Accordingly, reported link strengths should be interpreted as relative architectural weights rather than predictions of absolute end-to-end entanglement rates for a specific hardware platform.

\vspace{-1em}
\paragraph{Ground-station graph construction.}
At each time step $t$, we construct a weighted undirected graph
\begin{equation}
    G_t = (\mathcal{G}, E_t),
\end{equation}
where vertices correspond to ground-stations and edges correspond to feasible entanglement links generated by the satellites at time $t$. The edge weight $w_{i,j}(t)$ is defined as the maximum $R_{i,j}(t)$ over all contributing satellites. We define the instantaneous average link strength as the mean of all nonzero edge weights in $G_t$. This metric represents the typical per-link transmission efficiency available at time $t$, conditioned on link existence.

\vspace{-1em}
\paragraph{Connectivity metrics and wait-time analysis.}
At each time step $t$, satellites induce a weighted ground-station graph $G_t$ as described above. Network connectivity is quantified as the size of the largest connected component (LCC) \cite{pastor2007evolution} of this induced graph. To evaluate connectivity under a maximum allowable waiting time $W_{\max}$, we define a union graph
\begin{equation}
    \widetilde{G}_{t,W_{\max}} = \left(\mathcal{G}, \bigcup_{\tau=t}^{t+W_{\max}} E_\tau \right),
\end{equation}
and compute the corresponding LCC.

Connectivity metrics reported in the main text are evaluated primarily with respect to the traffic matrix: each city in the traffic matrix is mapped to its nearest grid point, and connectivity is assessed on the induced subgraph corresponding to the traffic-matrix nodes.

Time-to-connectivity for a given threshold $\vartheta$ is computed using a forward-wait construction: for each start time $t$, we record the minimum waiting time $w$ such that at least a fraction $\vartheta$ of the total number of traffic-matrix pairs have an entanglement path available (See Appendices~\ref{app:forward_wait_time} and~\ref{app:metrics_methodology}).

\vspace{-1em}
\paragraph{Traffic-matrix evaluation.}
Each city in the traffic matrix is mapped to its nearest ground station. City-level connectivity metrics are evaluated by inducing subgraphs of $G_t$ or $\widetilde{G}_{t,W_{\max}}$ on the corresponding ground-station subsets and testing path existence between city pairs as required by the traffic matrix. 
To note, Bell-state measurement success probability, detector noise, background photons, and quantum-memory decoherence are not explicitly modeled in the present simulator; link weights therefore represent idealized entanglement-generation opportunities rather than finalized high-fidelity Bell-pair rates.

\vspace{-1em}
\paragraph{Reproducibility, waiting-time construction, and uncertainty estimates.}
All simulations are performed using a single simulation start epoch and propagated continuously over a fixed analysis horizon of four hours. Connectivity and waiting-time metrics are computed deterministically relative to specified connectivity thresholds. For each threshold, we form a Boolean time series indicating whether the threshold is met, and extract the durations of contiguous \emph{down intervals} (time spans during which the threshold is not satisfied). Reported percentiles (e.g., $p10/p50/p90$) are computed over the empirical distribution of these down-interval durations within the same orbital realization (see Appendix~\ref{app:metrics_methodology}).

No averaging over multiple independent orbital realizations or randomized initial conditions is performed. All constellation parameters, ground-station configurations, visibility thresholds, and simulator settings are recorded and exported alongside results to enable full reproducibility. Because satellite motion is deterministic and results are evaluated over a single orbital realization, reported variability reflects temporal fluctuations rather than ensemble uncertainty across independent systems.

Shaded uncertainty bands shown in figures correspond to $\pm 1$ standard-error-of-the-mean (SEM) estimates computed from the same deterministic realization while accounting for temporal correlation. Specifically, for any scalar time series (e.g., average active-link strength) or derived event sequence (e.g., down-interval durations), we estimate the integrated autocorrelation time $\tau_{\mathrm{int}}$ using an initial-positive-sequence estimator: sample auto-correlations are summed over positive lags and truncated at the first non-positive value~\cite{grotendorst_quantum_2002} (See Appendix~\ref{app:metrics_autocorr_problem} and ~\ref{app:metrics_tauint}). We then compute an effective sample size $N_{\mathrm{eff}} = N/(2\tau_{\mathrm{int}})$ and the corresponding SEM as $\mathrm{SE}=\sigma/\sqrt{N_{\mathrm{eff}}}$. These uncertainty bands therefore quantify correlation-adjusted temporal variability within a single run, rather than statistical uncertainty across independent constellation draws. Waiting-time statistics are event-based; under the initial-positive-sequence estimator, the down-interval duration sequences typically truncate immediately (yielding $\tau_{\mathrm{int}}\approx 0.5$ events), indicating no detectable positive inter-event autocorrelation in the finite trace. By contrast, continuous per-second metrics such as average link strength and LCC size show nontrivial temporal correlations on the order of tens of seconds, motivating separate reporting of autocorrelation scales. See Tables~\ref{tab:tau_int_waits_by_th} and \ref{tab:tau_int_timeseries} for details.

\begin{table}[h]
\centering
\caption{Autocorrelation scales for connectivity waiting-time statistics, computed from \emph{event-level} sequence of down-interval durations $\{L_n^{(\vartheta)}\}$ (see Appendix~\ref{app:forward_wait_time} and Appendix~\ref{app:metrics_methodology} for definitions), where each event is one maximal contiguous below-threshold episode for threshold $\vartheta$. ere $\tau_{\mathrm{int}}$ is estimated using the initial-positive-sequence (IPS) truncation rule applied to the event sequence and is therefore reported in \emph{events} (intervals), not seconds. $N_{\mathrm{eff}}$ denotes the corresponding effective number of down-intervals events used for standard-error calculations of event-averaged quantities; $N_{\mathrm{eff}}=0$ indicates that the threshold is satisfied essentially continuously in the finite trace, yielding too few down events for meaningful uncertainty quantification. Reported waits are the median $p10/p50/p90$ down-run durations (seconds) across experiments. In most cases IPS truncates immediately, giving $\tau_{\mathrm{int}}\approx 0.5$ events and indicating no detectable positive inter-event autocorrelation within the finite trace.} 
\label{tab:tau_int_waits_by_th}
\begin{tabular}{|c |c |c |c |c |}
\hline
\textbf{Threshold} & \textbf{Wait stat} & $\boldsymbol{\tau_{\mathrm{int}}}$ \textbf{(events)} & $\boldsymbol{N_{\mathrm{eff}}}$ & \textbf{Wait (sec)} \\
\hline
0.5 & $p10/p50/p90$ & $0.5 \pm 0.0$ & $0$--$56$ & $1.0$ / $5.0$ / $13.0$ \\
0.7 & $p10/p50/p90$ & $0.5 \pm 0.0$ & $0$--$55$ & $2.3$ / $23.0$ / $64.2$ \\
0.9 & $p10/p50/p90$ & $0.5 \pm 0.0$ & $1$--$55$ & $6.1$ / $73.0$ / $121.0$ \\
\hline
\end{tabular}
\end{table}

\begin{table}[h]
\centering
\caption{Typical integrated autocorrelation times and effective sample sizes for per-second time-series metrics computed over the full simulation horizon (single deterministic orbital realization). $\tau_{\mathrm{int}}$ is estimated from the epoch-level sequences using Geyer's initial-positive-sequence (IPS)~\cite{geyer1992practical} truncation rule (pair-sum form) and is reported in seconds (with $\delta t=1$~s). $N_{\mathrm{eff}}=N/(2\tau_{\mathrm{int}})$ denotes the corresponding effective number of time samples used for standard-error-of-the-mean estimates. These metrics are not conditioned on connectivity thresholds in the current analysis pipeline.}
\label{tab:tau_int_timeseries}
\begin{tabular}{|l|c|c|}
\hline
\textbf{Metric} & $\boldsymbol{\tau_{\mathrm{int}}}$ \textbf{(seconds)} & $\boldsymbol{N_{\mathrm{eff}}}$ \\
\hline
Average active-link strength & $17.7 \pm 6.5$ & $272$--$538$ \\
Largest connected component size & $18.5 \pm 8.8$ & $252$--$659$ \\
\hline
\end{tabular}
\end{table}

\vspace{-1em}
\paragraph{Fidelity considerations and architectural scope.}
End-to-end entanglement fidelity and finalized Bell-pair throughput are not explicitly tracked in the present simulator. Instead, link feasibility and strength are interpreted as \emph{entanglement-delivery opportunities}, capturing whether large connected components capable of supporting end-to-end entanglement delivery can emerge and persist across the backbone, and with what degree of concurrency.

We assume that elementary satellite--ground entanglement can be generated at high rates and with high initial fidelity $F_0$ under favorable conditions, consistent with recent free-space and satellite demonstrations~\cite{yin2017satellite, beckert2019space, villar2020entanglement,lim2026free}. Under idealized Bell-state measurements and classical feedforward, non-unit success probabilities primarily reduce usable throughput rather than eliminating connectivity altogether, provided sufficient multiplexing is available.

The connected components observed in our simulations imply effective path lengths typically on the order of fewer than $15$--$20$ satellite-served hops for traffic-matrix connectivity, substantially smaller than the number of repeater segments required for terrestrial networks spanning comparable distances~\cite{mantri20241-2-wayquantum}. As a result, even in the presence of non-unit BSM success probability and finite elementary-link fidelity, non-zero end-to-end entanglement rates are expected when high-rate sources and multiplexing are employed.

Our focus on time-to-connectivity and delivery concurrency therefore isolates a distinct architectural bottleneck: whether sufficiently large, simultaneous entanglement-supporting components can form under realistic satellite visibility and servicing constraints, independent of the specific physical-layer protocol used to realize each link.

\begin{table}[t]
\centering
\caption{Core simulation parameters used throughout the evaluation unless explicitly varied. These values are not intended to represent a specific flight-qualified payload, but to define a consistent operating point for architectural comparison.} 
\begin{tabular}{lll}
\toprule
    \textbf{Parameter} & \textbf{Symbol} & \textbf{Value} \\
    \midrule
    \multicolumn{3}{l}{\textbf{Time discretization}} \\
    Simulation horizon & $T_{\mathrm{sim}}$ & $4$ hr \\
    Time step & $\delta t$ & $1$ s \\
    \multicolumn{3}{l}{\textbf{Visibility model}} \\
    Maximum zenith angle & $z_{\max}$ & $57^\circ$ \\
    Minimum elevation & $\theta_{\min}$ & $33^\circ$ \\
    \multicolumn{3}{l}{\textbf{Ground-station grid}} \\
    Base angular spacing & $\Delta_0$ & $3.6^\circ$ \\
    Land-snap radius & --- & $100$ km \\
    \multicolumn{3}{l}{\textbf{Satellite constellation}} \\
    Primary inclination & $i$ & $53^\circ$ \\
    Polar inclination & $i_{\mathrm{polar}}$ & $98^\circ$ \\
    \multicolumn{3}{l}{\textbf{Optical link model}} \\
    Receiver aperture radius & $a_R$ & $0.5$ m \\
    Beam waist parameter & $w_0$ & $0.10$ m \\
    Wavelength & $\lambda$ & $810$ nm \\
    Zenith transmission & $\eta_{\mathrm{zen}}$ & $0.8$ \\
    Source attempt rate & $R_s$ & $100$ MHz \\
    \multicolumn{3}{l}{\textbf{Edge feasibility}} \\
    Minimum link score & $R_{\min}$ & $1~\mathrm{s}^{-1}$ \\
\bottomrule
\end{tabular}

\label{tab:sim_params}
\end{table}

\section{Results}\label{sec:Results}

In this section, we present results corresponding to the three design questions introduced in Section~\ref{sec:Perf_Eval}. Unless otherwise stated, results are reported using the traffic matrix described earlier, with concurrent connectivity and time to connectivity as the primary performance metrics. Section~\ref{subsec:design_insight1} examines the impact of ground-station geometry, Section~\ref{subsec:design_insight2} evaluates the role of constellation structure, and Section~\ref{subsec:design_insight3} analyzes the effect of satellite servicing capability. Together, these results lead to a set of architectural design insights for satellite-serviced global quantum backbones. Throughout the study, satellite altitude, number of orbital planes, and number of satellites per orbital plane are treated as explicit sweep parameters. Unless otherwise stated, we report results for altitudes ranging from 500\,km to 1400\,km, orbital-plane counts ranging from 30 to 360, and satellites per plane ranging from 9 to 60 (see Sec.~\ref{sec:system} for definitions). All results report best performance over these configurational sweeps. Unless stated otherwise, the polar fraction $f_\mathrm{polar}$ is fixed at 10\%, and number of terminals is fixed at $7$.

Before presenting the design insights, we reiterate that the goal here is to study latency- and concurrency-focused measures. 

We use  Refs. \cite{khatri_spooky_2021, shao_hybrid_2025} as reference designs that motivate specific baseline grid geometries within our unified parameterization. In particular, equal-angular latitude–longitude grids~\cite{khatri_spooky_2021} correspond approximately to the $\alpha=-1$ regime in our parameterization, while uniform Euclidean (equi-distant) grids~\cite{shao_hybrid_2025} correspond to $\alpha=0$ over a limited latitude range. By embedding these geometries into a unified framework and evaluating them using waiting-time–conditioned connectivity metrics, we extract architectural insights rather than implementation-level performance claims.

The metrics adopted in this work---time to connectivity under finite waiting constraints, traffic-matrix connectivity, and latency-conditioned active-link strength---are chosen to probe whether a satellite-serviced system can function as a persistent global backbone, rather than as a collection of high-quality but intermittent point-to-point links.  These metrics are therefore complementary to the rate- and fidelity-centric analyses emphasized in existing literature.

\subsection{Design Insight 1: Ground-Station Geometry Should Be Anisotropic}\label{subsec:design_insight1}

We examine the impact of ground-station (GS) geometry on backbone performance by comparing longitudinal ($\alpha=-1$), uniform ($\alpha=0$), and anisotropic ($\alpha>0$) triangular grids under identical satellite constellations and hardware assumptions. Our evaluation focuses on two network-level metrics: (i) the time required to meet a specified traffic-matrix connectivity requirement, and (ii) the average strength of links that become available under an explicit latency constraint. While intermediate values $0<\alpha<0.5$ also correspond to progressively more anisotropic layouts, we do not analyze them separately here. Instead, we focus on $\alpha\geq0.5$ to provide a clear contrast between isotropic and strongly anisotropic GS geometries.

For the analysis reported in this subsection, we assume that satellites support concurrent connectivity to up to seven ground stations in a `hub--spoke--ring' fashion. In this model, each satellite is equipped with seven telescopes, and entanglement distribution is restricted to pairwise connections between neighboring ground stations. We use the augmented dual-shell constellation comprising a primary shell at $53^\circ$ inclination and a secondary polar shell at $98^\circ$, with $10\%$ of satellites allocated to the polar shell\footnote{We use these inclinations as representative values to anchor the study in a plausible LEO design space; the architectural trends reported here are not tied to these specific angles and can be re-evaluated under alternative shells.}.

Figure~\ref{fig:city_waittime} reports the satellite budget required to meet different traffic-matrix connectivity thresholds under varying latency constraints. Across all thresholds, anisotropic grids consistently achieve the target connectivity at substantially lower latency than both longitudinal and uniform grids, while requiring fewer satellite resources.

While Figure~\ref{fig:city_waittime} focuses on the \emph{latency required} to achieve a target connectivity level, it does not reveal the \emph{quality} of links that are available under explicit timing constraints. We therefore evaluate the \emph{latency-conditioned link strength}: the average aggregate entanglement rate across city pairs that satisfy a given connectivity requirement within a maximum allowable wait time $W_{\max}$. This metric captures the tradeoff between responsiveness and link quality, and complements the wait-time analysis by quantifying how much usable entanglement is available within bounded time windows.

Figure~\ref{fig:city_strength_bywaittime} reports the latency-conditioned link strength for instantaneous backbone formation, corresponding to $W_{\max}=1~\mathrm{s}$. Results are shown for three representative satellite budgets and multiple target connectivity thresholds. As shown in the plots, anisotropic grids generally achieve comparable or higher link strengths across the thresholds considered. Results for larger waiting-time windows, $W_{\max} \in \{10~\mathrm{s}, 60~\mathrm{s}, 1~\mathrm{hour}, 4~\mathrm{hours}\}$, are provided in the Supplementary Material~\cite{supplementary_materials}. These cases correspond to a forward-looking regime in which long-lived quantum memories can preserve entanglement over timescales ranging from seconds to hours.

A notable exception occurs at higher satellite budgets with lower connectivity requirements, where longitudinal layouts ($\alpha=-1$) can exhibit slightly higher instantaneous aggregate link strength than anisotropic layouts. This behavior arises from extreme ground-station clustering near the poles. In these layouts, many ground stations are concentrated at high latitudes; when polar satellites are present, these nearby stations can be simultaneously served by the same satellite acting as a hub. Because the corresponding satellite--ground paths are short, atmospheric and geometric attenuation are relatively small, leading to temporarily elevated aggregate entanglement generation rates.

However, this apparent advantage is largely a geometric artifact of polar overcrowding and does not translate into improved end-to-end network performance. Closely clustered polar GSs primarily form short local connections, whereas long-distance connectivity across the global traffic matrix still requires multi-hop entanglement swapping. As the number of hops increases, probabilistic Bell-state measurements and swapping operations compound losses, rapidly reducing usable end-to-end entanglement rates. This effect becomes particularly clear when polar satellites are removed, as shown in the Supplementary Material~\cite{supplementary_materials}. Constellations without polar satellites exhibit a pronounced drop in achievable entanglement rate, especially for layouts with $\alpha \le 0$. In these cases, the densely clustered polar ground stations fall outside the footprint of the mid-inclination shell ($53^\circ$) and can no longer be served, causing the apparent rate advantage to disappear. This confirms that the elevated rates observed for $\alpha \le 0$ are localized servicing artifacts caused by polar overcrowding, rather than evidence of improved global backbone performance.

Taken together, these results show that GS geometry is a first-order determinant of the \emph{latency} with which a satellite-serviced backbone can satisfy traffic-matrix connectivity requirements under realistic timing constraints. By redistributing GS density away from polar regions and toward latitudes with sparser satellite coverage, anisotropic grids accelerate the onset of global connectivity while reducing overall resource requirements. The qualitative trends observed under these assumptions are robust to strictly pairwise connectivity and persist in single-shell constellations, as discussed in subsequent subsections.

\begin{figure}[!htbp]
\begin{subfigure}[t]{0.95\linewidth}
\centering
    \begin{subfigure}[t]{0.55\linewidth}
    \centering
    {(a) Proportion of Connected Global Financial \\ and Population Centres in the Traffic Matrix  ($\vartheta$)\par}
        \begin{subfigure}[t]{0.49\linewidth}
            \includegraphics[height=4.4cm, valign=t, trim=0 2.5cm 26cm 0.5cm, clip]{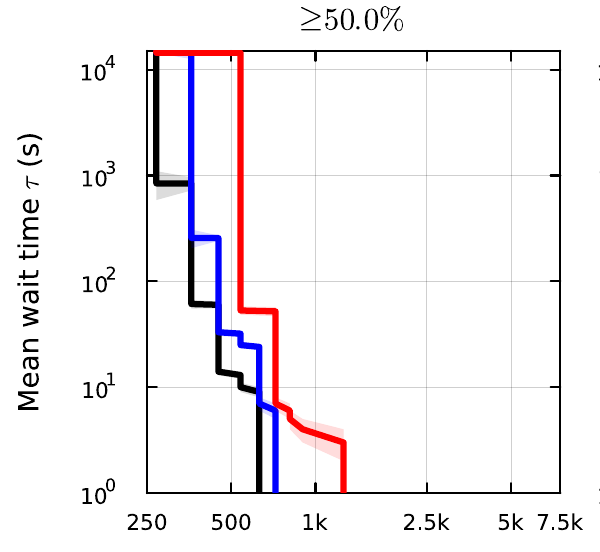}
        \end{subfigure}
            \hspace{-0.3em}
        \begin{subfigure}[t]{0.43\linewidth}
            \includegraphics[height=4.4cm, valign=t, trim=27.4cm 2.5cm 0 0.5cm, clip]{figures/k_7_pf_0.1_city_wait_vs_sats_connectivity_panels.pdf}
        \end{subfigure}
        {\small Satellite Budget (\# of satellites available)\par}
    \end{subfigure}
    \hspace{-1em}
    \begin{subfigure}[t]{0.43\linewidth}
        \centering
    {(b) Satellite Budget\par}
        \raisebox{-0.5cm}{%
        \includegraphics[height=4.4cm, valign=t, trim=1.3cm 2.3cm 17.4cm 0.5cm, clip]{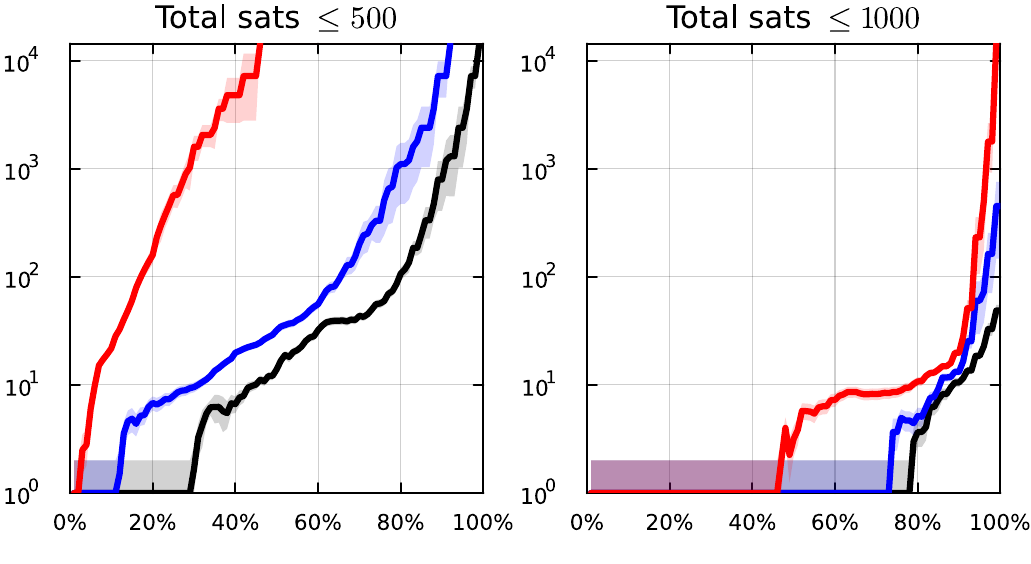}
        }
        {\small Connectivity Threshold $\vartheta$ \par}
    \end{subfigure}

\end{subfigure}

    \begin{subfigure}{0.95\linewidth}
        \includegraphics[width=\linewidth, trim = 5cm 0 5cm 10.8cm, clip]{figures/pf_0.1_k_7_wait_vs_sats_connectivity_panels_continousthreshold.pdf}
    \end{subfigure}
            
\caption{
    \textbf{Satellite budget–latency tradeoffs for traffic-matrix connectivity under different ground-station geometries. Curves correspond to longitudinal ($\alpha=-1$), isotropic ($\alpha=0$), and anisotropic ($\alpha>0$) ground-station grids.}
    Shaded regions indicate $\pm$1 standard error over simulation realizations.
    (\textbf{a}) Mean waiting time to achieve fixed global connectivity thresholds (50\%, and 100\% of city pairs connected) as a function of the available satellite budget.
    Vertical lines mark the satellite budget at which the threshold is first satisfied.
    (\textbf{b}) Mean waiting time as a function of the required connectivity fraction, shown for fixed satellite-budget regimes (from left to right: increasing total satellites).  
    Each panel illustrates how tightening the connectivity requirement increases latency under a constrained satellite budget, and how this scaling depends on ground-station geometry. 
    Across both representations, anisotropic grids consistently achieve higher connectivity at lower latency and with fewer satellite resources, indicating improved alignment between ground-station density and latitude-dependent satellite access patterns.
    }

    \label{fig:city_waittime}
\end{figure}

\begin{figure}
    \centering
        \begin{subfigure}{0.33\linewidth}
            \centering
        {\hspace{2em}Satellite Budget $\leq 1250$\par}
        \includegraphics[height = 4.5cm, trim = 0 7.1cm 0.3cm 3.3cm,
    clip]{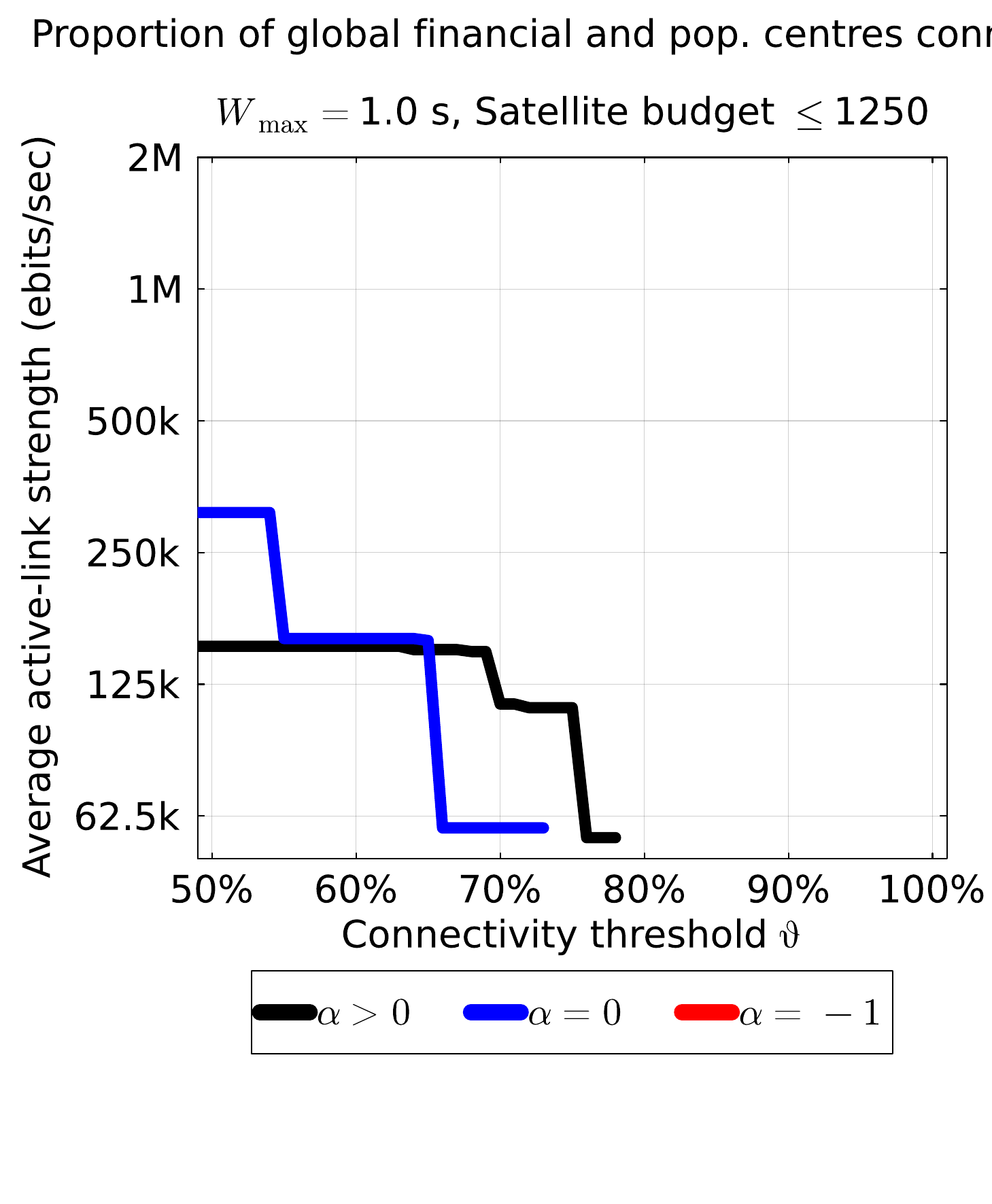}
    \end{subfigure}
    % \hfill
    \begin{subfigure}{0.30\linewidth}
        \centering
        {\hspace{1.5em}Satellite Budget $\leq 2500$\par}
        \includegraphics[height = 4.5cm, trim = 1.9cm 7.1cm 0.3cm 3.3cm,
    clip]{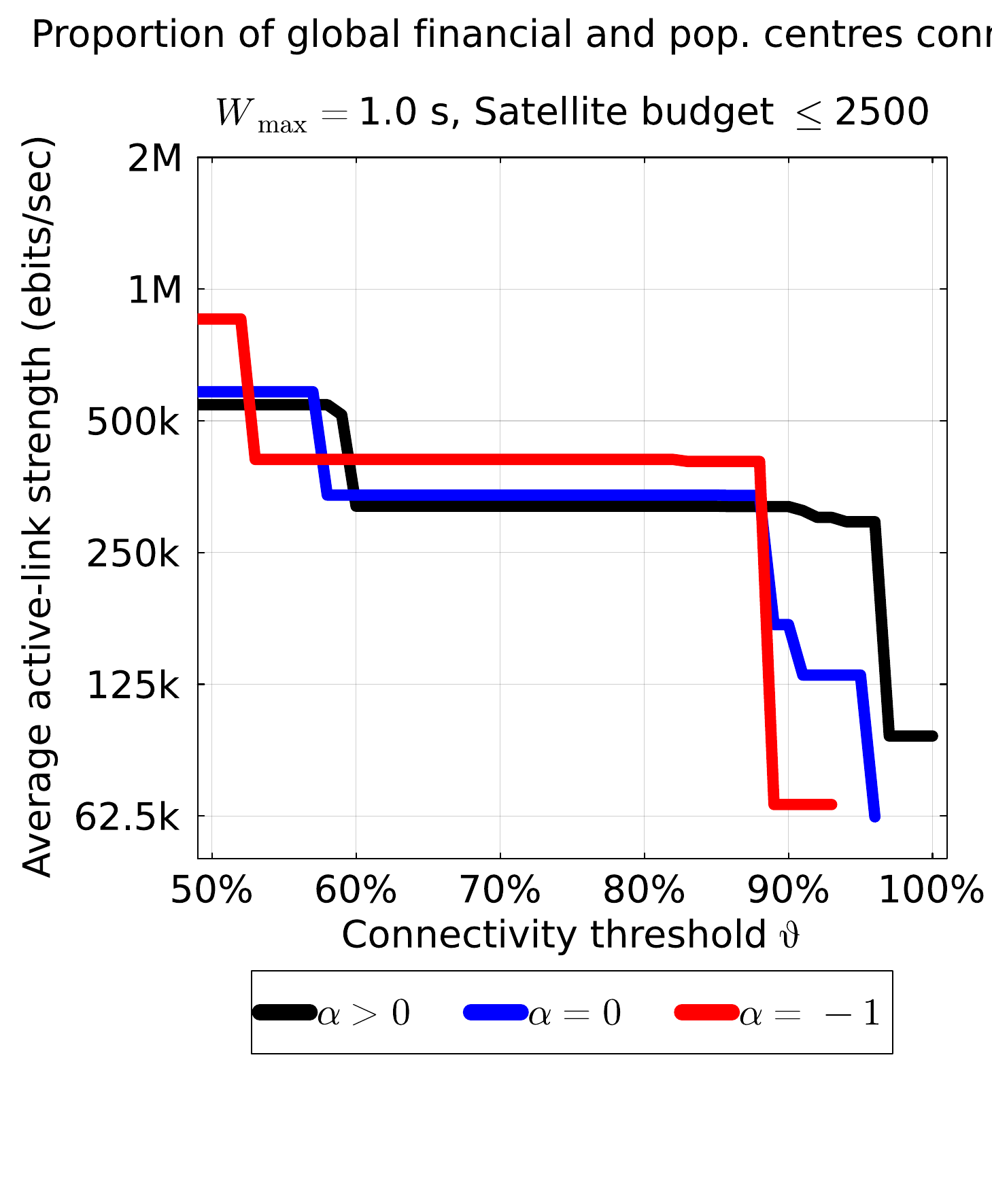}
    \end{subfigure}
    % \hfill
        \begin{subfigure}{0.3\linewidth}
        {\hspace{3em} Satellite Budget $\leq 5000$\par}
        \includegraphics[height = 4.5cm, trim = 1.9cm 7.1cm 0.3cm 3.3cm,
    clip]{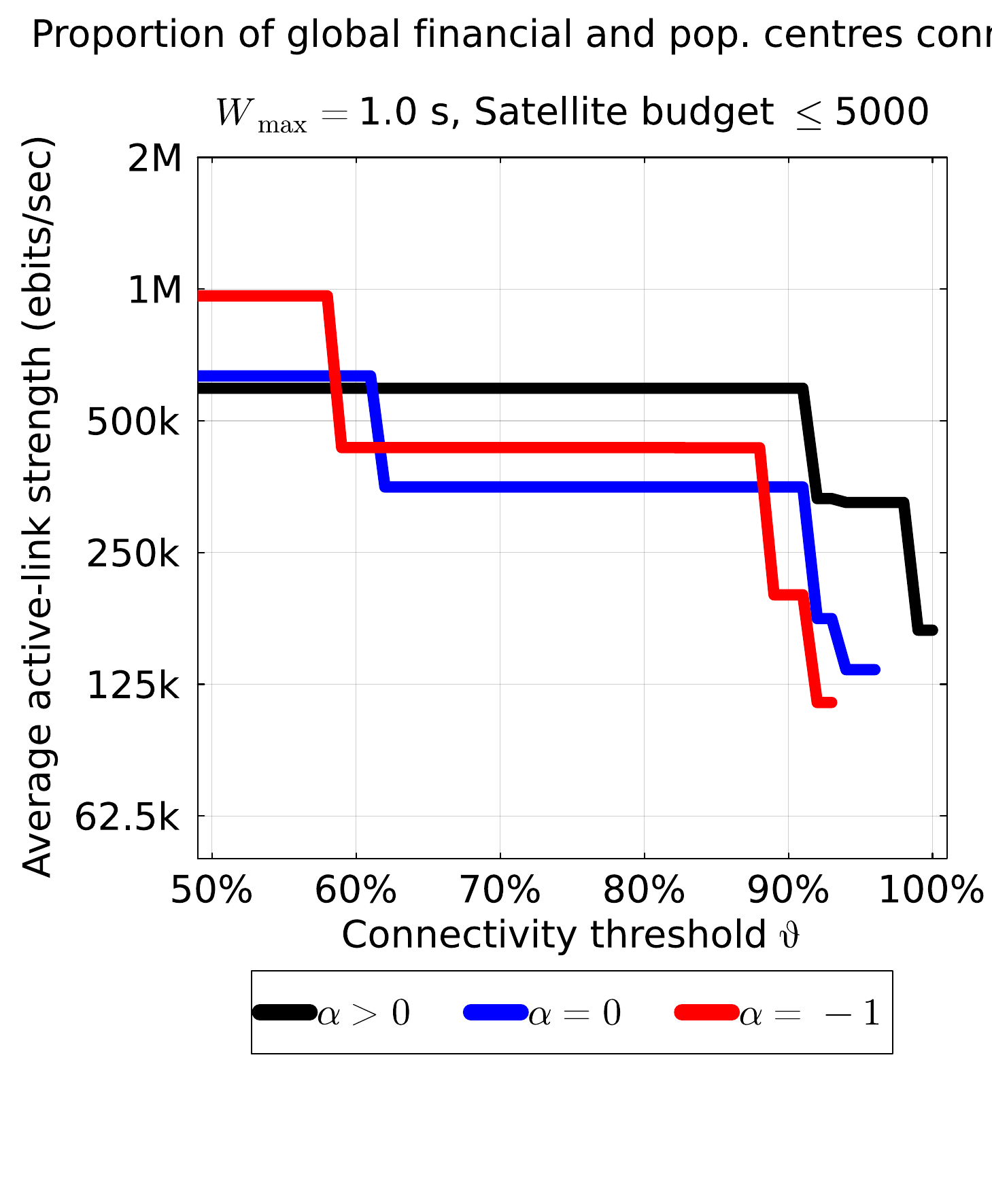}
    \end{subfigure}
    \begin{subfigure}{0.95\linewidth}
    \centering
                {Connectivity Threshold $\vartheta$ \par}
    \end{subfigure}
        \begin{subfigure}{0.95\linewidth}
        \includegraphics[width=\linewidth, trim = 5cm 0 5cm 10.8cm, clip]{figures/pf_0.1_k_7_wait_vs_sats_connectivity_panels_continousthreshold.pdf}
    \end{subfigure}
    \caption{
    \textbf{Average \emph{latency-conditioned} active-link strength as a function of the required traffic-matrix connectivity threshold $\vartheta$, under different satellite budgets}.
    The three panels restrict the design search to constellations with at most 1250, 2500, and 5000 satellites, respectively.
    For each threshold, the plotted value is the best average active-link strength achieved among configurations satisfying the satellite-budget constraint and a maximum allowable wait time of $W_{\max}=1\,\mathrm{s}$.
    Since the simulation is evaluated in $1\,\mathrm{s}$ time shards, this regime corresponds to effectively instantaneous connectivity.
    Black, blue, and red curves correspond to anisotropic ($\alpha>0$), isotropic/equi-distant ($\alpha=0$), and longitudinal ($\alpha=-1$) ground-station layouts, respectively.
    Satellites support concurrent service to up to $T=7$ ground stations using the hub--spoke--ring model.
    As the connectivity threshold is tightened, the longitudinal layout exhibits the steepest degradation in active-link strength, indicating that its high-strength operating points are concentrated at loose connectivity requirements and do not extend well to near-global instantaneous connectivity.
    By contrast, the anisotropic layout shows a more gradual loss of strength at stringent thresholds and benefits more consistently from increasing the satellite budget: once the anisotropic grid reaches a target connectivity level with comparatively fewer satellites, additional satellites primarily increase the achievable active-link strength rather than merely restoring connectivity.
    Together with the wait-time results, this indicates that anisotropic ground-station placement improves the separation between the resources needed to achieve global connectivity and the additional resources needed to strengthen the resulting backbone.
}
    \label{fig:city_strength_bywaittime}
\end{figure}

\subsection{Design Insight 2: Satellite Constellations Benefit from Inclination Diversity}
\label{subsec:design_insight2}

We compare the impact of satellite constellation structure on backbone performance by evaluating single-shell and augmented dual-shell Low-Earth-Orbit (LEO) constellations under \emph{identical} ground-station placements, traffic matrices, hardware assumptions, and total satellite budgets. In both cases, the total number of satellites is held fixed; the dual-shell constellations differ in the distribution of satellites across the two inclination shells. 

For the single-shell design, all satellites are placed in a shell with inclination $53^\circ$. For the dual-shell design, satellites are split between a primary shell at $53^\circ$ and a secondary near-polar shell at $98^\circ$, with the polar shell comprising 5\%--15\% of the satellites. In the following analysis we concentrate on the anisotropic ground station grid ($\alpha>0$) to isolate the effect of constellation design from ground-station geometry. Supplementary Materials~\cite{supplementary_materials}  shows results for the full family of ground-station layouts $\alpha\in\{-1,0,>0\}$ to illustrate how shell diversity interacts with ground-station geometry.

Figure~\ref{fig:wait_vs_sats_ads_ss}~(left) shows that redistributing a fixed satellite budget across the two shells systematically reduces the satellite budget required to achieve a given connectivity latency. This effect is most pronounced at moderate to high satellite counts and at stringent connectivity thresholds, where performance depends on maintaining simultaneous connectivity across a large fraction of the traffic matrix.
The augmented dual-shell design improves this regime by introducing inclination diversity: the additional shell increases service opportunity for latitude bands that are weakly covered by the baseline single shell and improves the availability of intermediate relay nodes, without shifting the entire constellation toward a higher-inclination orbit.
As a result, ADS can reduce the waiting time needed to reach high traffic-matrix connectivity thresholds under the same total satellite budget.

\begin{figure}[!htbp]
\centering
    \begin{subfigure}{0.33\linewidth}
        \begin{overpic}[height=4.3cm, trim=0cm 8cm 0cm 1.9cm, clip]{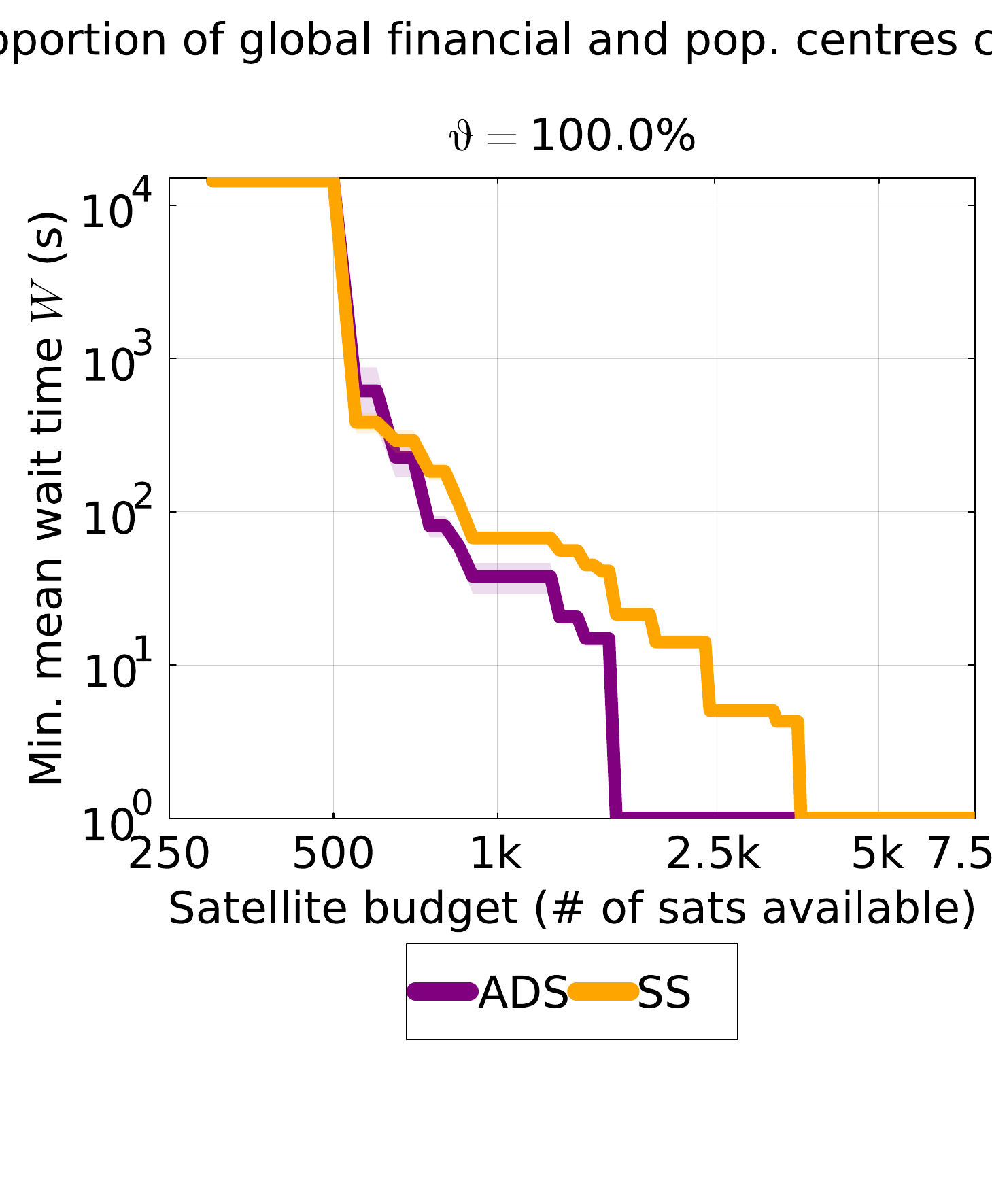}
                \put(8, 75){\textbf{(a)}}
        \put(30, -5){\small  Satellite Budget }
        \end{overpic}
        \label{fig:wait_vs_sat_city}
    \end{subfigure}
    \begin{subfigure}{0.3\linewidth}
        \begin{overpic}[height=4.3cm, trim=2cm 8cm 0cm 1.9cm, clip]{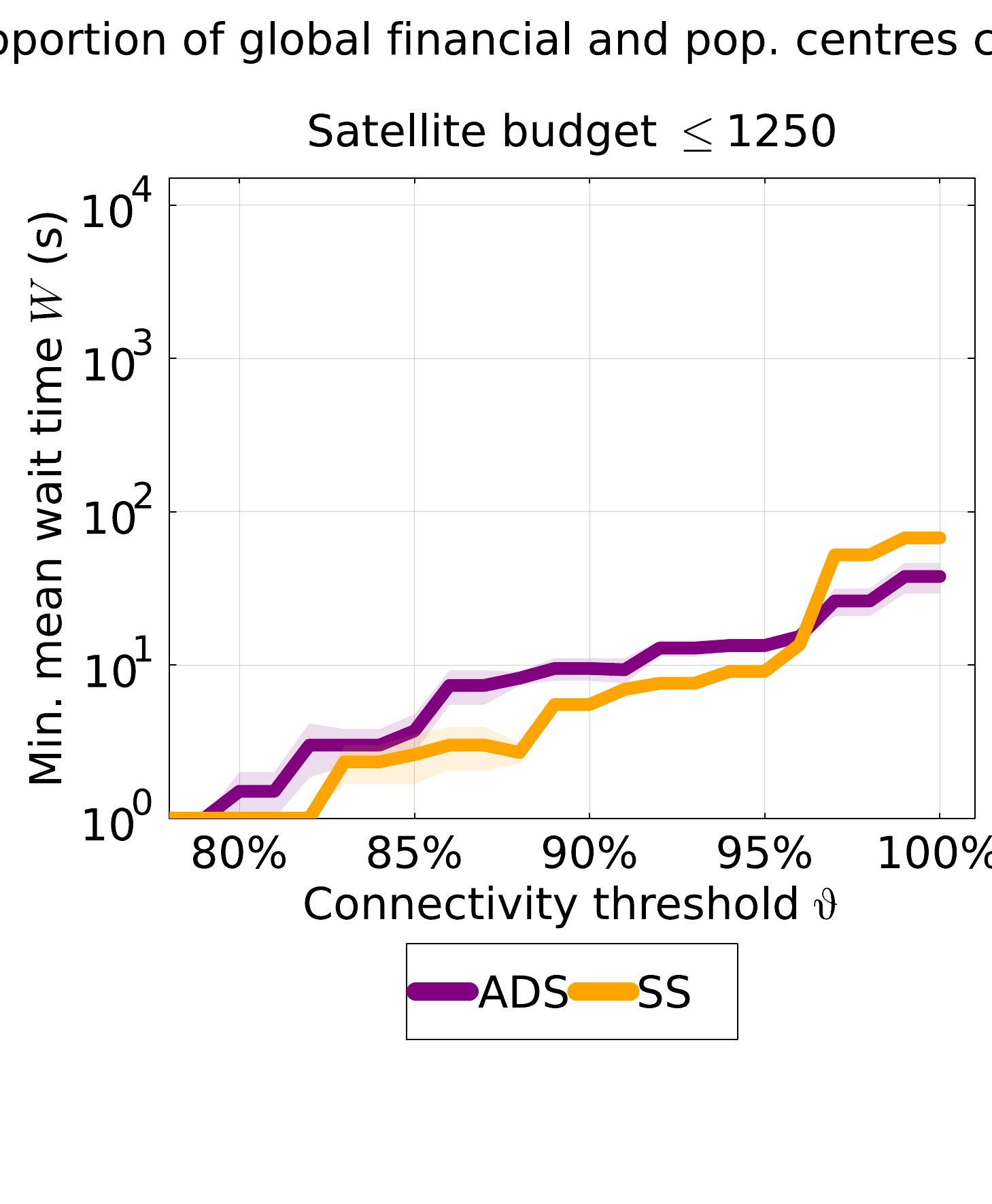}
        \put(50, -5){\small Connectivity Threshold $\vartheta$}
                        \put(5, 82){\textbf{(b)}}
        \end{overpic}
    \end{subfigure}
        \begin{subfigure}{0.3\linewidth}
            \begin{overpic}[height=4.3cm, trim=2cm 8cm 0cm 1.9cm, clip]{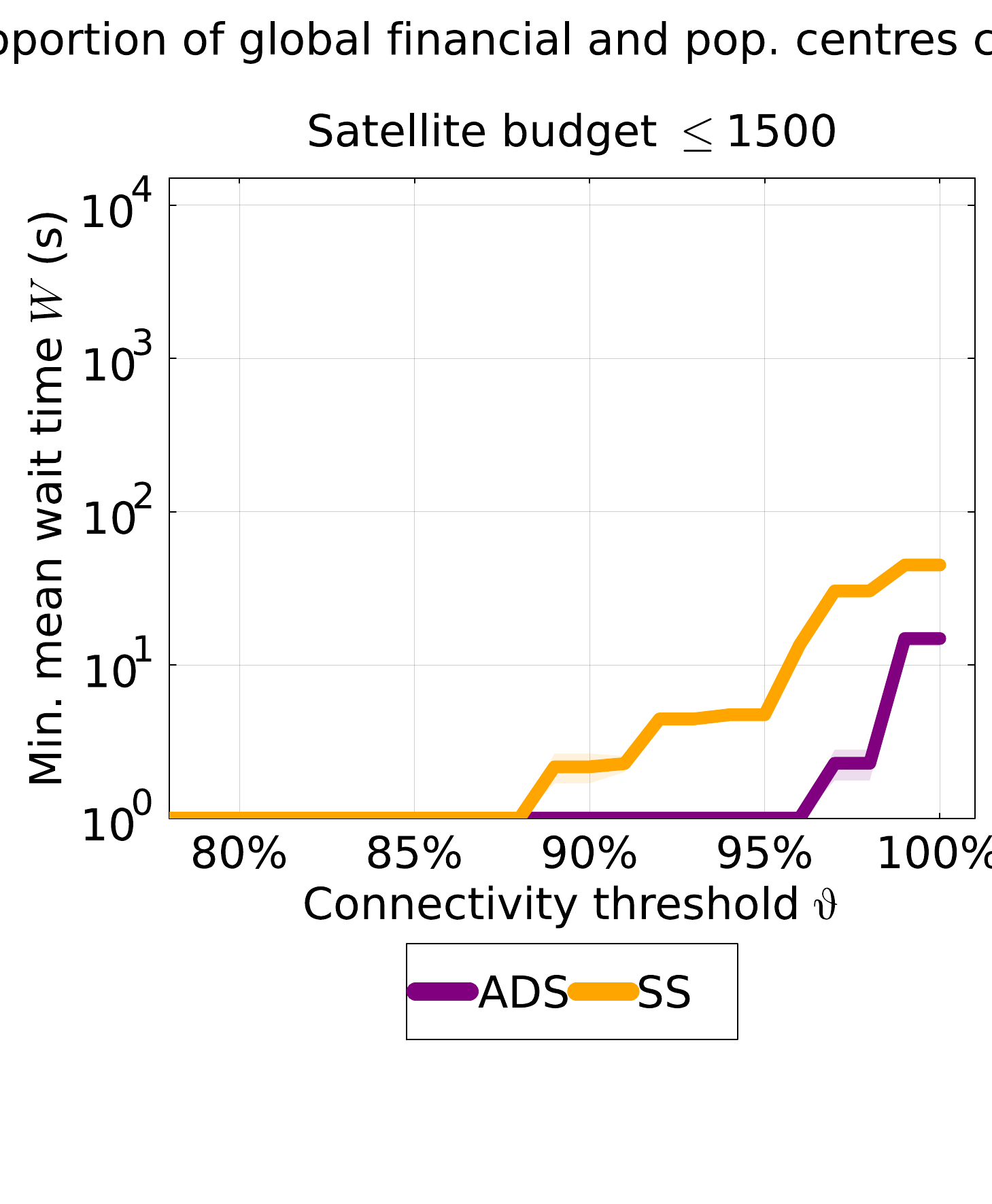}
                            \put(5, 82){\textbf{(c)}}
                            \end{overpic}
        \end{subfigure}
    \\
            \begin{subfigure}{0.95\linewidth}
            \vspace{1em}
    \centering
                \includegraphics[height = 1cm, trim = 1.9cm 0.7cm 0.3cm 26cm,
    clip]{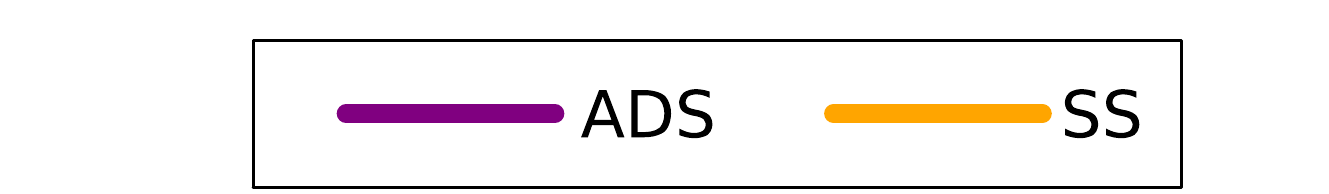}
    \end{subfigure}
   \caption{
    \textbf{Minimum mean waiting time required to satisfy concurrent traffic-matrix connectivity targets for augmented dual-shell (ADS) and single-shell (SS) constellations.}
    Purple curves denote ADS constellations and orange curves denote SS constellations.
    All plotted configurations are restricted to satellite service concurrency $T=7$, ensuring that the comparison isolates constellation architecture rather than per-satellite service capacity.
    (a) Minimum mean wait time required to achieve complete concurrent connectivity, $\vartheta=100\%$, as a function of the available satellite budget.
    (b--c) Minimum mean wait time as a function of the required concurrent-connectivity threshold $\vartheta$, conditioned on total satellite budgets of 1250 and 1500 satellites, respectively.
    At lower satellite budgets, ADS and SS exhibit similar latency behavior because the primary (inclined) orbital shell remains satellite-limited.
    Once the main shell is sufficiently populated, the additional polar component in ADS improves coverage diversity and reduces synchronized visibility gaps, producing a sharp transition between the 1250- and 1500-satellite regimes.
    Beyond this transition, ADS provides substantially faster connectivity at stringent thresholds, with many lower threshold targets becoming available essentially instantaneously as the satellite budget increases.
    }
\label{fig:wait_vs_sats_ads_ss}
    \end{figure}

Single-shell constellations in the configuration studied here also under-utilize high-latitude and polar ground stations. Satellites confined to the $53^\circ$ inclination shell concentrate their ground tracks at mid-latitudes, resulting in intermittent and poorly timed visibility at the highest latitudes. This limitation is particularly consequential given the geometry of the Earth: high-latitude access can reduce great-circle distances for some East--West intercontinental routes and provides additional routing diversity. By excluding polar or near-polar shells, single-shell designs forgo both the excess satellite visibility available at high latitudes and the favorable geographic position of polar ground stations.

The dual-shell constellation mitigates this limitation by allocating satellites to shells with complementary inclinations, enabling persistent access to high-latitude regions while maintaining coverage at mid-latitudes. Crucially, this improvement arises from geometric diversity rather than increased resources: the same satellite budget yields greater temporal overlap between independently generated links, directly translating into reduced wait times at fixed connectivity thresholds.

Figure~\ref{fig:strength_vs_sats_city} reports the average strength of active entanglement links as a function of satellite budget. It is evident that the dual-shell constellation achieves link strengths comparable to or higher than those of the single-shell design for the same satellite budget. This in addition to the smaller latency for concurrent connectivity. It can also be observed that beyond a satellite budget of 2500 satellites the link strength for the SS design saturates whereas the ADS design can further improve by servicing additional paths at higher altitudes. This further solidifies the intuition that the ADS design better utilizes satellite and ground station resources.

\begin{figure}[!htbp]
    \centering
        \begin{subfigure}{0.33\linewidth}
            \centering
        {\hspace{2em}Satellite Budget $\leq 1250$\par}
        \includegraphics[height = 4.5cm, trim = 0 7.1cm 0.3cm 3.3cm,
    clip]{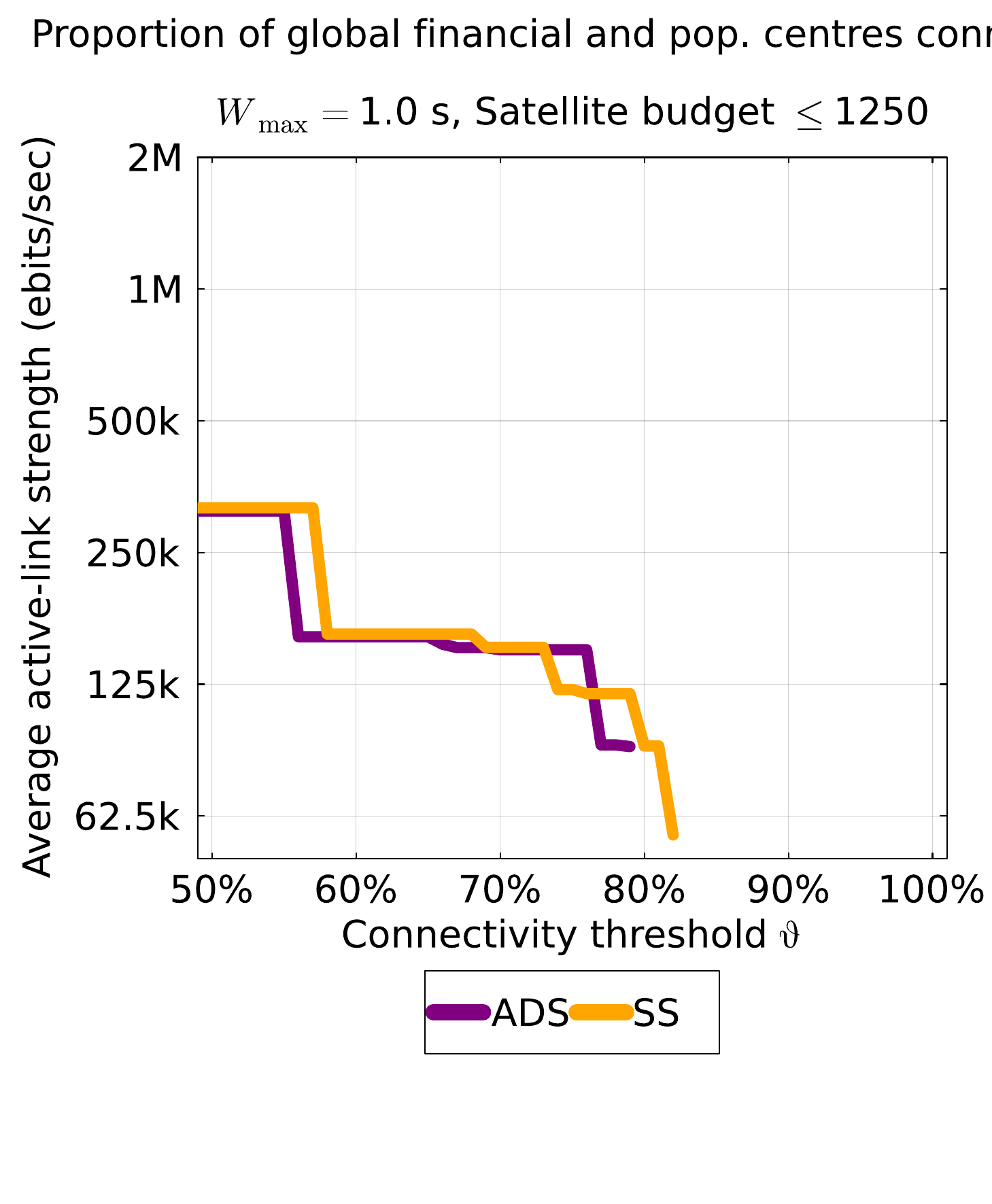}
    \end{subfigure}
    % \hfill
    \begin{subfigure}{0.30\linewidth}
        \centering
        {\hspace{1.5em}Satellite Budget $\leq 2500$\par}
        \includegraphics[height = 4.5cm, trim = 1.9cm 7.1cm 0.3cm 3.3cm,
    clip]{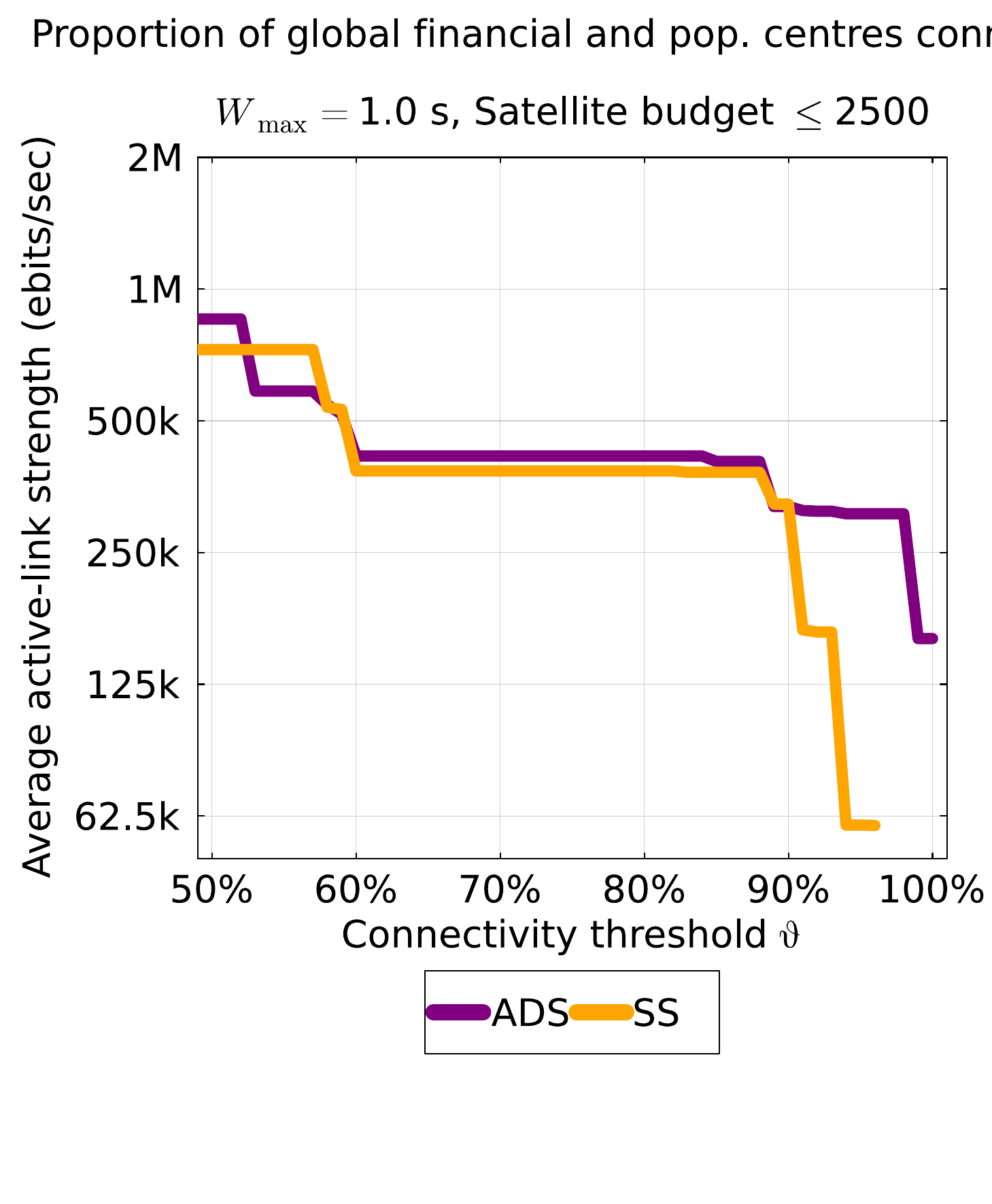}
    \end{subfigure}
    % \hfill
        \begin{subfigure}{0.3\linewidth}
        {\hspace{3em} Satellite Budget $\leq 5000$\par}
        \includegraphics[height = 4.5cm, trim = 1.9cm 7.1cm 0.3cm 3.3cm,
    clip]{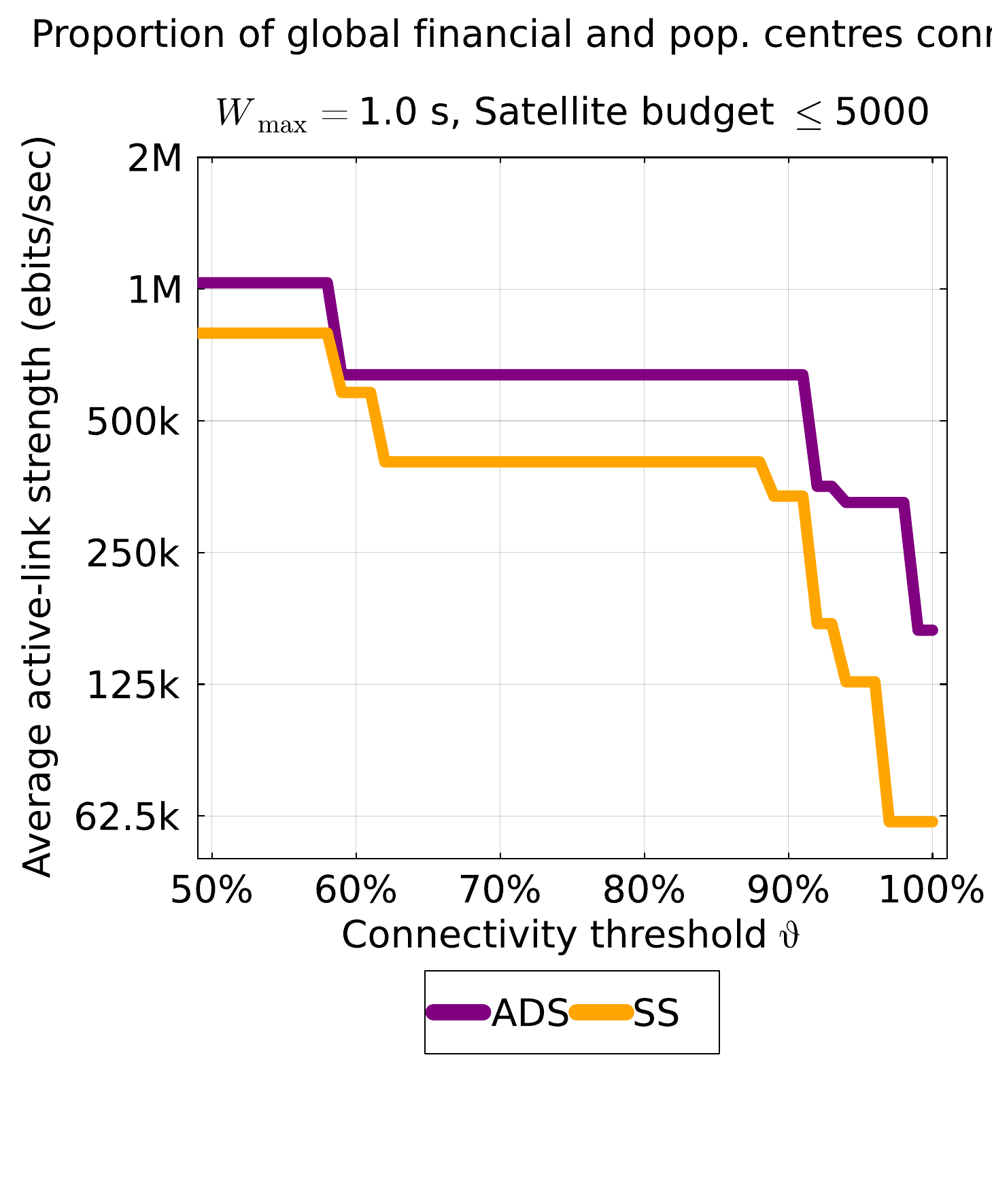}
    \end{subfigure}
    \begin{subfigure}{0.95\linewidth}
    \centering
                {Connectivity Threshold $\vartheta$ \par}
    \end{subfigure}
        \begin{subfigure}{0.95\linewidth}
    \centering
                \includegraphics[height = 1cm, trim = 1.9cm 0.7cm 0.3cm 26cm,
    clip]{figures/SS_ADS_legend.pdf}
    \end{subfigure}
    \caption{
    \textbf{Average \emph{latency-conditioned} active-link strength for augmented dual-shell (ADS) and single-shell (SS) constellation designs, evaluated at different concurrent traffic-matrix connectivity thresholds $\vartheta$.}
    The three panels condition the design search on total satellite budgets of 1250, 2500, and 5000 satellites, respectively.
    Purple curves denote ADS constellations and orange curves denote SS constellations.
    Only configurations with satellite service concurrency $T=7$ are included, ensuring that the comparison isolates constellation architecture rather than differences in per-satellite ground-station service capacity.
    For each threshold $\vartheta$, the plotted value is the average active-link strength among configurations that achieve the required fraction of simultaneously connected traffic-matrix pairs within the allowed waiting window.
    In the 1250-satellite regime, ADS and SS provide comparable active-link strength over the limited range of achievable concurrent-connectivity thresholds, indicating that the constellation remains primarily satellite-limited.
    As the satellite budget increases, ADS increasingly outperforms SS, especially at stringent thresholds, where the SS active-link strength drops sharply while ADS maintains a stronger concurrently connected backbone.
    Together with the preceding wait-time results, this indicates a transition in the role of the auxiliary ADS shell: once the primary shell is sufficiently populated to satisfy the concurrent-connectivity requirement, the additional ADS coverage diversity translates more effectively into higher end-to-end entanglement capacity rather than primarily reducing latency.
    }
    \label{fig:strength_vs_sats_city}
\end{figure}

In the above analysis we fix $W_{\max} = 1$ second which is a stringent coherence time requirement for ground station quantum memories. This is done considering near-term feasibility of our design and analysis. In Supplementary Material~\cite{supplementary_materials}, we also report results for $W_{\max} \in \{10~\rm s, \; 60~\rm s, \; 1~\rm hour, \; 4$~\rm hours\} cases, which represents a forward-looking regime in which long-lived quantum memories with coherence times on the order of hours or days are available. In this regime, the architectural advantages of the dual-shell constellation become more pronounced.

Overall, these results demonstrate that constellation structure is a first-order architectural parameter for satellite-serviced quantum backbones. For the shell pair evaluated here (a $53^\circ$ primary shell augmented by a $98^\circ$ near-polar shell), redistributing a fixed satellite budget across the two inclination shells reduces connectivity latency and improves latency-conditioned link availability at moderate to high thresholds, relative to the single-shell $53^\circ$ baseline. This suggests that inclination diversity is an effective lever for increasing global concurrency in this LEO design space.

\subsection{Design Insight 3: Satellites Should Support Multi-Party Connectivity}
\label{subsec:design_insight3}

We next examine how per-satellite servicing capability impacts backbone performance. Specifically, we compare two operational modes under identical ground-station layouts, orbital parameters, and physical-layer assumptions, but differing in how many ground-station pairs a satellite can support concurrently.

In the bi-partite connectivity (BPC) mode, each satellite is equipped with two optical terminals and can establish entanglement with at most one ground-station pair per time shard, corresponding to conventional point-to-point operation. In contrast, multi-party connectivity (MPC)\footnote{The term multi-party connectivity refers to simultaneous service of multiple ground stations, not to multipartite entanglement: all generated quantum links remain bipartite Bell pairs.} activates multiple links within a single footprint under the constrained hub--spoke--ring rule, allowing satellites equipped with $T>2$ terminals to serve multiple ground-station pairs concurrently. In our implementation, each satellite selects a single hub ground station and up to six nearest neighboring stations within its visibility footprint, with entanglement distribution restricted to nearest-neighbor connections among these stations. The parameter $T$ denotes the number of optical terminals per satellite and determines how many ground-station links can be activated concurrently within a footprint. Importantly, MPC does not alter orbital coverage or satellite trajectories; it increases only the number of concurrent bipartite links that can be supported by a satellite at a given time. As shown in Figure~\ref{fig:wait_vs_T}, increasing $T$ under a fixed total terminal budget sharply reduces the wait time needed to achieve full traffic-matrix connectivity, especially for anisotropic ground-station layouts.

\begin{figure}[!htbp]
\begin{subfigure}{0.95\linewidth}
    \centering
    \includegraphics[height = 6cm, trim = 0 5cm 0 2.6cm, clip]{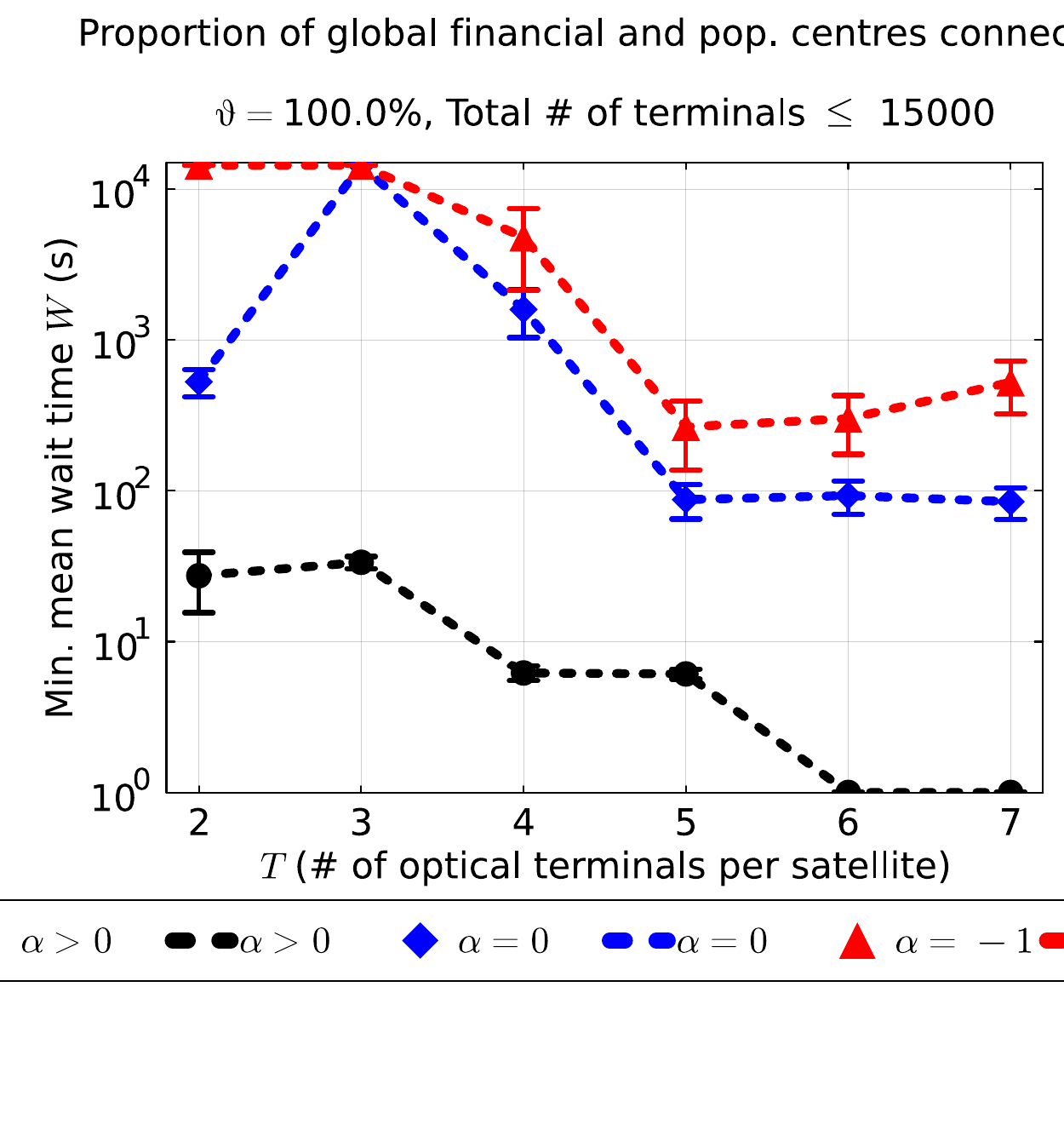}
    \end{subfigure}
                \begin{subfigure}
                {0.95\linewidth}
                \centering
        \includegraphics[width = 0.6\linewidth, trim = 6cm 0 7cm 10.8cm, clip]{figures/pf_0.1_k_7_wait_vs_sats_connectivity_panels_continousthreshold.pdf}
    \end{subfigure}
        \caption{
        \textbf{Impact of per-satellite GS service capability on global connectivity latency.}
        Minimum mean wait time $W$ to achieve full connectivity ($\vartheta = 100\%$) as a function of the number of optical terminals per satellite $T$, evaluated under a fixed total terminal budget ($N_\mathrm{sat} \times T \leq 15000)$.
        Curves correspond to different ground-station geometries: anisotropic ($\alpha > 0$, black), uniform ($\alpha = 0$, blue), and longitudinal ($\alpha = -1$, red). Error bars indicate $\pm 1$ standard error.
        Increasing $T$ generally reduces connectivity latency by enabling simultaneous service to multiple ground stations.
        The reduction is most pronounced for anisotropic layouts, which achieve near-instantaneous connectivity at moderate $T$, while uniform and longitudinal grids require higher $T$ and remain substantially slower.
        }
        \label{fig:wait_vs_T}
        \end{figure}

For the uniform ($\alpha=0$) and longitudinal ($\alpha=-1$) layouts, the improvement saturates beyond $T\simeq 5$, and the wait time remains far above the near-instantaneous regime even when additional terminals provide greater opportunity for concurrent edge activation. The wait time can also become mildly non-monotonic with $T$ in these layouts, suggesting that simply increasing the number of simultaneously serviceable links does not overcome the less favorable spatial structure of the ground-station graph. In contrast, the anisotropic layouts ($\alpha>0$) exploit the additional per-satellite connectivity much more effectively: the wait time drops to the minimum observable value for $T\geq 6$. This behavior is consistent with the interpretation that anisotropic ground-station placement has more favorable finite-size percolation properties, so that increasing the effective graph degree more readily drives the traffic matrix into a connected regime~\cite{percolation_review_bethe}.

\paragraph{Fixed satellite budget comparison.}
Figure~\ref{fig:wait_vs_sats_city_multitelescope} compares connectivity latency under bi-partite connectivity (BPC) and multi-party connectivity (MPC) as a function of satellite budget for anisotropic ground-station layouts ($\alpha > 0$). Supplementary Material~\cite{supplementary_materials} also reports similar comparison for uniform and longitudinal grids. For a fixed number of satellites, MPC reduces the wait time required to achieve a given traffic-matrix connectivity threshold drastically. This implies that in scenarios where per-satellite costs are the primary driver of resource constraints, MPC outperforms BPC by several orders of magnitude.
The gap between MPC and BPC widens as the required fraction of simultaneously connected cities increases, consistent with a per-satellite concurrency bottleneck in the bi-partite mode.

The difference stems from how each policy converts instantaneous visibility into edges. Under BPC operation, even when a satellite is simultaneously visible to many ground stations, it activates at most one ground-station pair per time shard; additional visible stations therefore do not translate into additional edges at that instant. MPC relaxes this constraint by activating multiple links within a single footprint using the constrained hub--spoke--ring service rule. As a consequence, MPC accelerates large-component formation, with the advantage becoming most pronounced in the high-threshold regime where many concurrent edges are needed to connect most traffic-matrix nodes. Increasing satellite count under BPC continues to reduce wait time, but requires substantially larger budgets in the high-threshold regime because each satellite can activate only one ground-station pair per time shard. As satellite density increases, additional satellites increasingly generate redundant edges (e.g., on pairs already served within the same shard, for which only the strongest link is retained). 

We sweep the full range of feasible servicing capabilities $T \in \{3,4,\dots,7\}$ and report the envelope of these configurations when referring to MPC performance.

\begin{figure}[!htbp]
    \centering
    \begin{subfigure}[t]{0.24\linewidth}
        \begin{overpic}[height=3.9cm,trim=0cm 6.9cm 0cm 1.6cm,clip]
        {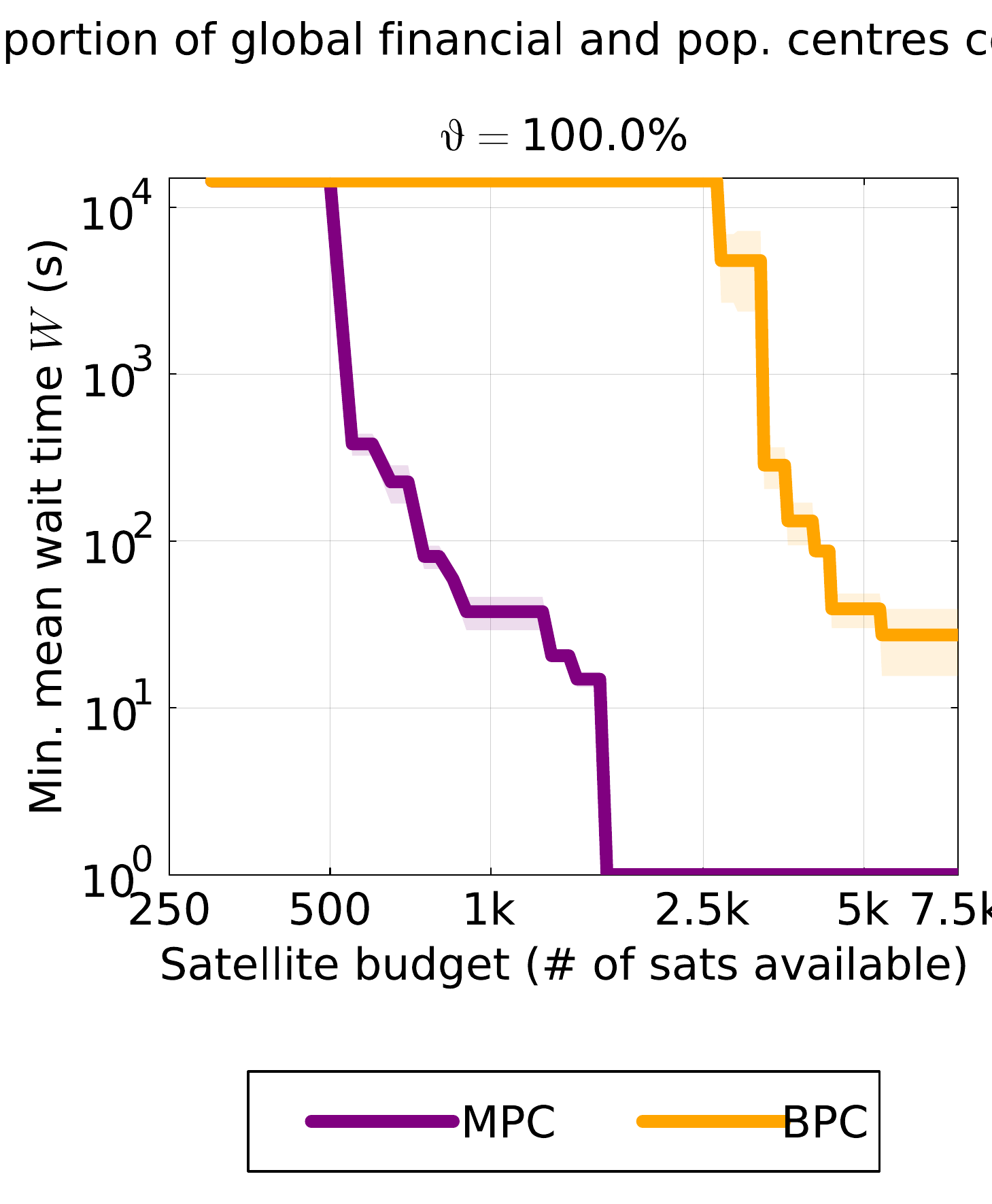}
            \put(5,80){\small\textbf{(a)}}
            \put(25, -7){\small Satellite Budget }
        \end{overpic}
    \end{subfigure}
    \hfill
    \begin{subfigure}[t]{0.72\linewidth}
        \begin{subfigure}[t]{0.31\linewidth}
            \begin{overpic}[height=3.9cm,trim=2cm 6.9cm 0cm 1.6cm,clip]
            {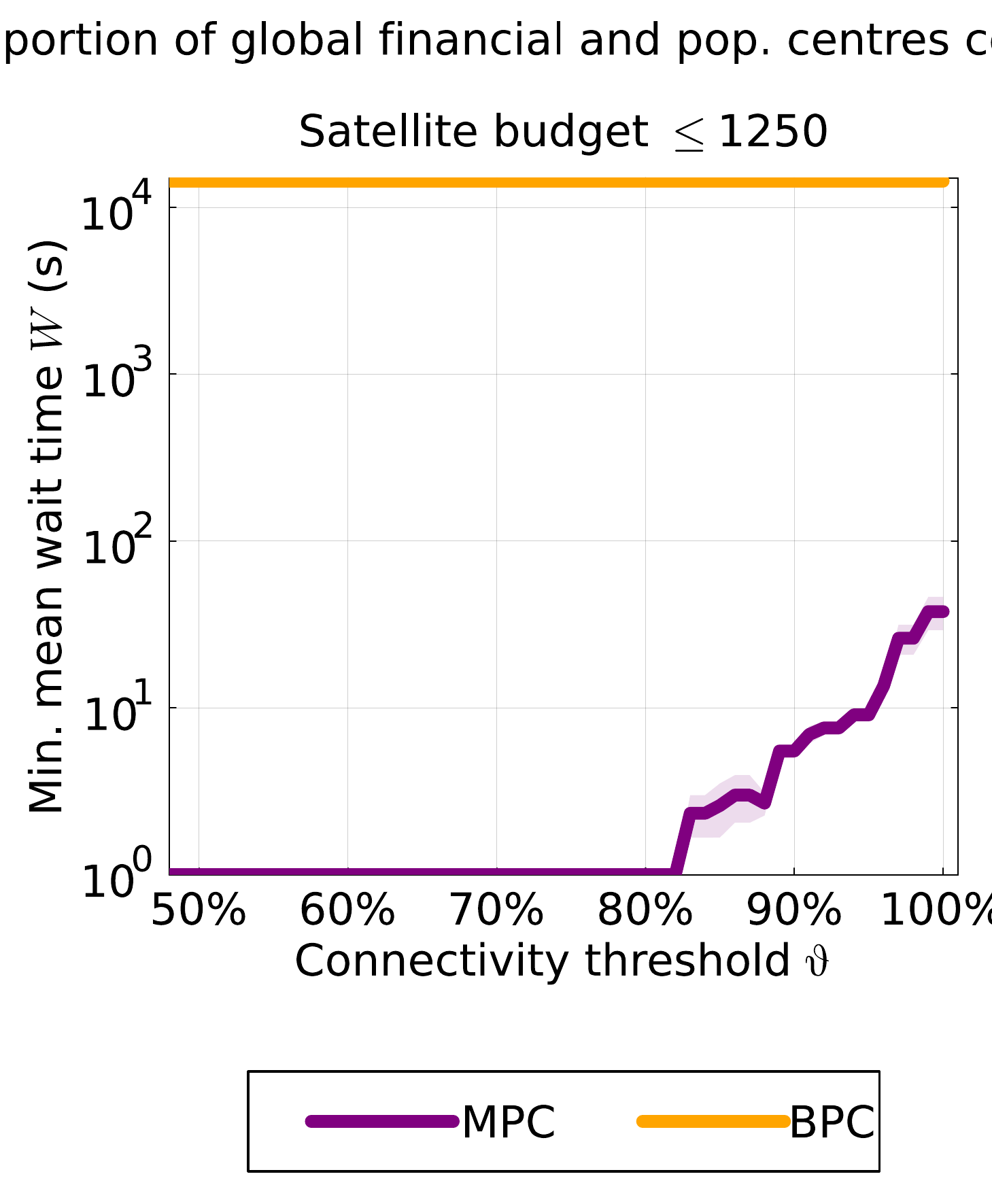}
                \put(5,85){\small\textbf{(b)}}
            \end{overpic}
        \end{subfigure}
        \hfill
        \begin{subfigure}[t]{0.31\linewidth}
            \begin{overpic}[height=3.9cm,trim=2cm 6.9cm 0cm 1.6cm,clip]
            {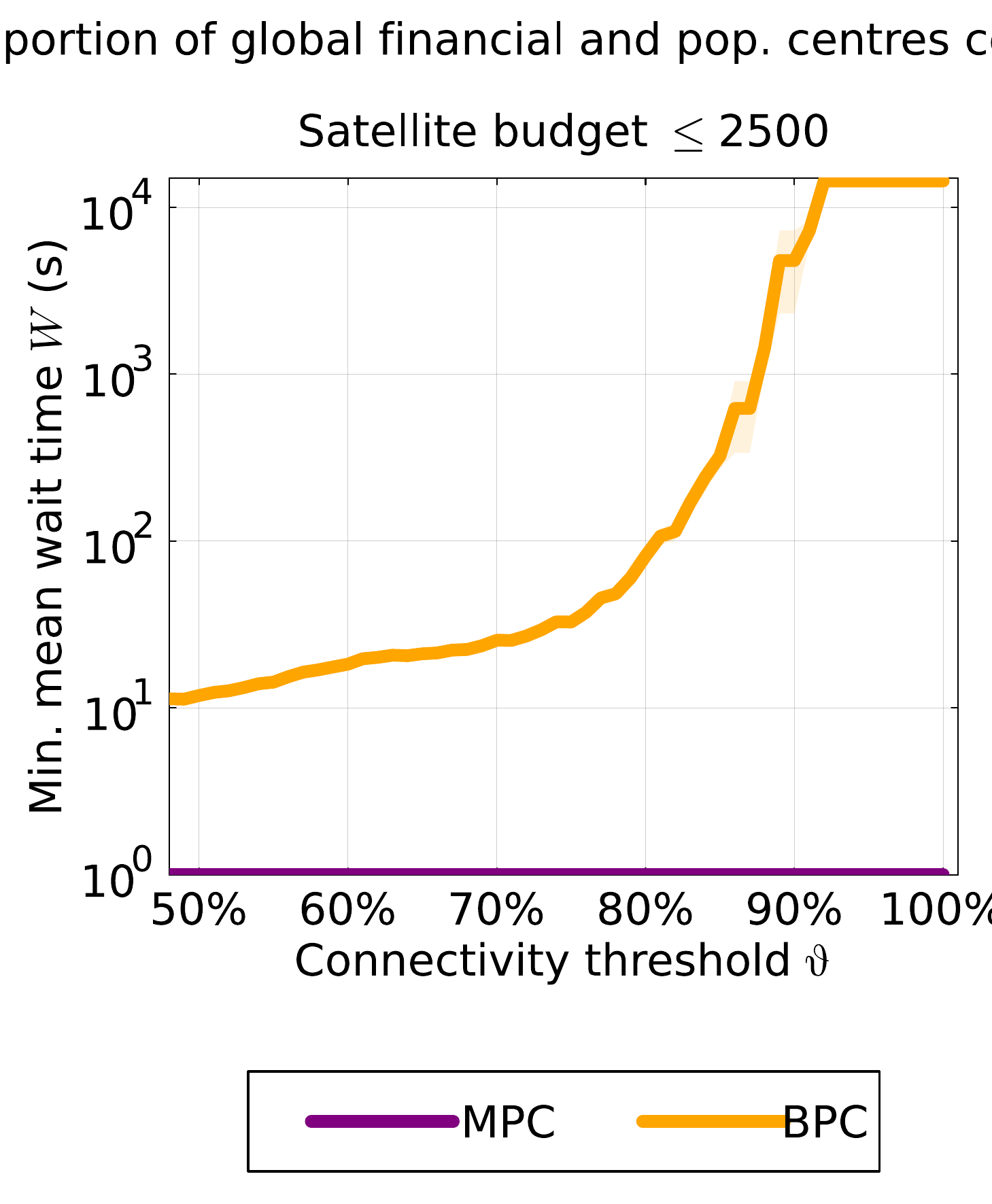}
                                \put(5,85){\small\textbf{(c)}}
                \put(15, -7){\small Connectivity Threshold $\vartheta$ }
            \end{overpic}
        \end{subfigure}
        \hfill
        \begin{subfigure}[t]{0.31\linewidth}
            \begin{overpic}[height=3.9cm,trim=2cm 6.9cm 0cm 1.6cm,clip]
            {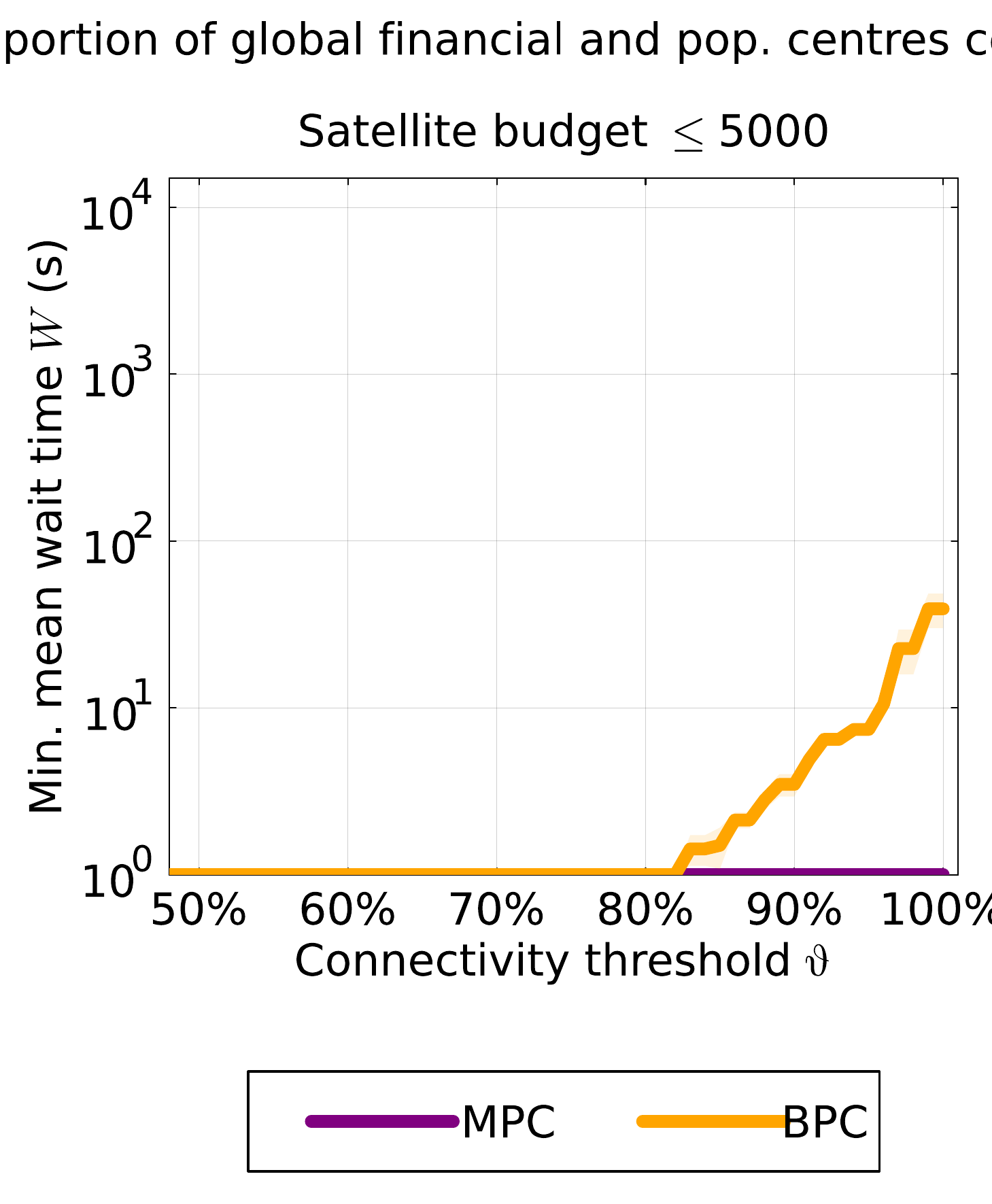}
                \put(5,85){\small\textbf{(d)}}
            \end{overpic}
        \end{subfigure}
        \vspace{2em}
    \end{subfigure}

    \begin{subfigure}{0.95\linewidth}
        \centering
        \includegraphics[height = 0.8cm, trim = 1.9cm 0.5cm 0.3cm 27cm, clip]
        {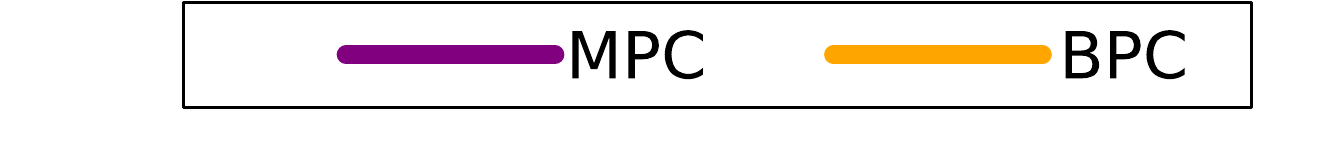}
    \end{subfigure}

   \caption{
    \textbf{Mean wait time to achieve global connectivity under bi-partite (BPC) and multi-party (MPC) satellite operation.}
    (a) Minimum mean wait time $W$ required to achieve full connectivity ($\vartheta = 100\%$) as a function of satellite budget.
    (b)--(d) Mean wait time as a function of the connectivity threshold $\vartheta$ for fixed satellite budgets of 1250, 2500, and 5000 satellites, respectively.
    Curves compare bi-partite connectivity (BPC, $T=2$) and multi-party connectivity (MPC, $T \in \{3,\dots,7\}$), where MPC enables each satellite to simultaneously serve a local hub--spoke--ring neighborhood.
    Across all regimes, MPC consistently achieves lower wait times and sustains connectivity at more stringent thresholds.
    In contrast, BPC fails to achieve high connectivity within the allowed time window at low satellite budgets, exhibiting effectively infinite wait times at large $\vartheta$.
    MPC reaches near-instantaneous connectivity at substantially lower satellite budgets, whereas BPC requires significantly larger constellations to achieve $\vartheta \gtrsim 80\%$.
    }
\label{fig:wait_vs_sats_city_multitelescope}
   
\end{figure}

\paragraph{Terminal-budget-normalized comparison.}
Because MPC corresponds to a family of servicing capabilities with $T \in \{3,\dots,7\}$ terminals per satellite, satellite-count comparisons alone conflate service policy with payload capability. We therefore also normalize results by the total number of optical terminals in orbit, $N_{\mathrm{term}} = N_{\mathrm{sat}} \times T$, where $T$ denotes terminals per satellite. This normalization provides a first-order hardware-budget proxy, while acknowledging that terminal SWaP and pointing complexity need not scale linearly with $T$.

Figure~\ref{fig:wait_vs_terminals_city_multitelescope} shows that MPC maintains a latency advantage at high connectivity thresholds even when plotted against total terminal budget, indicating that the gains are not explained solely by increased terminal count but also by increased degree of the graph formed on the ground station grid. Figure~\ref{fig:strength_vs_terminals_city_multitelescope} shows that the latency advantage of MPC is not obtained by sacrificing active-link strength. Under fixed terminal budgets, MPC sustains higher average active-link strength across all connectivity thresholds, while BPC degrades rapidly as the required connectivity threshold $\vartheta$ becomes more stringent. This behavior indicates that BPC is limited not only by slower edge accumulation, but also by its inability to maintain strong concurrently active links when many traffic-matrix pairs must be connected. MPC converts the same terminal budget into a denser set of simultaneous footprint-level edges, allowing it to preserve link strength while also reaching high-threshold connectivity more rapidly.

Overall, per-satellite servicing capability emerges as a first-order architectural parameter for satellite-serviced quantum backbones. Bi-partite operation can be adequate for low-threshold connectivity or sparse demand, but it imposes a hard concurrency limit that becomes dominant when near-global connectivity is required. Multi-party connectivity relaxes this limitation by converting simultaneous visibility into simultaneous links, enabling high-threshold traffic-matrix connectivity at lower latency under both satellite-budget and terminal-budget normalizations, albeit with a measurable link-strength trade-off in resource-limited regimes.

\begin{figure}[!htbp]
    \centering
    \begin{subfigure}[t]{0.24\linewidth}
    \phantomcaption

        \begin{overpic}[height=3.9cm,trim=0cm 7.8cm 0cm 1.6cm,clip]
        {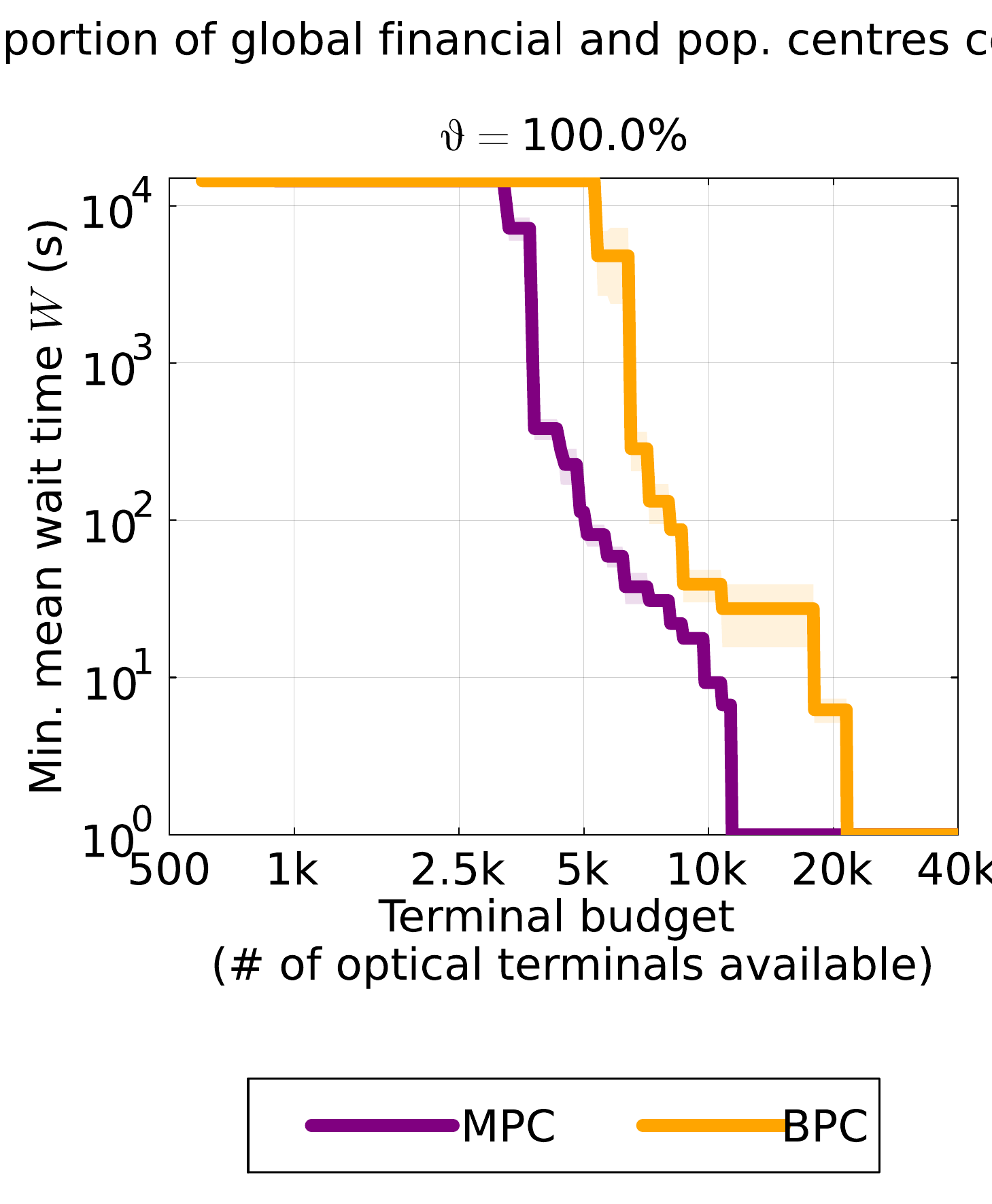}
            \put(5,75){\small\textbf{(a)}}
            \put(20, -7){\small Max. \# of Terminals }
            \put(30, -15){\small ($N_\mathrm{{sat}} \times T$)}
        \end{overpic}

    \end{subfigure}
    \hfill
    \begin{subfigure}[t]{0.72\linewidth}
        \begin{subfigure}[t]{0.31\linewidth}
            \begin{overpic}[height=3.9cm,trim=2cm 6.9cm 0cm 1.6cm,clip]
            {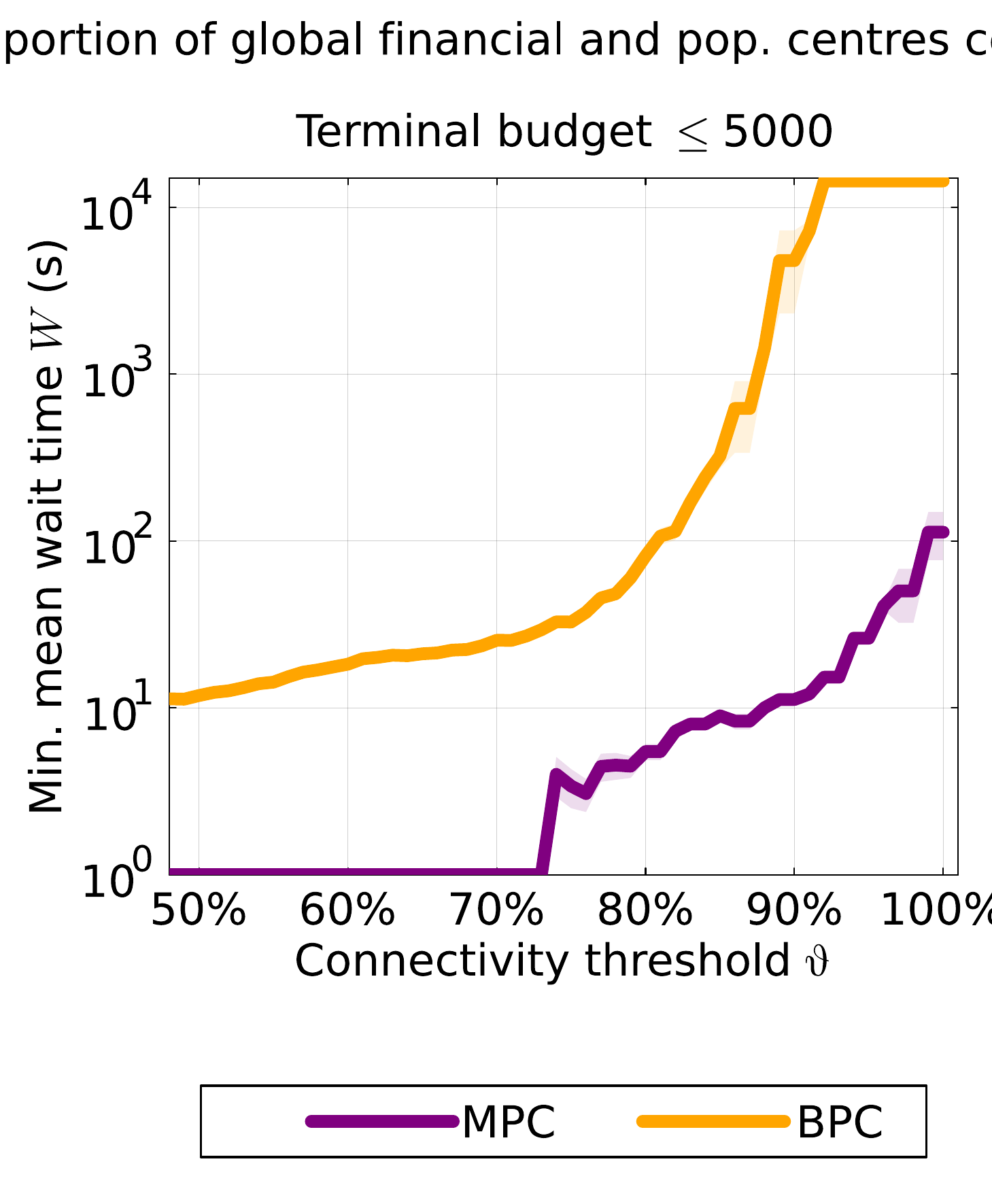}
                \put(5,85){\small\textbf{(b)}}
            \end{overpic}
        \end{subfigure}
        \hfill
        \begin{subfigure}[t]{0.31\linewidth}
            \begin{overpic}[height=3.9cm,trim=2cm 6.9cm 0cm 1.6cm,clip]
            {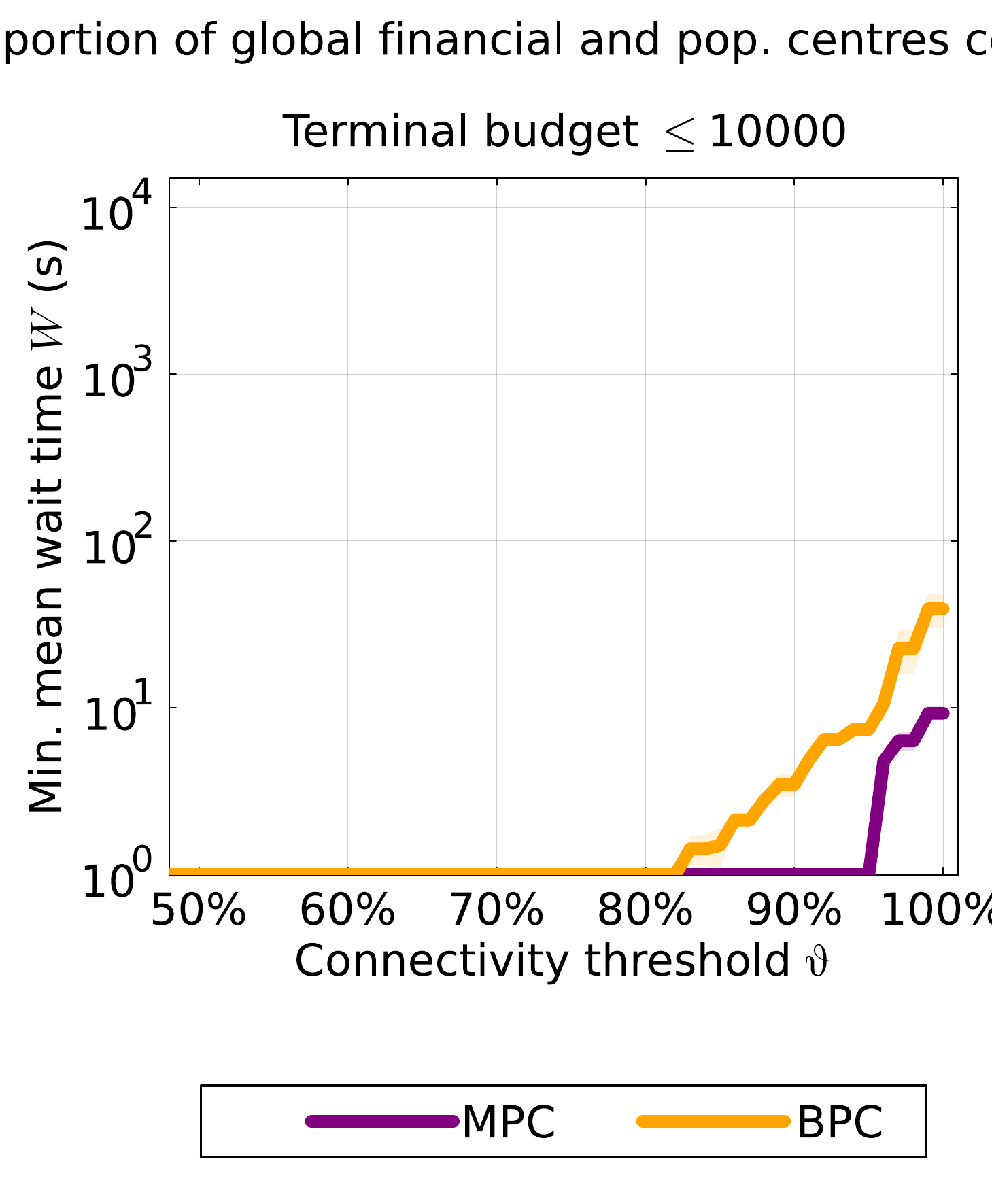}
                \put(5,85){\small\textbf{(c)}}
                \put(15, -7){\small Connectivity Threshold $\vartheta$ }
            \end{overpic}
        \end{subfigure}
        \hfill
        \begin{subfigure}[t]{0.31\linewidth}
            \begin{overpic}[height=3.9cm,trim=2cm 6.9cm 0cm 1.6cm,clip]
            {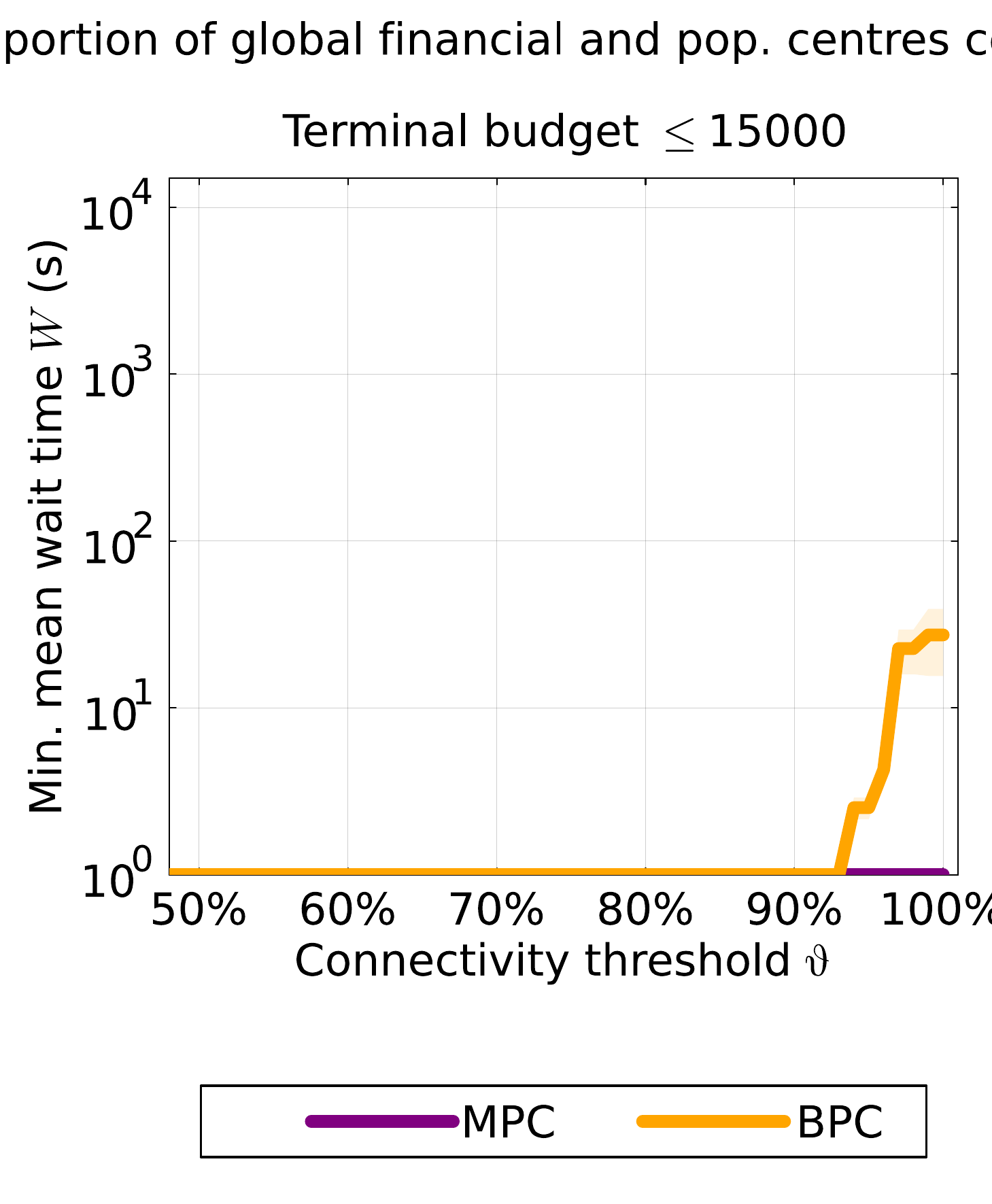}
                \put(5,85){\small\textbf{(d)}}
            \end{overpic}
        \end{subfigure}
        \vspace{2em}
    \end{subfigure}

    \begin{subfigure}{0.95\linewidth}
        \centering
        \includegraphics[height = 0.8cm, trim = 1.9cm 0.5cm 0.3cm 27cm, clip]
        {figures/BPC_MPC_legend.pdf}
    \end{subfigure}
    \caption{
        \textbf{Mean wait time to achieve global connectivity at fixed total terminal budget.}
        (a) Minimum mean wait time $W$ required to achieve full connectivity ($\vartheta = 100\%$) as a function of the total number of optical terminals $N_{\mathrm{term}} = N_{\mathrm{sat}}\times T$.
        (b)--(d) Mean wait time as a function of the connectivity threshold $\vartheta$ for fixed terminal budgets of 5000, 10000, and 15000 terminals, respectively.
        Curves compare multi-party connectivity (MPC, purple) and bi-partite connectivity (BPC, orange). Shaded bands indicate $\pm 1$ standard error.
        Under equal total terminal budgets, MPC achieves substantially lower wait times and remains feasible at high connectivity thresholds.
        In contrast, BPC exhibits rapidly increasing---and often effectively unbounded---wait times at stringent $\vartheta$, requiring significantly larger number of terminals to approach full connectivity.
        }
       \label{fig:wait_vs_terminals_city_multitelescope}
\end{figure}

\begin{figure}[!htb]
    \centering
    \begin{subfigure}{0.33\linewidth}
        \centering
        \begin{overpic}[height=4cm,trim=0 5.8cm 0.3cm 3.3cm,clip]
        {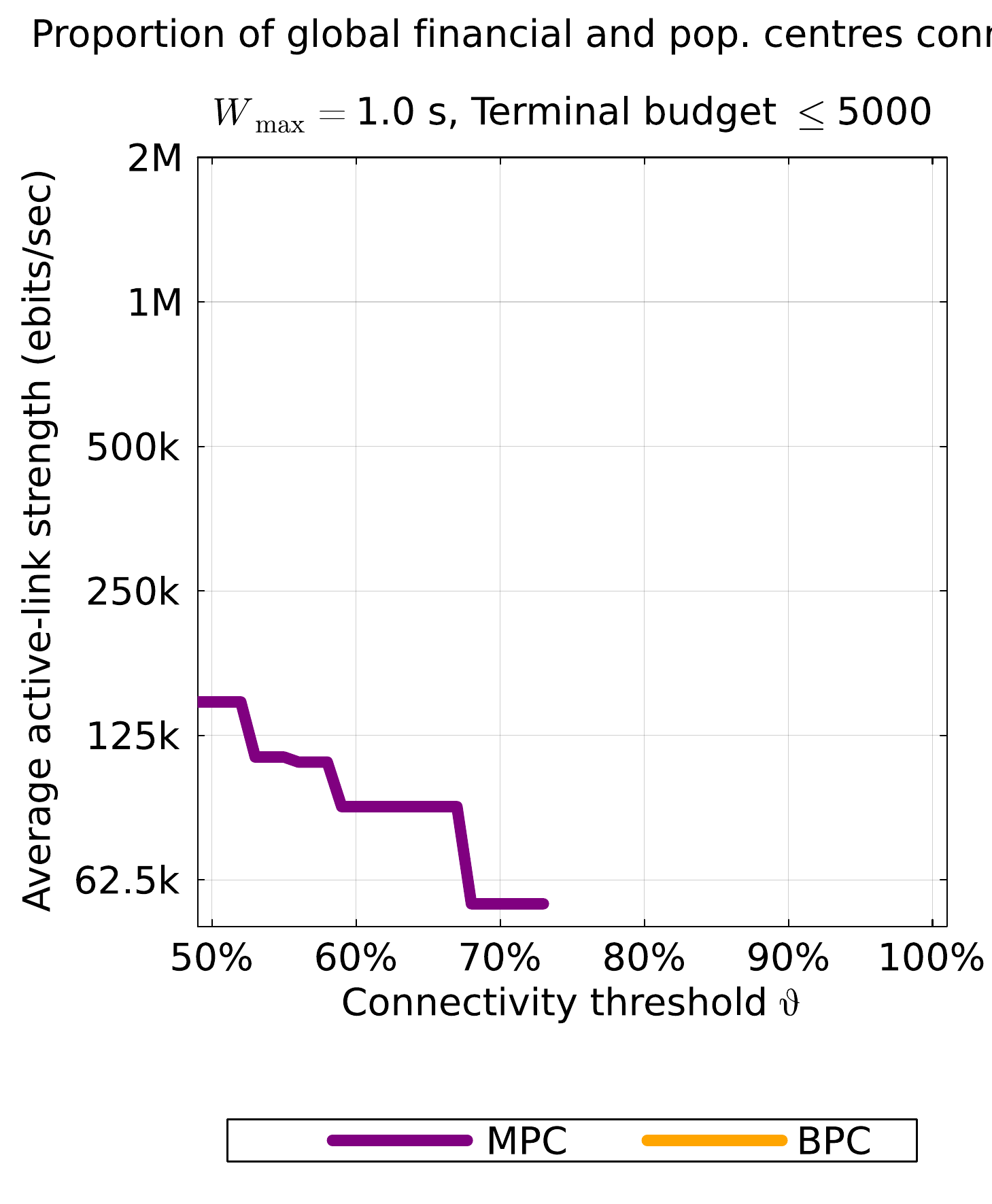}
            \put(10,88){\small\textbf{(a)}}
            \put(23,88){\small\centering Terminal Budget $\leq 5000$}
        \end{overpic}
    \end{subfigure}
    % \hfill
    \begin{subfigure}{0.30\linewidth}
        \centering
        \begin{overpic}[height=4cm,trim=1.9cm 5.8cm 0.3cm 3.3cm,clip]
        {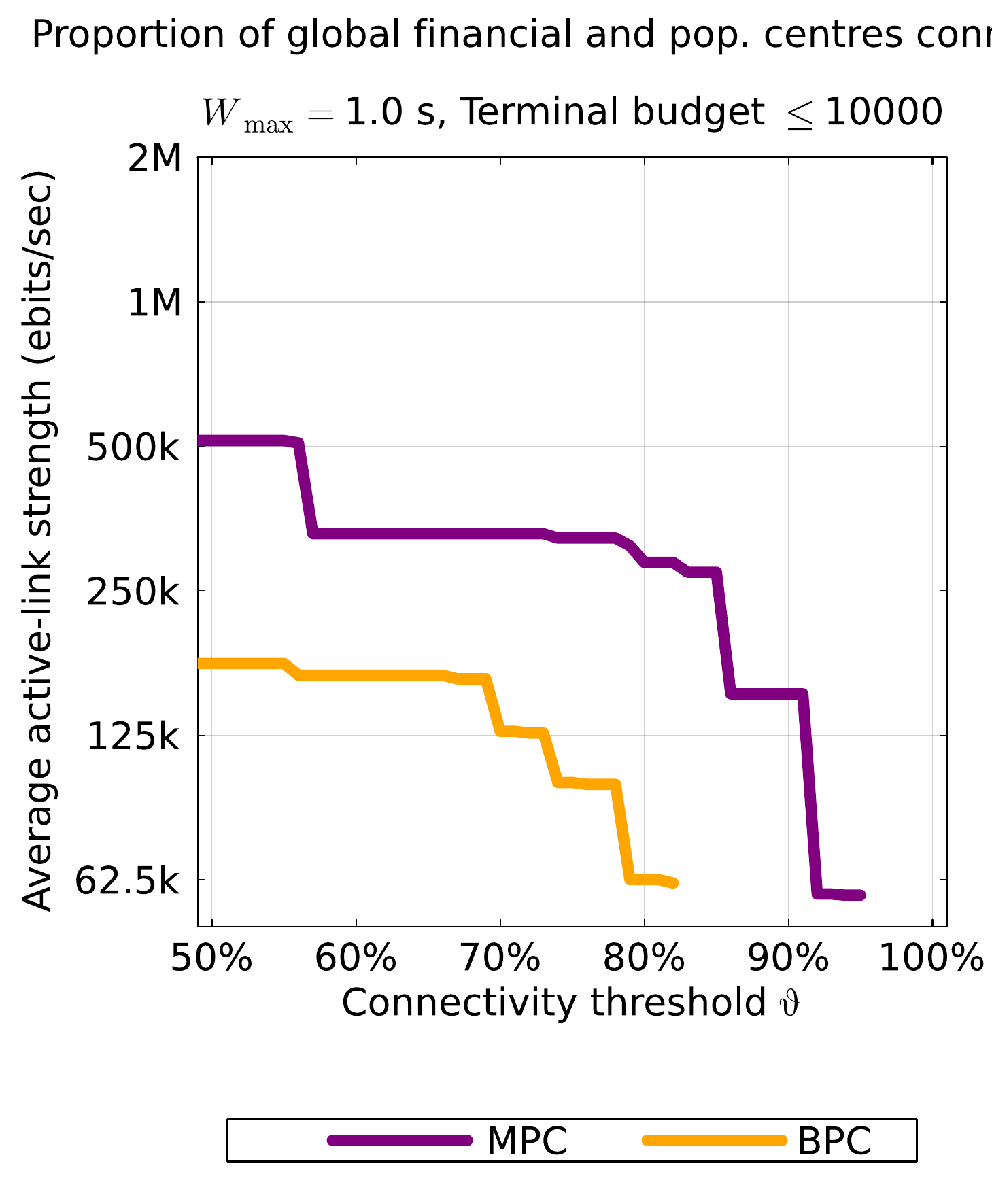}
            \put(5,95){\small\textbf{(b)}}
            \put(18,95){\small\centering Terminal Budget $\leq 10000$}
        \end{overpic}
    \end{subfigure}
    % \hfill
    \begin{subfigure}{0.30\linewidth}
        \centering
        \begin{overpic}[height=4cm,trim=1.9cm 5.8cm 0.3cm 3.3cm,clip]
        {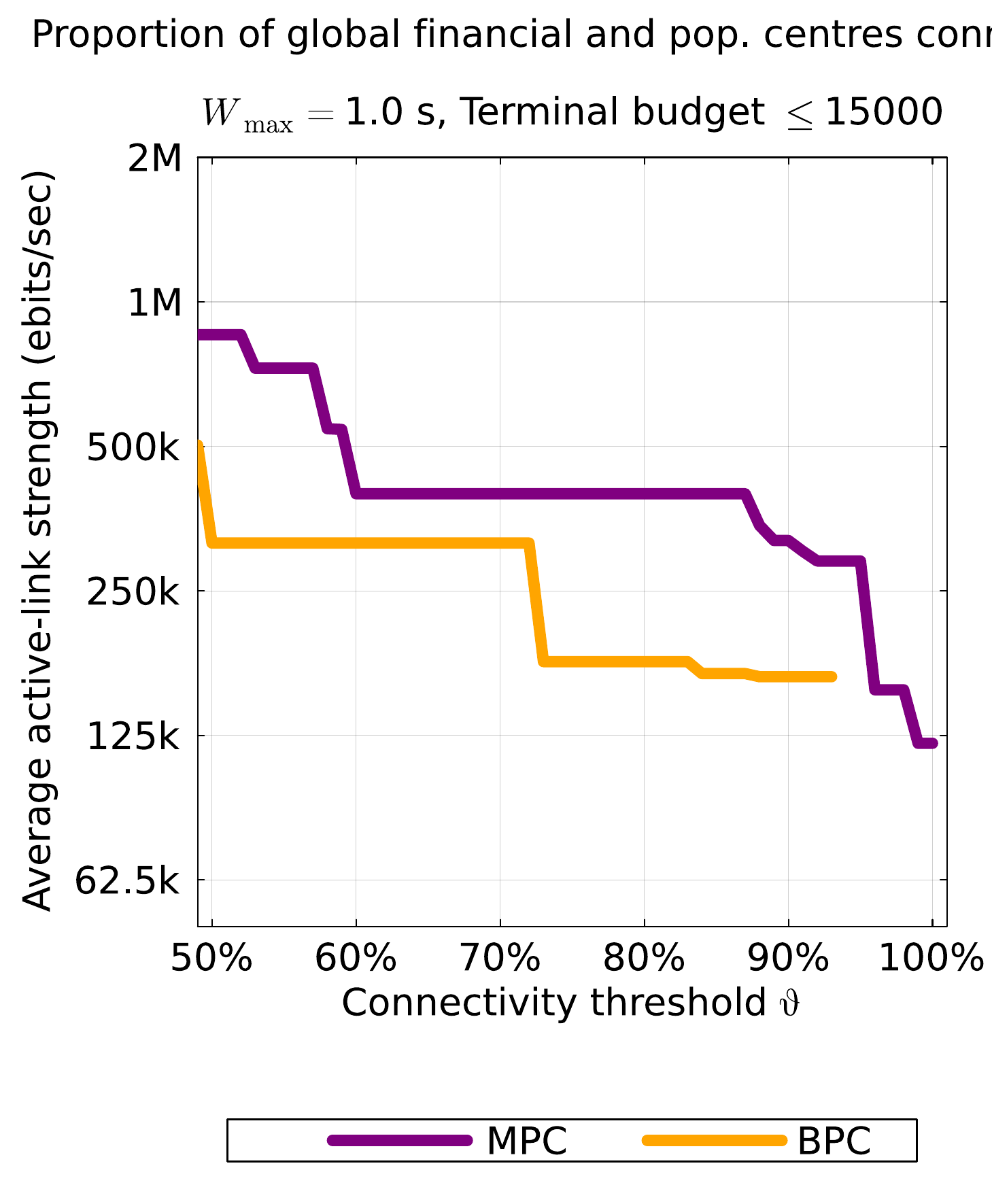}
            \put(5,95){\small\textbf{(c)}}
            \put(18,95){\small\centering Terminal Budget $\leq 15000$}
        \end{overpic}
    \end{subfigure}

    \begin{subfigure}{0.95\linewidth}
        \centering
        {\small Connectivity Threshold $\vartheta$ \par}
    \end{subfigure}

    \begin{subfigure}{0.95\linewidth}
        \centering
        \includegraphics[height = 0.8cm, trim = 1.9cm 0.5cm 0.3cm 27cm, clip]
        {figures/BPC_MPC_legend.pdf}
    \end{subfigure}
    \caption{
    \textbf{Latency-conditioned average active-link strength under bi-partite (BPC) and multi-party (MPC) connectivity.}
    Average active-link strength as a function of the connectivity threshold $\vartheta$ for fixed terminal budgets of 5000, 10000, and 15000, shown in (a)--(c), respectively.
    Curves compare multi-party connectivity (MPC, purple) and bi-partite connectivity (BPC, orange).
    At stringent connectivity thresholds, BPC degrades rapidly and often ceases to provide feasible configurations within the allowed waiting window.
    MPC remains viable over a broader range of $\vartheta$, indicating that multi-party servicing converts terminal resources into concurrent backbone connectivity more effectively.
    As the terminal budget increases, both policies improve, but MPC sustains usable active-link strength at higher connectivity thresholds.
    }
    \label{fig:strength_vs_terminals_city_multitelescope}
\end{figure}

\begin{table}[!htb]
\centering
\caption{Qualitative comparison to representative satellite-assisted quantum network studies. Metrics and operating regimes differ; we therefore compare modeling assumptions and the closest corresponding baseline in our framework.}
\begin{tabular}{llll}
\toprule
\textbf{Work} & \textbf{Orbit regime} & \textbf{GS placement} & \textbf{Primary objective} \\
\midrule
Khatri \emph{et al.}~\cite{khatri_spooky_2021} & LEO (global) & equal-angular lat--lon & end-to-end \\&&&entanglement distribution \\
Shao \emph{et al.}~\cite{shao_hybrid_2025} & MEO/HEO (regional) & uniform square grid  & rate--fidelity trade-offs \\ &&(continental)& (hybrid fiber/sat) \\ 
Present work & LEO (global) & $\alpha\in\{-1,0,>0\}$ lattice family & concurrent traffic-matrix \\ &&& connectivity under waiting \\&&&  constraints \\
\bottomrule
\end{tabular}

\label{tab:prior_work_compare}
\end{table}

\subsection{Secondary Architectural Insights}
\label{subsec:secondary_contributions}

Beyond the three primary architectural insights---anisotropic ground-station placement, augmented dual-shell constellations, and multi-party satellite connectivity---we identify several physical and constellation-level tuning parameters that shape performance within a fixed architectural class. These parameters do not alter the qualitative ordering established by the primary architectural principles, but they determine the latency--link-strength trade-offs available to a given architecture.

\vspace{-1em}
\paragraph{Satellite altitude.}
Satellite altitude has a pronounced impact on both connectivity latency and average link strength across all connectivity thresholds and constellation sizes considered. For both representative orbital plane counts shown (120 and 240 planes; see Section~\ref{subsec:planes_RAAN} for definitions), higher-altitude constellations consistently achieve lower mean wait times for a fixed satellite budget (Fig.~\ref{fig:altitude_waittime}). This ordering is preserved as the number of planes is increased, indicating that altitude-driven effects dominate over longitudinal densification (Fig.~\ref{fig:planes_strength}a). The improvement arises from the larger visibility footprints and longer contact durations associated with higher-altitude satellites, which enable simultaneous access to more ground stations and accelerate the growth of large connected components. However, this reduction in wait times is accompanied by lower average link strengths due to increased channel losses (Fig.~\ref{fig:planes_strength}b). 

These results indicate that satellite altitude acts as a dominant physical design lever within the regimes studied here, jointly shaping both connectivity latency and link quality, while plane count provides a secondary, geometry-driven refinement.

\vspace{-1em}
\paragraph{Number of orbital planes.}
For a fixed satellite altitude increasing the number of orbital planes has virtually no effect on average active link strength, which is primarily governed by altitude-dependent channel loss (Fig.~\ref{fig:planes_strength}b). This is accompanied by a reduction in the mean wait times, driven by an increase in the number of satellites in the constellation. However, this reduction in wait time saturates beyond a moderate number of planes once full longitudinal visibility is achieved. This efficient utilization of resources also indicates the percolation benefits of an anisotropic grid complemented through an MPC service.

\begin{figure}[!htbp]
\vspace{2em}
    \centering
        \begin{subfigure}{0.48\linewidth}
        \begin{subfigure}{0.48\linewidth}
            \centering 
            \begin{overpic}[height = 4cm, trim = 0cm 1.3cm 26cm 2.4cm, clip]{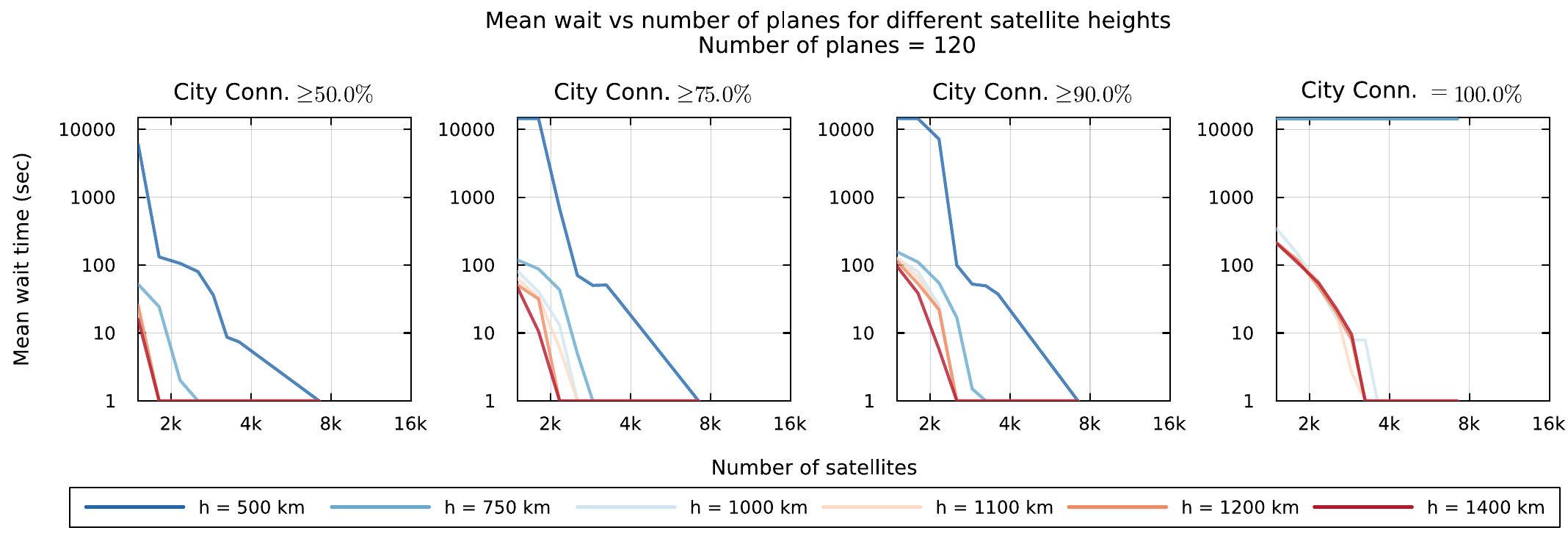}
            \put(50,86){\small  $\vartheta \geq 50\%$}
            \put(60,96){\small \textbf{(a)} \# of orbital planes = 120}
        \end{overpic}
        \end{subfigure}
    \hspace{1em}
        \begin{subfigure}{0.45\linewidth}
            \centering 
            \begin{overpic}[height = 4cm, trim = 27.5cm 1.3cm 0 2.4cm, clip]{figures/k_7_pf_0.1_wait_vs_planes_num_planes_120_panels.pdf}
            \put(45, 101){\small $\vartheta = 100\%$}
        \end{overpic}
        \end{subfigure}
    \end{subfigure}
    \hspace{1em}
    \begin{subfigure}{0.45\linewidth}
        \begin{subfigure}{0.45\linewidth}
            \centering
            \begin{overpic}[height = 4cm, trim = 1.2cm 1.3cm 26cm 2.4cm, clip]{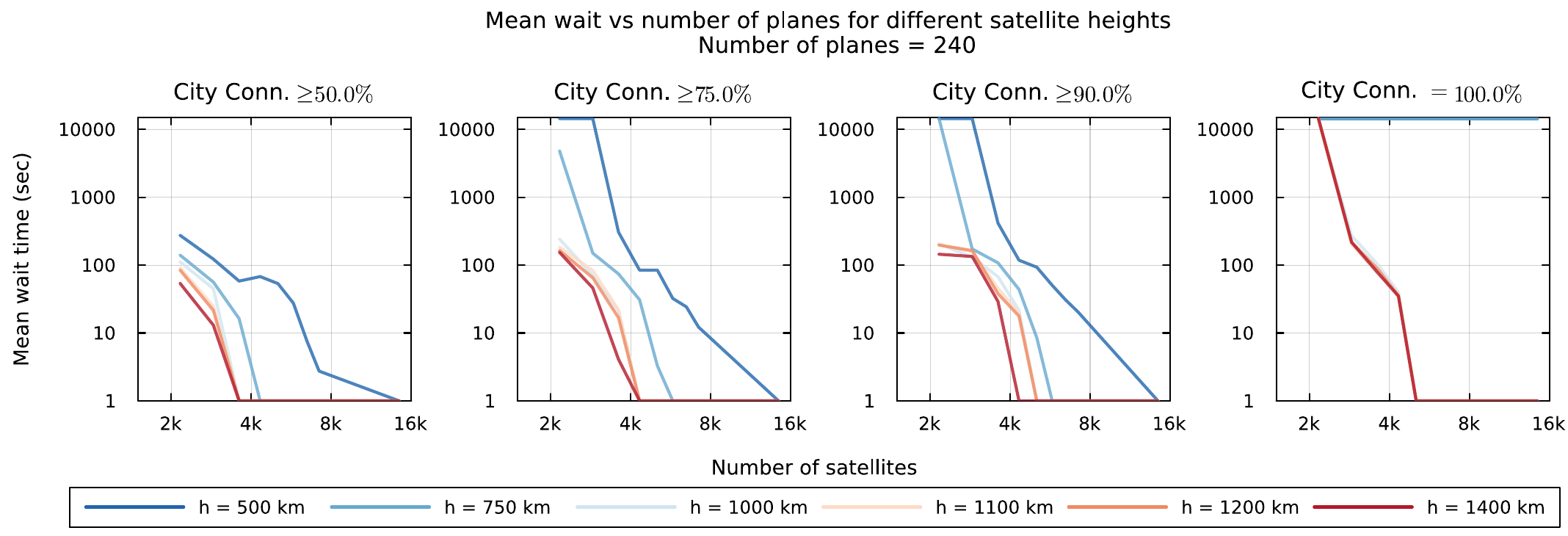}
                        \put(45, 98){\small  $\vartheta \geq 50\%$}
                                    \put(60,108){\small \textbf{(b)} \# of orbital planes = 240}
        \end{overpic}
        \end{subfigure}
    \hspace{1em}
        \begin{subfigure}{0.45\linewidth}
            \centering
            \begin{overpic}[height = 4cm, trim = 27.5cm 1.3cm 0cm 2.4cm, clip]{figures/k_7_pf_0.1_wait_vs_planes_num_planes_240_panels.pdf}
                        \put(45, 102){\small$\vartheta = 100\%$}
        \end{overpic}
                \end{subfigure}
        \end{subfigure}
        \centering
            \begin{subfigure}{0.95\linewidth}
        \centering
        \begin{overpic}[width=\linewidth, trim = 0 0cm 0.1cm 11cm, clip]{figures/k_7_pf_0.1_wait_vs_planes_num_planes_240_panels.pdf}
        \put(45,4){\small Number of Satellites}
        \end{overpic}
    \end{subfigure}
        \caption{
    \textbf{Mean wait time as a function of total satellite budget for different satellite altitudes and connectivity thresholds.}
    Panels (a), and (b) correspond to constellations with 120 and 240 orbital planes, respectively, shown as representative plane counts.
    The first and third plots correspond to $\vartheta \geq 50\%$, while the second and fourth plots correspond to $\vartheta=100\%$.
    Colors indicate different satellite altitudes, with bluer colors representing satellites closer to the Earth's surface, while reds representing higher altitude LEO satellites.
    Increasing the number of planes reduces wait times across all altitudes by improving longitudinal coverage; however, the relative ordering with altitude is preserved across both panels.
    Higher-altitude constellations consistently achieve lower wait times for a fixed satellite budget, with the advantage becoming more pronounced at stringent connectivity thresholds.
    }
    \label{fig:altitude_waittime}
\end{figure}
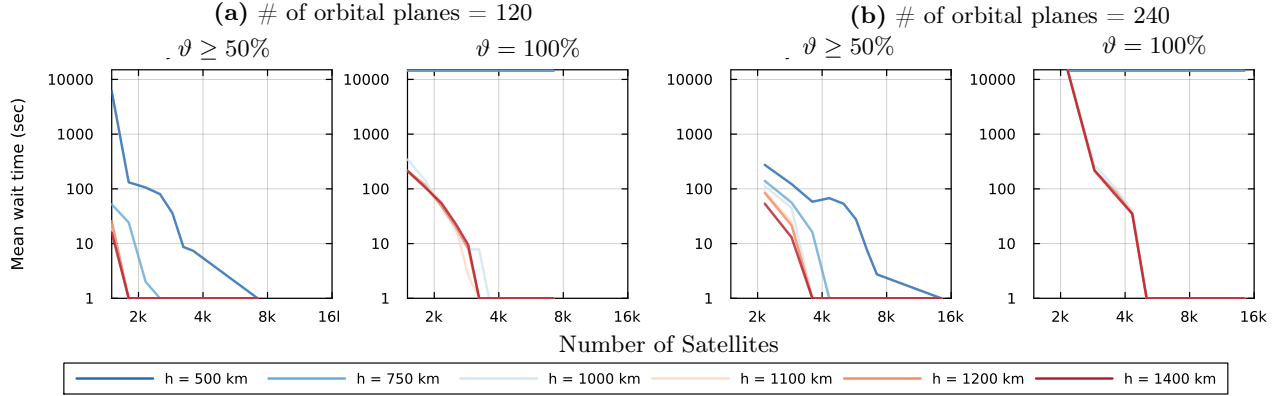

\vspace{-1em}
\paragraph{Number of satellites per orbital plane.}

While increasing the number of planes have only marginal returns, tuning the satellite density per plane reduces along-track spacing and improves temporal overlap between successive satellite footprints. This yields modest reductions in wait time at low densities by smoothing short-term coverage fluctuations (Fig.~\ref{fig:planes_waittime}). However, these gains saturate rapidly. 

\begin{figure}[t!]
\vspace{2em}
        \begin{subfigure}{0.48\linewidth}
        \begin{subfigure}{0.48\linewidth}
            \centering 
            \begin{overpic}[height = 3.7cm, trim = 0cm 1.6cm 26.2cm 2.4cm, clip]{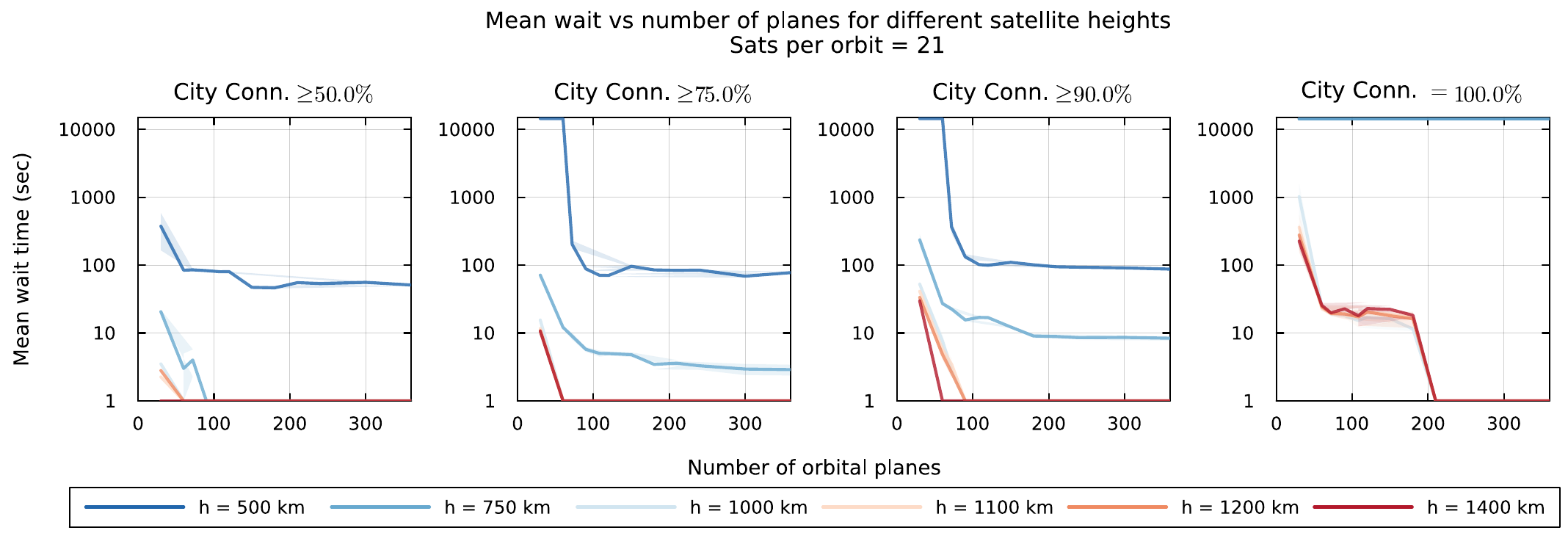}
            \put(60,97){\small \textbf{(a)} Mean Wait Time $W$ (s)}
            \put(49,86){\small $\vartheta \geq 50\%$}
            \put(55,-5){\small Number of Orbital Planes}
        \end{overpic}
        \end{subfigure}
    \hspace{1em}
        \begin{subfigure}{0.45\linewidth}
            \centering 
            \begin{overpic}[height = 3.7cm, trim = 27.6cm 1.6cm 0 2.4cm, clip]{figures/k_7_pf_0.1_wait_vs_planes_spp21_panels.pdf}
            \put(39, 100){\small $\vartheta = 100\%$}
        \end{overpic}
        \end{subfigure}
    \end{subfigure}
    \hspace{1em}
    \begin{subfigure}{0.48\linewidth}
        \begin{subfigure}{0.48\linewidth}
            \centering
            \begin{overpic}[height = 3.7cm, trim = 0cm 2cm 26cm 2.4cm, clip]{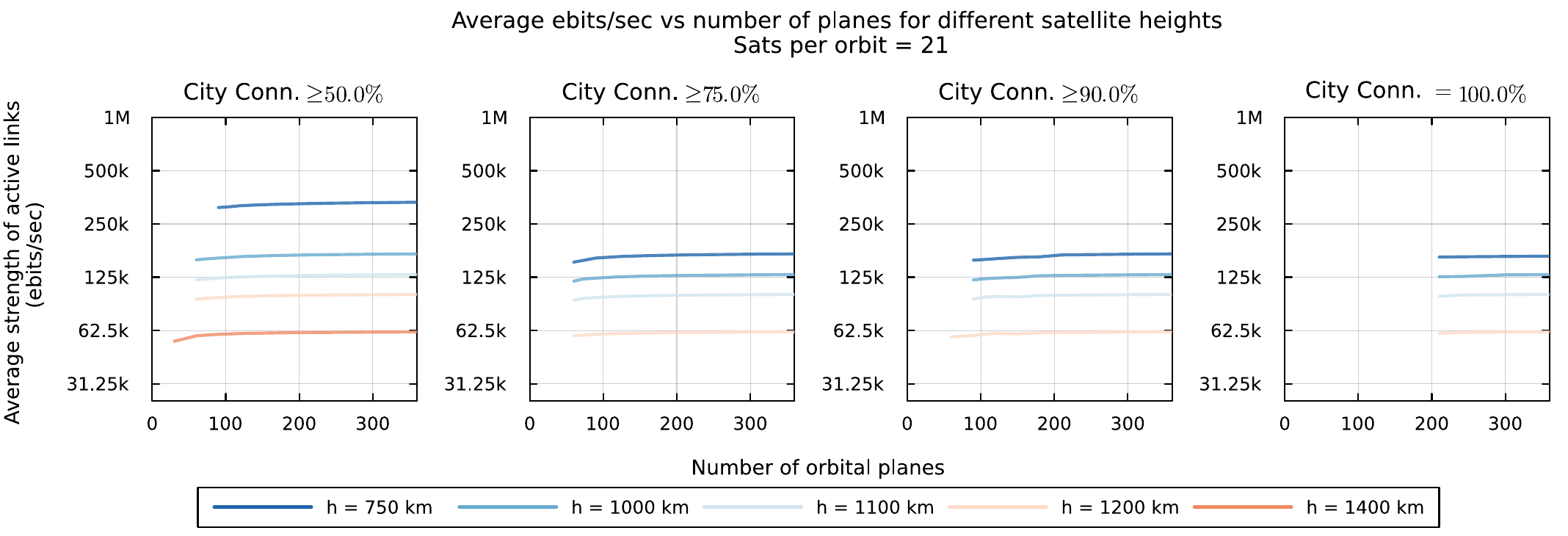}
            \put(15,89){\small \textbf{(b)} Latency-conditioned link strength for $W_\mathrm{max} = 1$~s}
                    \put(49, 80){\small $\vartheta \geq 50\%$}
                    \put(55,-5){\small Number of Orbital Planes}
        \end{overpic}
        \end{subfigure}
    \hspace{1em}
        \begin{subfigure}{0.45\linewidth}
            \centering
            \begin{overpic}[height = 3.7cm, trim = 27.5cm 2cm 0cm 2.4cm, clip]{figures/k_7_pf_0.1_ebits_vs_planes_spp21_panels.pdf}
            \put(39, 94){\small $\vartheta = 100\%$}
        \end{overpic}
                \end{subfigure}
        \end{subfigure}
            \begin{subfigure}{0.99\linewidth}
            \vspace{0.75em}
        \centering
        \includegraphics[height = 0.6cm, trim = 0 0cm 0cm 11cm, clip]{figures/k_7_pf_0.1_wait_vs_planes_spp21_panels.pdf}
    \end{subfigure}
\caption{
    \textbf{Effect of orbital-plane count and satellite altitude on connectivity latency and average link strength.}
    Panel (a) shows the mean wait time required to reach a specified city--city connectivity threshold $\vartheta$, while panel (b) shows the corresponding latency-conditioned average active-link strength for $W_{\max}=1\,\mathrm{s}$.
    The first and third plots correspond to $\vartheta \geq 50\%$, while the second and fourth plots correspond to $\vartheta=100\%$.
    Colors indicate satellite altitude.
    All results are shown for a fixed number of satellites per plane ($\mathrm{SPP}=21$), an anisotropic ground-station grid, and the nearest-neighbor multi-party connectivity (MPC) model.
    Increasing the number of orbital planes reduces wait times at low plane counts by improving longitudinal coverage, but the gains saturate once sufficient longitudinal diversity is available.
    Increasing altitude can also reduce wait time by enlarging each satellite footprint, especially at low plane counts and moderate connectivity thresholds.
    However, at stringent connectivity thresholds, this latency benefit becomes limited: the bottleneck is no longer only local visibility, but the simultaneous formation of a sufficiently connected global backbone.
    In contrast, altitude has a strong effect on active-link strength, with higher-altitude satellites producing weaker links because of increased channel loss.
    Thus, altitude introduces a visibility--loss trade-off: larger footprints can improve connectivity latency, but only at the cost of reduced link strength.
    This trade-off motivates an intermediate-altitude operating regime, rather than simply choosing the highest altitude or the largest number of orbital planes.
    }
\label{fig:planes_strength}
\end{figure}

\begin{figure}[bth!]
\vspace{2em}
    \centering
        \begin{subfigure}{0.48\linewidth}
        \begin{subfigure}{0.48\linewidth}
            \centering 
            \begin{overpic}[height = 3.7cm, trim = 0cm 1.6cm 26.2cm 2.4cm, clip]{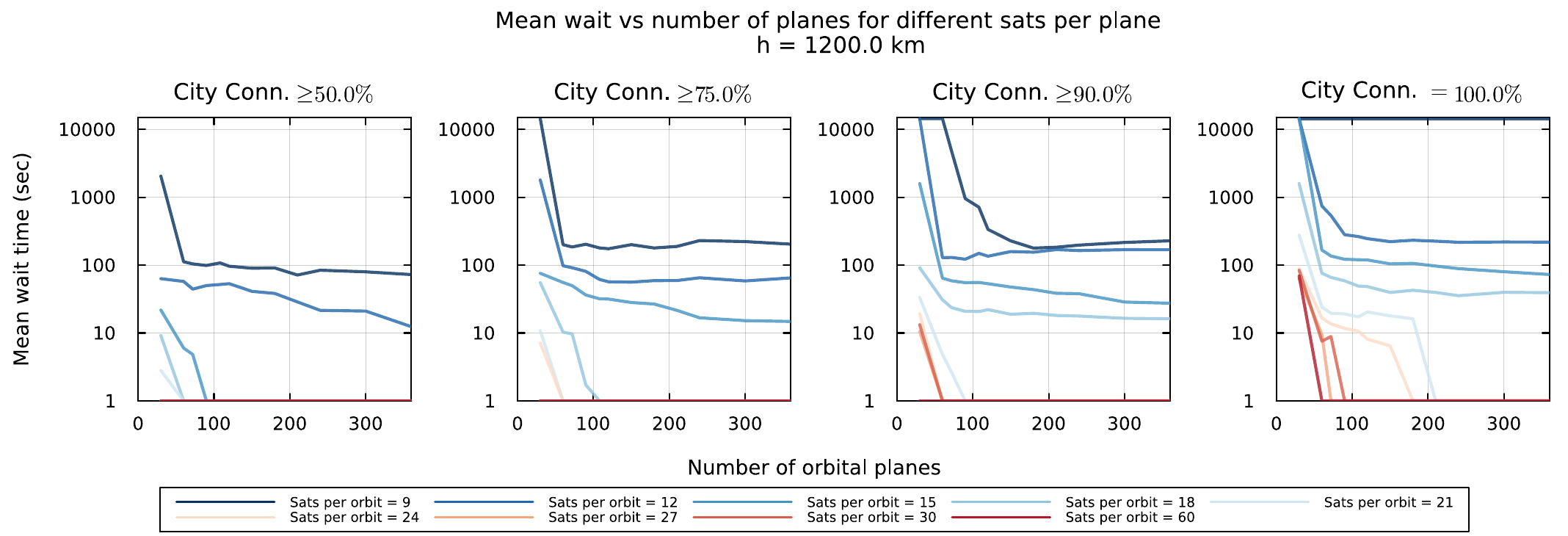}
            \put(60,97){\small \textbf{(a)} Mean Wait Time $W$ (s)}
            \put(45,86){\small $\vartheta \geq 50\%$}
            \put(55,-5){\small Number of Orbital Planes}
        \end{overpic}
        \end{subfigure}
    \hspace{1em}
        \begin{subfigure}{0.45\linewidth}
            \centering 
            \begin{overpic}[height = 3.7cm, trim = 27.6cm 1.6cm 0 2.4cm, clip]{figures/k_7_pf_0.1_wait_vs_planes_h1200km_panels.pdf}
            \put(35, 101){\small $\vartheta = 100\%$}
        \end{overpic}
        \end{subfigure}
    \end{subfigure}
    \hspace{1em}
    \begin{subfigure}{0.48\linewidth}
        \begin{subfigure}{0.48\linewidth}
            \centering
            \begin{overpic}[height = 3.7cm, trim = 0cm 2cm 26cm 2.4cm, clip]{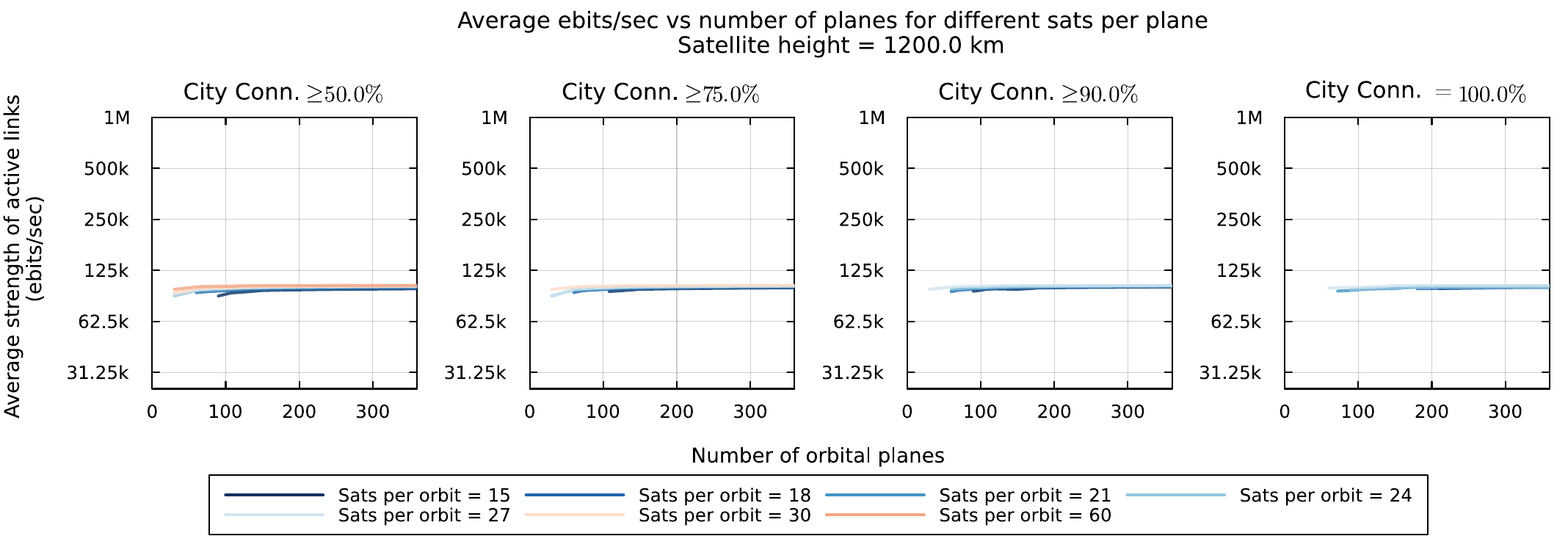}
            \put(15,89){\small \textbf{(b)} Latency-conditioned link strength for $W_\mathrm{max} = 1$~s}
                    \put(45, 80){\small $\vartheta \geq 50\%$}
                    \put(55,-5){\small Number of Orbital Planes}
        \end{overpic}
        \end{subfigure}
    \hspace{1em}
        \begin{subfigure}{0.45\linewidth}
            \centering
            \begin{overpic}[height = 3.7cm, trim = 27.5cm 2cm 0cm 2.4cm, clip]{figures/k_7_pf_0.1_ebits_vs_planes_h1200km_panels.pdf}
                        \put(35, 93){\small $\vartheta = 100\%$}
        \end{overpic}
                \end{subfigure}
        \end{subfigure}
            \begin{subfigure}{0.95\linewidth}
            \vspace{0.75em}
        \centering
        \includegraphics[height = 0.65cm, trim = 0 0cm 0cm 11cm, clip]{figures/k_7_pf_0.1_wait_vs_planes_h1200km_panels.pdf}
    \end{subfigure}
    \caption{
    \textbf{Effect of orbital-plane count and satellites per orbital plane on connectivity latency and link quality.}
    Panel (a) shows the mean wait time required to reach a specified city--city connectivity threshold $\vartheta$, while panel (b) shows the corresponding latency-conditioned average active-link strength for $W_{\max}=1\,\mathrm{s}$.
    The first and third plots correspond to $\vartheta \geq 50\%$, while the second and fourth plots correspond to $\vartheta=100\%$.
    Colors indicate the number of satellites per orbital plane (SPP), with blue curves denoting lower SPP and red curves denoting higher SPP.
    All results are shown for a fixed satellite altitude of $h=1200\,\mathrm{km}$, an anisotropic ground-station grid, and the nearest-neighbor multi-party connectivity (MPC) model.
    Increasing the number of orbital planes initially reduces wait time by improving longitudinal coverage and reducing visibility gaps.
    However, beyond a moderate number of planes, further subdivision yields diminishing returns because satellites are redistributed across additional planes without substantially increasing useful visibility opportunities.
    In contrast, the latency-conditioned active-link strength is nearly invariant across both the number of orbital planes and SPP at fixed altitude, indicating that link quality is dominated primarily by altitude-dependent channel loss rather than orbital-plane subdivision; this altitude dependence is examined separately in Fig.~\ref{fig:planes_strength}.
    Overall, these results show that, once a sufficient number of orbital planes is available, increasing the number of satellites per plane provides larger wait-time gains than further increasing the number of planes, especially for stringent connectivity thresholds.
    }
    \label{fig:planes_waittime}
\end{figure}

It is important to note that the secondary parameters discussed above are not independent: when the total satellite budget is constrained, increasing the number of orbital planes reduces satellites per plane, introducing a trade-off between longitudinal coverage and temporal continuity that can limit or even degrade latency performance. Finally, while careful tuning of altitude, plane count, and orbital packing can provide incremental improvements, these secondary knobs cannot compensate for suboptimal ground-station geometry, single-shell deployment, or single-link satellite operation. This reinforces an architecture-first design principle for satellite-serviced quantum backbones, with secondary parameters used for refinement rather than as primary drivers of global-scale, low-latency connectivity.
\section{Conclusion and Discussion}
\label{sec:Conclusion}

This work examined the architectural requirements for satellite-serviced quantum network backbones capable of supporting \emph{concurrent, large-scale} entanglement connectivity. Using a network-level simulation framework that couples orbital geometry, ground-station placement, and constrained satellite servicing, we evaluated how backbone design choices---including ground-station geometry and satellite constellation structure---determine connectivity latency, concurrency, and effective link availability. Rather than optimizing a specific hardware platform, our goal was to identify the architectural features that govern global-scale performance under realistic timing and resource constraints.

Three primary design insights emerge. First, ground-station geometry is a dominant determinant of connectivity latency. Longitudinally collapsed and uniform grids, commonly adopted in prior proposals for their simplicity, oversupply ground stations in regions that already experience high satellite visibility while under-serving regions with sparse orbital access. Anisotropic ground-station placement addresses this imbalance by redistributing station density as a function of latitude, accelerating the formation of large connected components across all connectivity thresholds. By reducing time-to-connectivity, anisotropic grids relax demands on ground-station quantum-memory coherence and make high-threshold connectivity achievable under near-term hardware assumptions.

Second, constellation structure plays a critical role through spatial coverage diversity. Single-shell constellations exhibit visibility gaps that delay global connectivity, particularly when near-complete traffic-matrix coverage is required. Redistributing a fixed satellite budget across diverse inclination shells mitigates these gaps, and enables earlier and more stable connectivity without increasing constellation size. This effect is absent at low connectivity thresholds with smaller constellations, but becomes pronounced for higher connectivity thresholds or with larger constellations. This indicates that augmented dual-shell architectures are essential for backbones intended to provide persistent global service rather than opportunistic long-distance links.

Third, per-satellite servicing capability imposes a fundamental concurrency constraint. Under bi-partite operation, each satellite can activate at most one ground-station pair at a time, limiting the rate at which connectivity accumulates even when many stations are simultaneously visible. Multi-party connectivity relaxes this constraint by allowing a single satellite to service multiple ground stations concurrently within its footprint. Our results show that this capability substantially reduces connectivity latency by increasing the rate of edge accumulation in the ground-station graph across all connectivity-thresholds. These gains arise from improved utilization of existing geometric visibility while providing improved per-link channel quality.

These results highlight a quantum-specific distinction from classical LEO satellite networking. While classical LEO networks also face dynamic topology, handovers, latency variation, and transient congestion, these effects can be mitigated through routing, buffering, retransmission, congestion control, and store-and-forward operation. In quantum networks, the transported resource is entanglement: it cannot be copied, and cannot be stored without loss in quality. Thus, useful end-to-end service requires simultaneous multi-hop path availability within finite memory lifetimes. Consequently, fixed ground-station geometry and per-satellite concurrency are not merely deployment or transport-layer details, but architectural constraints that determine whether a global quantum backbone can form.

Secondary analyses clarify how constellation-level tuning parameters shape performance within a fixed architecture. Satellite altitude emerges as the dominant physical lever, governing the trade-off between connectivity latency and average link strength by setting the scale of visibility footprints and propagation loss. In contrast, increasing the number of orbital planes or satellites per plane provides diminishing returns once moderate densities are reached, and can even degrade latency when total satellite budgets are constrained by reducing temporal continuity along each orbit. 

Together, these results support an \emph{architecture-first} perspective on satellite-serviced quantum backbones. Ground-station placement, constellation structure, and satellite servicing capability fundamentally shape achievable global performance, while secondary parameters serve primarily as refinement tools. The advantages of anisotropic grids, shell diversity, and multi-party connectivity are complementary rather than redundant: each addresses a distinct bottleneck that cannot be eliminated by scaling satellite count alone.

Several limitations and directions for future work remain. The present study deliberately focuses on \emph{architectural and geometric constraints}: how ground-station placement, constellation structure, and per-satellite servicing policies determine the \emph{availability}, \emph{concurrency}, and \emph{temporal accumulation} of entanglement opportunities at the network level. We model link feasibility and connectivity formation under explicit visibility and servicing constraints, but do not explicitly track detailed quantum- or physical-layer impairments such as atmospheric turbulence, cloud cover, pointing jitter, detector noise, Bell-state measurement success probability, memory decoherence, purification yield, or end-to-end fidelity evolution. These effects will shift absolute performance values by introducing additional loss and noise, intermittency, and variability in link availability. However, we expect the primary architectural orderings identified here to be robust, because they arise from geometric coverage patterns and concurrency bottlenecks---namely, how quickly independent entanglement edges can be activated and accumulated across the ground-station graph. Incorporating detailed physical-layer models and explicit quantum protocol dynamics, including fidelity evolution, purification strategies, and memory-limited scheduling, is therefore an important direction for deployment-level performance prediction.

A separate direction is to move beyond bipartite elementary links and explicitly model multi-partite entanglement distribution. The multi-terminal satellite architectures considered here naturally admit settings in which a single satellite could distribute GHZ states~\cite{greenberger1989going} or other multi-partite resources~\cite{W_state, DMBQC2007} among multiple ground stations within the same visibility footprint. Such resources may be useful for conference-key agreement~\cite{Pickston_2023, GHZ_secret}, distributed sensing~\cite{DRC17, ZZS18}, distributed computing~\cite{coffman_distributed_2000, caleffi2022_2} and multi-user entanglement services~\cite{pant_routing_2019, miguel2023quantum, DMBQC2007}. Multi-partite entanglement distribution could change both the service topology and the appropriate connectivity metrics. While studying multi-partite protocols is beyond the scope of the present work, the architectural framework developed here provides a direct basis for doing so.

We also do not perform a detailed techno-economic analysis. Launch cost, satellite manufacturing cost, ground-station deployment cost, and terminal size--weight--power trade-offs are highly platform-dependent and are evolving rapidly as optical payloads and satellite buses mature~\cite{christie2022launch, li2022techno, sweeting2018modern}. While our terminal-normalized comparisons provide a first-order proxy for payload and operational complexity, they do not capture broader system-level cost drivers such as constellation replenishment, ground infrastructure ownership models, or economies of scale. A comprehensive cost--performance analysis that jointly considers architectural design choices, hardware capabilities, and economic constraints will be essential for translating architectural feasibility into deployable global quantum network infrastructures. In particular, architectural choices that reduce time-to-connectivity and improve concurrency may translate directly into cost advantages by relaxing requirements on quantum memory lifetime, terminal count, and satellite dwell-time utilization.

Finally, while we focus on concurrent global connectivity as a stringent operating regime, this objective is directly motivated by emerging applications such as multi-user QKD, distributed quantum sensing, and entanglement-assisted network services, where isolated long-distance links are insufficient. Extending the present framework to heterogeneous traffic matrices, demand-weighted objectives, and integrated satellite--terrestrial repeater architectures represents a natural next step.

Overall, this study shows that scalable satellite-assisted quantum-network backbones require the joint design of latitude-aware ground-station lattices, shell-diverse constellations, and multi-point satellite service. These architectural variables determine not only link availability, but the waiting-time and concurrency regimes in which global entanglement connectivity becomes feasible.

\FloatBarrier

\section*{Code and Data availability}\label{sec:code_and_data}
The simulator and analysis scripts are implemented in Julia. To support reproducibility, all constellation parameters, ground-station configurations, and simulator settings are exported alongside the reported results. The codebase will be released publicly upon publication; in the interim it is available from the authors upon reasonable request.

\section*{Supplementary Materials}\label{sec:supplementary_materials}
Supplementary Materials are available at this persistant repository~\cite{supplementary_materials}.

\section*{Author Contributions}\label{sec:author_contributions}
DT and PM formulated the overall problem and design. DT supervised the project. PM developed the modeling framework, designed and implemented the simulations, performed the numerical evaluations, and analyzed the results. DT and SH contributed to the technical development of the model and interpretation of the results. AW contributed to initial discussions. PM wrote the manuscript with input from DT and SH. All authors reviewed and edited the manuscript.

\section*{Competing Interests}\label{sec:competing_interests}
The authors declare no competing interests.

\section*{Acknowledgment}\label{sec:acknowledgment}
We thank the Manning College of Information and Computer Sciences at the University of Massachusetts Amherst for providing access to their High Performance Computation Cluster. All authors acknowledge funding support from the NSF ERC Center for Quantum Networks, Grant No. EEC-1941583. DT, AW, and PM acknowledge NSF Grant No. CNS-2402861 and DOE Grant No. AK0000000018297. 
\appendix
\appendixpage
% ============================================================
% Appendix: Forward Wait Time
% ============================================================
\section{Forward Wait Time}
\label{app:forward_wait_time}

In time-varying quantum networks, links between nodes evolve stochastically due to physical processes such as satellite visibility windows, entanglement generation attempts, and environmental photon loss. To characterize how long a network remains below a connectivity threshold and how long a user must wait for the criterion to be satisfied, we define two related but distinct discrete-time quantities: the \emph{down-run duration} (the length of a contiguous episode in which the criterion fails) and the \emph{forward wait time} (the residual time from an arbitrary reference epoch until the criterion is next satisfied). We report down-run statistics as our primary metric and relate them to forward-wait behavior via renewal-theoretic identities.

Throughout this appendix we work in discrete time, consistent with the simulator described in Section~\ref{subsec:simulation_methodology} in the main text. Time is sampled at uniform epochs
\begin{align*}
    t_k := t_{\mathrm{start}} + k\,\delta t, \qquad k = 0,1,\dots,K,
\end{align*}
where $t_{\mathrm{start}}$ denotes the initial simulation epoch, $\delta t>0$ is the fixed time step, and $K$ is the final discrete-time index. The simulation horizon is $T_{\mathrm{sim}}=K\,\delta t$. All connectivity events and threshold evaluations are therefore indexed by the epoch counter $k$. A continuous-time analogue is obtained by replacing the discrete hitting index in Definition~\ref{def:fwd_wait_discrete} by the usual hitting time $\inf\{\tau \ge 0 : \mathcal{C}(G(t_0+\tau)) \text{ holds}\}$.

\paragraph{Discrete-time setup.}
Our simulations produce a discrete-time trace $\{G(t_k)\}_{k=0}^K$, where $G(t_k)=(V,E(t_k))$ is the graph of available entanglement links at epoch $t_k$. Let $\mathcal{C}(G(t_k))$ be a predicate describing the desired connectivity event, such as the largest connected component exceeding a prescribed fraction of nodes or the fraction of city pairs connected exceeding a threshold. Define the Boolean indicator
\begin{align*}
    U_k := \mathbf{1}\{\mathcal{C}(G(t_k))\} \in \{0,1\}, \qquad k=0,1,\dots,K.
\end{align*}
The network is \emph{up} at epoch $t_k$ if $U_k=1$ and \emph{down} otherwise.

% ------------------------------------------------------------
\subsection{Down-run duration: definition and reported metric}
\label{app:downrun}

\begin{definition}[Down-run duration]
\label{def:downrun}
A \emph{down run} is a maximal contiguous block of epochs at which $U_k=0$. Formally, a down run of length $\ell$ starting at index $k_0$ satisfies
\[
    U_{k_0-1}=1 \ \text{(or $k_0=0$)},\quad
    U_{k_0}=U_{k_0+1}=\cdots=U_{k_0+\ell-1}=0,\quad
    U_{k_0+\ell}=1 \ \text{(or $k_0+\ell-1=K$)}.
\]
Its duration is $L := \ell\,\delta t$ seconds.
\end{definition}

Let $L_1,L_2,\ldots,L_M$ denote the down-run durations extracted from the analysis window. A terminal down run that is truncated by the end of the window is right-censored; we exclude such truncated terminal runs from the reported down-run statistics.

We summarize down-run durations by reporting their empirical mean and standard deviation,
\begin{align*}
    \hat{\mu}_L &:= \frac{1}{M}\sum_{n=1}^M L_n, \qquad
    \hat{\sigma}_L := \left(\frac{1}{M-1}\sum_{n=1}^M (L_n - \hat{\mu}_L)^2\right)^{1/2},
\end{align*}
together with selected quantiles ($p_{10}$, $p_{50}$ (median), $p_{90}$).

% ------------------------------------------------------------
\subsection{Forward wait time: definition and relation to down runs}
\label{app:fwt}

\begin{definition}[Forward Wait Time]
\label{def:fwd_wait_discrete}
For an epoch $t_k$, define the hitting index
\[
    J_k := \min\{j \ge k : U_j = 1\},
\]
with $J_k=+\infty$ if the set is empty. The forward wait time is
\[
    \hat{W}_{\mathcal{C}}(t_k) := (J_k-k)\,\delta t,
\]
with $\hat{W}_{\mathcal{C}}(t_k)=+\infty$ when $J_k=+\infty$ (a right-censored observation under a finite trace).
\end{definition}

Note that $\hat{W}_{\mathcal{C}}(t_k)=0$ whenever $U_k=1$ (the criterion already holds). In contrast, the down-run duration $L_n$ is the \emph{total} length of a down episode. Thus, $L_n$ characterizes how long the network stays below threshold once it drops below it, while $\hat{W}_{\mathcal{C}}(t_k)$ is the \emph{residual} time until recovery from the particular epoch $t_k$.

% ------------------------------------------------------------

\subsection{Stopping time property}
Let $(\Xi,\mathcal{F},\mathbb{P})$ be a probability space on which the process $\{G(t_k)\}_{k\ge 0}$ is defined, and let $(\mathcal{F}_k)_{k\ge 0}$ be its natural filtration, $\mathcal{F}_k=\sigma(G(t_0),\ldots,G(t_k))$.

\begin{proposition}[Stopping time]
\label{prop:stopping_time_fwd_wait_discrete}
Assume $U_j$ is $\mathcal{F}_j$-measurable for all $j$. Then, for any fixed $k$, the random index $J_k$ is a stopping time with respect to $(\mathcal{F}_j)_{j \ge k}$, and hence $\hat{W}_{\mathcal{C}}(t_k)=(J_k-k)\delta t$ is measurable with respect to $\mathcal{F}_{J_k}$ on $\{J_k<\infty\}$.
\end{proposition}

\begin{proof}
For any $m\ge k$, the event $\{J_k\le m\}$ is equivalent to $\bigcup_{j=k}^m \{U_j=1\}$. Each event $\{U_j=1\}$ belongs to $\mathcal{F}_j \subseteq \mathcal{F}_m$, hence the finite union belongs to $\mathcal{F}_m$.
\end{proof}

% ------------------------------------------------------------
\subsection{Empirical estimation over a finite window}
Given $\hat{W}_{\mathcal{C}}(t_k)$ from Definition~\ref{def:fwd_wait_discrete}, we estimate forward-wait performance from a finite simulated trace by computing $\hat{W}_{\mathcal{C}}(t_k)$ at each epoch $t_k$ in the analysis window. If the connectivity criterion is never satisfied again within the remaining trace, we set $\hat{W}_{\mathcal{C}}(t_k)=+\infty$ (a right-censored observation). In our reported results, we discard right-censored samples and report the empirical mean of the remaining finite forward-wait values:
\[
    \hat{\bar{W}}
    :=
    \frac{1}{|\mathcal{T}|}\sum_{t_k\in\mathcal{T}} \hat{W}_{\mathcal{C}}(t_k),
    \qquad
    \mathcal{T}:=\{t_k:\hat{W}_{\mathcal{C}}(t_k)<\infty\}.
\]
Discarding $+\infty$ values yields a finite-horizon estimate and can bias $\hat{\bar{W}}$ downward when censoring is non-negligible.

% ------------------------------------------------------------
\subsection{Connection to renewal theory and the inspection paradox (stationary special case)}
\label{app:renewal}

This subsection links the discrete-time forward-wait definition to classical alternating-renewal intuition. We emphasize that the discussion below is a stationary special case and is not required for computing the empirical down-run statistics reported in the main text.

\paragraph{Renewal epochs and cycle decomposition.}
Let $\tau_1 < \tau_2 < \cdots$ denote successive \emph{renewal epochs} (upcrossings) at which the criterion becomes satisfied immediately after a period of non-satisfaction, i.e., indices $k$ such that $U_k=1$ and $U_{k-1}=0$ (with the convention $U_{-1}:=0$ so that an initial up epoch counts as a renewal). Define the $n$th inter-renewal (cycle) length in seconds as
\[
    \Delta_n := (\tau_{n+1}-\tau_n)\,\delta t .
\]
Within the $n$th cycle, the indicator process first remains \emph{high} (up) for an up-run duration $H_n$ and then remains \emph{low} (down) for a down-run duration $L_n$, so that
\[
    \Delta_n = H_n + L_n,
\]
where
\[
    H_n := \bigl|\{k\in\{\tau_n,\ldots,\tau_{n+1}-1\}:U_k=1\}\bigr|\,\delta t,\qquad
    L_n := \bigl|\{k\in\{\tau_n,\ldots,\tau_{n+1}-1\}:U_k=0\}\bigr|\,\delta t.
\]
Equivalently, $H_n$ is the maximal contiguous block of ones beginning at $\tau_n$, and $L_n$ is the maximal contiguous block of zeros that follows it and ends at $\tau_{n+1}-1$.

\paragraph{Mean residual wait when starting in a down epoch (inspection paradox).}
Under a stationary i.i.d.\ alternating-renewal approximation in which $\{(H_n,L_n)\}_{n\ge 1}$ are i.i.d.\ with finite second moments (see the cyclostationary caveat below), consider an observer who samples a uniformly random epoch \emph{conditioned on} currently being down (i.e., $U_k=0$). The remaining time until the end of the current down run is a residual-life random variable. The inspection paradox gives the exact mean residual wait
\begin{align*}
    \bar{W}_\downarrow
    := \mathbb{E}\!\left[W_{\mathcal{C}}(t_0)\mid U(t_0)=0\right]
    = \frac{\mathbb{E}[L^2]}{2\,\mathbb{E}[L]}
    = \frac{\mathbb{E}[L]}{2} + \frac{\mathrm{Var}[L]}{2\,\mathbb{E}[L]},
\end{align*}
where $L$ denotes a generic down-run duration distributed as $L_n$.

\paragraph{Unconditional mean forward wait for $W_{\mathcal{C}}$.}
Let
\[
    p := \mathbb{P}\{U(t_0)=1\} = \frac{\mathbb{E}[H]}{\mathbb{E}[H]+\mathbb{E}[L]}
\]
denote the long-run uptime fraction in the stationary alternating-renewal model. Since $W_{\mathcal{C}}(t_0)=0$ whenever $U(t_0)=1$ by definition, the unconditional mean forward wait satisfies
\[
    \bar{W} := \mathbb{E}\!\left[W_{\mathcal{C}}(t_0)\right] = (1-p)\,\bar{W}_\downarrow
    = \frac{\mathbb{E}[L^2]}{2\left(\mathbb{E}[H]+\mathbb{E}[L]\right)} .
\]

\paragraph{Inspection-paradox-corrected forward wait (derived; not directly reported).}
Under the stationary alternating-renewal approximation, $\bar{W}_\downarrow$ can be estimated from the down-run event sequence via the estimator
\[
    \hat{\bar{W}}_{\downarrow,\mathrm{th}}
    :=
    \frac{\hat{\mu}_L}{2} + \frac{\hat{\sigma}_L^2}{2\hat{\mu}_L},
\]
highlighting that variability in down-run lengths inflates the expected residual wait beyond $\hat{\mu}_L/2$.

\paragraph{Internal consistency check (not reported).}
As an internal validation under the stationary approximation, one can compute the epoch-level mean residual wait over down epochs,
\[
    \hat{\bar{W}}_\downarrow
    :=
    \frac{1}{|\mathcal{T}_\downarrow|}\sum_{t_k\in\mathcal{T}_\downarrow} \hat{W}_{\mathcal{C}}(t_k),
    \qquad
    \mathcal{T}_\downarrow := \{t_k : U_k = 0,\ \hat{W}_{\mathcal{C}}(t_k) < \infty\},
\]
and compare it against $\hat{\bar{W}}_{\downarrow,\mathrm{th}}$. Discrepancies quantify finite-window effects and trace-boundary censoring bias.

\paragraph{Remark (different ``forward'' time).}
The classical renewal-theory formula $\mathbb{E}[R]=\mathbb{E}[\Delta^2]/(2\mathbb{E}[\Delta])$ applies to the \emph{forward recurrence time to the next renewal epoch} (time until the next upcrossing) at a uniformly random epoch. This differs from $W_{\mathcal{C}}$, which is defined to be zero on up epochs and instead measures time until the criterion next holds.

\paragraph{Cyclostationary caveat.}
The i.i.d.\ renewal discussion above is an instructive stationary special case. In satellite-driven networks the process is typically cyclostationary: the deterministic geometry repeats with orbital period $P$ (approximately 90 minutes for LEO constellations) while stochastic link-generation success is superimposed. The cycle lengths $\Delta_n$ (and hence $H_n$ and $L_n$) are therefore not i.i.d., and renewal-theoretic identities serve only as qualitative guidance. The down-run duration statistics $\hat{\mu}_L$, $\hat{\sigma}_L$, and quantiles remain well-defined and directly estimable regardless.

% ------------------------------------------------------------
\subsection{Phase dependence in satellite-driven networks}
Cyclostationarity implies that connectivity (and hence waiting behavior) can depend on the reference orbital phase $\phi_k := t_k \bmod P$. Define the \emph{population} phase-conditional mean forward wait
\[
    \bar{W}(\phi) := \mathbb{E}\!\left[W_{\mathcal{C}}(t_0) \mid t_0 \bmod P = \phi\right],
\]
and the phase-averaged mean
\[
    \bar{W} = \frac{1}{P}\int_0^P \bar{W}(\phi)\,d\phi.
\]
In our evaluation we do not explicitly estimate the phase-resolved profile $\bar{W}(\phi)$. Instead, our reported down-run statistics aggregate all episodes extracted from the analysis window, yielding phase-averaged summaries over the phases represented in the trace. When the analysis window spans multiple orbital periods, this aggregation captures recurring coverage gaps through their contribution to the overall down-run distribution.

If desired, phase structure can be resolved by binning down runs by their start phase $\phi := t_{k_0}\bmod P$ (where $k_0$ is the start index of the run) and reporting down-run summaries within each phase bin. Such phase-resolved down-run statistics can reveal ``dead zones'' (phases with systematically long down episodes) that may be obscured by simple time-average connectivity metrics.

The renewal-theoretic identities in Appendix~\ref{app:forward_wait_time} describe population means under a stationary alternating-renewal approximation. Appendix~\ref{app:metrics_methodology} addresses a separate issue: how to attach correlation-adjusted standard errors to empirical estimators computed from a finite, temporally dependent trace (either per-second samples or event-indexed down-interval sequences).
% ============================================================
% Appendix: Metric Computation and Statistical Estimation
% (companion to Appendix~\ref{app:forward_wait_time})
% ============================================================
\section{Metric computation and statistical estimation}
\label{app:metrics_methodology}

This appendix documents how the simulator time series are converted into the scalar metrics reported in the main text, and how uncertainty estimates are computed in the presence of temporal dependence. It complements Appendix~\ref{app:forward_wait_time}, which defines down-run durations and forward-wait times.

\paragraph{Event sequences.}
In this appendix, an \emph{event} refers to a single contiguous \emph{down interval} for a specified threshold $\vartheta$: a maximal consecutive index set,
\begin{align*}
\{k_s,k_s+1,\ldots,k_e\}\subseteq\{1,\ldots,K\}
\end{align*}
such that $U_k^{(\vartheta)}=0$ for all $k\in\{k_s,\ldots,k_e\}$, with boundary conditions (when the indices exist within the analysis window) $U_{k_s-1}^{(\vartheta)}=1$ and $U_{k_e+1}^{(\vartheta)}=1$.
The corresponding down-interval duration (seconds) is
\begin{align*}
    L_n^{(\vartheta)} := (k_e-k_s+1)\,\delta t,
\end{align*}
and the \emph{event-level} sequence is $\{L_n^{(\vartheta)}\}_{n=1}^{M_\vartheta}$, where $M_\vartheta$ is the number of down intervals observed in the analysis window.
If the trace terminates while $U_k^{(\vartheta)}=0$, the final interval is right-censored and excluded from $\{L_n^{(\vartheta)}\}$.

\subsection{Discrete-time traces and notation}
\label{app:metrics_traces}

The simulator produces a discrete-time trace $\{G(t_k)\}_{k=1}^K$ at epochs $t_k=t_{\mathrm{start}}+k\,\delta t$. At each epoch, the network state includes a graph of available entanglement links $G(t_k)=(V,E(t_k))$. We use the following derived Boolean indicators:
\begin{align*}
    U_k^{(\vartheta)} &:= \mathbf{1}\{\mathcal{C}(G(t_k)) \ge \vartheta\}, \\
    U_{k,\mathrm{city}}^{(\vartheta)} &:= \mathbf{1}\{\mathcal{C}_{\mathrm{city}}(G(t_k)) \ge \vartheta\},
\end{align*}
where $\mathcal{C}$ is a network-wide connectivity measure (e.g., LCC fraction) and $\mathcal{C}_{\mathrm{city}}$ is the analogous city-pair connectivity measure.

All reported scalar metrics are computed over a fixed analysis window (typically the full trace after any warm-up), using a uniform sampling cadence $\delta t$.

\subsection{Connectivity metrics}
\label{app:metrics_connectivity}

\paragraph{Largest connected component (LCC) fraction.}
Let $\mathrm{LCC}(t_k)$ denote the size (node count) of the largest connected component of $G(t_k)$ at time $t_k$, and $|V|$ the total number of ground stations. The LCC fraction is
\begin{align*}
    \mathrm{LCCFrac}(t_k) := \frac{\mathrm{LCC}(t_k)}{|V|} \in [0,1].
\end{align*}

\paragraph{City-pair connectivity fraction.}
Let $\mathcal{L}$ denote the set of city pairs under consideration. At each epoch, define the indicator that a city pair $(i,j)\in\mathcal{L}$ is connected via the current graph $G(t_k)$; the city connectivity fraction is
\begin{align*}
    \mathrm{CityFrac}(t_k) := \frac{1}{|\mathcal{L}|}\sum_{(i,j)\in\mathcal{L}} \mathbf{1}\{i \leftrightarrow j \text{ in } G(t_k)\} \in [0,1].
\end{align*}

\paragraph{Connectivity thresholds.}
For each threshold $\vartheta\in\vartheta$ (e.g., $\vartheta=\{0.5,0.6,0.7,0.8,0.9\}$), we evaluate the induced up/down process $U_k^{(\vartheta)}$ (or $U_{k,\mathrm{city}}^{(\vartheta)}$). From this Boolean trace we extract the sequence of \emph{maximal down intervals} (events) and their durations, as defined above and in Appendix~\ref{app:forward_wait_time}. The main text reports \emph{down-run duration statistics} (quantiles and mean) computed from these event durations, rather than the epoch-level forward-wait average, since down-run durations remain directly interpretable under cyclostationarity.

\subsection{Elementary entanglement generation rate (time-averaged ebits)}
\label{app:metrics_ebits}

We consider an epoch elementary entanglement generation rate metric (ebits/s), denoted $E_k$. This metric only includes the links generated between two ground stations serviced directly through a satellite and without accounting for any entanglement swapping operation performed on the ground. The time-average is computed as the sample mean over the analysis window,
\begin{align*}
    \hat{\mu}_E := \frac{1}{K}\sum_{k=1}^K E_k,
\end{align*}
together with the sample median and sample standard deviation for descriptive purposes.

\subsection{Why naive standard errors are invalid under autocorrelation}
\label{app:metrics_autocorr_problem}

Many reported quantities are sample means computed from a time series (e.g., $\hat{\mu}_E$) or from a sequence of events extracted from that series (e.g., the list of down-run durations $L_1,\ldots,L_M$). In satellite-driven networks, these sequences are typically \emph{temporally dependent} due to orbital geometry and service-set overlap, so consecutive observations are not independent. The classical standard error formula
\begin{align*}
    \mathrm{SE}(\hat{\mu}) \approx \frac{\hat{\sigma}}{\sqrt{N}}
\end{align*}
is only justified under weak dependence assumptions that effectively approximate independence. When autocorrelation is present, it understates uncertainty by over-counting the number of independent samples.

\subsection{Autocorrelation-based effective sample size}
\label{app:metrics_tauint}

Let $\{X_i\}_{i=1}^N$ denote a real-valued sequence (either epoch-level samples such as $E_k$ or event-level samples such as $L_n$), with mean $\mu := \mathbb{E}[X_i]$, variance $\sigma^2 =: \mathrm{Var}(X_i)$, and the autocorrelation function $\gamma(\ell) =: \mathrm{Cov}(X_i, X_{i +\ell})$, for $\ell \ge 0$. The lag-$\ell$ autocorrelation function is defined as,
\begin{align*}
    \rho(\ell) := \frac{\mathrm{Cov}(X_i, X_{i+\ell})}{\mathrm{Var}(X_i)}.
\end{align*}
The sample mean is $\hat{\mu} := (\sum_{i=1}^N X_i)/N$.
For dependent samples, the variance of the sample mean is not $\sigma^2/N$. Under standard mixing conditions and a stationary approximation, its large-$N$ asymptotic form is~\cite{geyer1992practical}

\begin{align*}
    \mathrm{Var}(\hat{\mu}) \approx \frac{1}{N} \left(\gamma(0) + 2\sum_{\ell=1}^{\infty} \gamma(\ell) \right) = \frac{\sigma^2}{N} \left(1 + 2\sum_{\ell=1}^{\infty} \rho(\ell) \right),
\end{align*}
which motivates the \emph{integrated autocorrelation time}
\begin{align*}
    \tau_{\mathrm{int}} := \frac{1}{2} + \sum_{\ell=1}^\infty \rho(\ell),
\end{align*}
so that, 
\begin{align*}
    \mathrm{Var}(\hat{\mu})
\approx
\frac{\sigma^2}{N}\,(2\tau_{\mathrm{int}}).
\end{align*}
The effective sample size is therefore defined  as
\begin{align*}
    N_{\mathrm{eff}} := \frac{N}{2\tau_{\mathrm{int}}},
\end{align*}
which is the number of independent samples that would yield the same variance of the mean estimator.
In practice, we estimate $\sigma^2$ and $\tau_{\mathrm{int}}$ from the observed sequence and report the autocorrelation-corrected standard error
\begin{align*}
    \mathrm{SE}(\hat{\mu}) \approx \frac{\hat{\sigma}}{\sqrt{N_{\mathrm{eff}}}}.
\end{align*}

\begin{comment}
A standard approximation for the mean estimator is then
\begin{align*}
    \mathrm{Var}(\hat{\mu}) \;\approx\; \frac{\sigma^2}{N}\,(2\tau_{\mathrm{int}}),
    \qquad
    N_{\mathrm{eff}} := \frac{N}{2\tau_{\mathrm{int}}},
    \qquad
    \mathrm{SE}(\hat{\mu}) \;\approx\; \frac{\hat{\sigma}}{\sqrt{N_{\mathrm{eff}}}}.
\end{align*}
\end{comment}
We therefore report (or internally compute) an autocorrelation-corrected standard error using $N_{\mathrm{eff}}$ rather than $N$.
\subsection{Estimating $\tau_{\mathrm{int}}$ from finite data (IPS truncation)}
\label{app:metrics_ips}

In finite traces, naive estimation of $\tau_{\mathrm{int}}$ by summing empirical autocorrelations over many lags is unstable because the tail of the sample autocorrelation function is dominated by estimation noise and often fluctuates around zero. To obtain a robust estimate, we use the \emph{initial positive sequence (IPS)} truncation heuristic:

\begin{enumerate}
    \item Compute the sample autocorrelation $\widehat{\rho}(\ell)$ for lags $\ell=1,2,\ldots,\ell_{\max}$, where $\ell_{\max}$ is chosen as a modest fraction of the series length (e.g., $\ell_{\max}=\min\{N-1,\lfloor N/10\rfloor\}$).
    \item Define the IPS estimate
    \begin{align*}
        \widehat{\tau}_{\mathrm{int}}
        := \frac{1}{2} + \sum_{\ell=1}^{\ell^\star} \widehat{\rho}(\ell),
        \qquad
        \ell^\star := \min\{\ell \ge 1 : \widehat{\rho}(\ell) \le 0\}-1,
    \end{align*}
    i.e., we sum positive empirical autocorrelations and stop at the first non-positive lag.
\end{enumerate}

This truncation avoids adding a noisy tail whose expected contribution is approximately zero but whose variance can be large. The resulting $N_{\mathrm{eff}}=N/(2\widehat{\tau}_{\mathrm{int}})$ is a conservative summary of how many \emph{effectively independent} observations the correlated sequence contains.

\subsection{Which quantities use autocorrelation correction}
\label{app:metrics_where_used}

The forward wait time $\hat{W}_{\mathcal{C}}(t_k)$ in Definition~\ref{def:fwd_wait_discrete} is defined at \emph{every epoch} $t_k$ and measures the residual time (in seconds) until the criterion next holds starting from that epoch. Its epoch-level average,
\[
    \hat{\bar{W}}
    :=
    \frac{1}{|\mathcal{T}|}\sum_{t_k\in\mathcal{T}} \hat{W}_{\mathcal{C}}(t_k),
    \qquad
    \mathcal{T}:=\{t_k:\hat{W}_{\mathcal{C}}(t_k)<\infty\},
\]
weights long down episodes more heavily because many consecutive epochs $t_k$ lie inside the same down run, each contributing a (typically large) residual wait. By contrast, the down-run duration statistics $\{L_n\}_{n=1}^{M}$ treat each \emph{down interval} as a single event and summarize the \emph{episode lengths} directly (e.g., via $p10/p50/p90$ and the mean), independent of how many epochs fall within each episode. Thus, $\hat{\bar{W}}$ characterizes the expected wait from a uniformly sampled epoch, whereas the distribution of $\{L_n\}$ characterizes the durations of distinct below-threshold episodes once they occur. In a stationary alternating-renewal special case, these viewpoints are linked by the inspection-paradox identity $\mathbb{E}[W_{\downarrow}]=\mathbb{E}[L^2]/(2\mathbb{E}[L])$, illustrating how variability in down-run lengths inflates epoch-level waiting relative to the typical episode duration

\paragraph{Epoch-level time-series means.}
For quantities computed as means over epochs (e.g., $\hat{\mu}_E$ , or mean LCC fraction), we compute $\widehat{\tau}_{\mathrm{int}}$ from the epoch-level sequence and report an autocorrelation-corrected standard error via $N_{\mathrm{eff}}$.

\paragraph{Event-level down-run statistics.}
Down-run quantiles ($p_{10},p_{50},p_{90}$) are computed empirically from the list of run durations and do not require independence assumptions for their definition. For mean down-run durations, we optionally compute an autocorrelation-corrected standard error using the same $N_{\mathrm{eff}}$ methodology applied to the sequence of down-run lengths $\{L_n\}$ (noting that the sequence of episodes can remain correlated under cyclostationarity).

\subsection{Cyclostationarity caveat}
\label{app:metrics_cyclo_caveat}

The orbital geometry induces a cyclostationary structure: statistical properties can depend on orbital phase. The autocorrelation-based correction above is therefore best interpreted as a practical finite-trace adjustment rather than an exact stationary-theory guarantee. Nevertheless, it correctly captures the central issue for uncertainty quantification: temporally correlated samples contain less independent information than their raw count suggests, and $N_{\mathrm{eff}}$ provides a principled way to down-weight this over-count.

\pagenumbering{arabic}
\renewcommand{\thepage} {\arabic{page}}
\bibliographystyle{IEEEtran}
\bibliography{references}
\end{document}